\documentclass{easychair}

\usepackage{cutelmet}
\usepackage{amssymb}
\usepackage{mathtools}
\usepackage{proof}
\usepackage{bussproofs}
\usepackage{latexsym}
\usepackage{stella}
\newtheorem{theorem}{Theorem}[section]
\newtheorem{corollary}{Corollary}[section]
\newtheorem{proposition}{Proposition}[section]

\newtheorem{example}{Example}[section]
                   
\def\@ptdef[#1]{\begin{definition}[#1]\rm}
\newtheorem{definition}{Definition}[section]
                     {{\hfill \Endmark}\end{definition}\@endparenv}

\title{Schemata, Cyclic Proofs and Herbrand Systems}
\titlerunning{Schemata and Cyclic Proofs}
\authorrunning{A. Leitsch et. al.}

\author{
Alexander Leitsch\inst{1}
\and
Anela Loli{\'c}\inst{2}
\and
Stella Mahler\inst{3}  
}

\institute{
  Institute of Logic and Computation,\\TU Wien, 
   Vienna, Austria\\
  \email{leitsch@logic.at}
\and
   Institute of Logic and Computation,\\TU Wien, 
   Vienna, Austria\\
   \email{anela@logic.at}
\and
   Institute of Logic and Computation,\\TU Wien, 
   Vienna, Austria\\
   \email{stella@logic.at}
 }
 
\begin{document}

\maketitle

\begin{abstract} Inductive proofs can be represented by proof schemata, a formalism that represents infinite sequences of proofs by recursive definitions. Since proof schemata avoid the explicit application of induction rules, they admit novel applications, one of which is the realization of Herbrand's theorem in the presence of induction. In this paper, we develop a new type of proof schema based on point transition systems. For skolemized proof schemata without quantified cuts, so-called Herbrand systems, that is, schemata of Herbrand instances of quantified formulas, can be computed. Herbrand systems also allow the representation of schemata of Herbrand sequents, thereby realizing Herbrand's theorem for proof schemata. We compare proof schemata with cyclic proofs and define a transformation from a large class of cyclic proofs to proof schemata. Finally, we show that proof schemata based on point transition systems are capable of proving the 2-Hydra statement, a well-known example that is provable by the cyclic proof system $\mathsf{CLKID}^\omega$ but not in $\mathsf{LKID}$.
\end{abstract}

\section{Introduction}\label{sec.introduction}

The role of a proof extends beyond the verification of a statement. Formal proofs often contain information about the mathematical objects involved, including algorithms or bounds that are not reflected in the theorem itself. Understanding and extracting such information has become an important objective of proof theory. Building on foundational ideas of Kreisel \cite{GKreisel1951,GKreisel1952,GKreisel1958}, the field of proof mining studies methods for uncovering the computational and mathematical content hidden in formal derivations. From this perspective, proofs are viewed not merely as certificates of truth but as structured objects whose analysis can provide deeper insight into the underlying mathematics.

An important tool for analyzing proofs is Herbrand's theorem, which is one of the most important results of mathematical logic. It expresses the fact that in a formal (cut-free) proof of an end-sequent in prenex form the propositional and quantifier inferences can be separated. In the formalism of sequent calculus this means that a so-called Herbrand sequent can be extracted from a proof; the propositional inferences operate above the sequent, the quantifier inferences below. This form of separation was first described in form of the midsequent theorem in Gentzen's famous paper~\cite{Gen35}. In proof theory Herbrand's theorem is also an important tool in proof mining; a prominent example is the analysis of two proofs of Roth's theorem (see~\cite{Luckhardt.1989}). In automated theorem proving Herbrand's theorem is used as a tool to prove completeness of refinements of resolution. The theorem also yields a compact representation of proofs by abstracting away the propositional inferences. Compact representations of proofs also play a major role in computational proof analysis: formal proofs obtained by cut-elimination from formalized mathematical proofs are typically very long covering up the mathematical content of the proofs. Experiments with the system CERES (cut-elimination by resolution, see~\cite{BL00} and~\cite{BL11}) have revealed that Herbrand forms are mathematically significant in displaying the main mathematical arguments of a proof in a natural way~\cite{hetzl2008herbrand}.

The core of the first-order CERES method (for the most general definition of the method, see~\cite{LCL.2026})  consists of the construction of a resolution refutation of a characteristic formula (which represents the derivations of the cut formulas); this resolution refutation is then combined with a so-called proof projection to a proof of the given theorem containing only quantifier-free cuts (called a CERES normal form). As Herbrand's theorem holds also in presence of any quantifier-free cuts, a Herbrand sequent could be extracted from such a CERES normal form. But Herbrand sequents cannot directly be extracted from proofs with (quantified) cuts and cut-elimination is required as a preprocessing. However, in presence of the induction principle Herbrand's theorem actually fails. As most interesting mathematical proofs contain applications of mathematical induction, the extraction of ``Herbrand information" from such proofs seems to be impossible. Instead, it was shown that an extension of the CERES method to {\em proof schemata} could capture cut-elimination in inductive proofs by replacing Herbrand instances by Herbrand systems (representing schemata of substitutions which, in turn, can be used to represent schemata of Herbrand sequents). A first thorough analysis of a schematic CERES method can be found in~\cite{LPW17}; the inductive proofs investigated in this paper are those representable by a single parameter - in the formalisation by a proof schema.
In~\cite{Thesis.Lolic.2020} the approach in~\cite{LPW17} was extended to arbitrarily many induction parameters thus considerably increasing the strength of the method. The most advanced extension of proof schemata with arbitrarily many induction parameters can be found in~\cite{LCL.2026}. In the schematic CERES method two types of proof schemata are considered: 1. the input schema representing an inductive proof and 2. a refutation schema based on schematic resolution applied to the characteristic formula schema, which is quantifier-free. While the input schemata are based on primitive recursion of proofs, more complex forms of recursions were introduced for the characteristic formula schemata in~\cite{CLL.2021} and in~\cite{LCL.2026}.  In~\cite{CLL.2021} the refutational calculus was substantially extended by the use of {\em point transition systems} (see also ~\cite{LCL.2026}). The expressivity of the extended calculus was demonstrated on a formula schema which could not be refuted by earlier methods. However, recursion based on point transition systems has not yet been applied to input proof schemata.

While proof schemata provide one approach to the analysis of inductive proofs, there exist other formalisms for inductive inference admitting cut-elimination.
Among the most prominent are the cyclic proof systems of Brotherston and Simpson~\cite{Brotherston.2005, brotherston2011sequent} and the inductive calculi of McDowell and Miller~\cite{mcdowell2000cut}.
For a long time, the extraction of Herbrand-type structures from such calculi remained open, in contrast to the development of Herbrand systems for proof schemata~\cite{LPW17, Thesis.Lolic.2020, LL.2025, LCL.2026}.
Recently, however, Afshari et al.~\cite{afshari2025herbrand} showed that Herbrand structures can indeed be extracted from cyclic proofs via higher-order recursion schemes.
Our work builds on the cyclic framework of Brotherston and Simpson and investigates its connection to the proof schema formalism and schematic CERES.

Cyclic proofs provide an alternative formalism for inductive reasoning by replacing explicit induction rules with global proof structures containing cycles. Instead of applying an induction rule locally, cyclic proofs allow derivations in which certain sequents may reoccur, forming cycles that are justified by a global soundness condition. 
This condition, typically formulated as a trace condition, guarantees that every infinite path through the derivation exhibits a notion of syntactic progress.
From the perspective of proof schemata, these cyclic dependencies naturally correspond to recursive calls between schematic proof components.

In this paper, we consider the cyclic proof system for inductive definitions $\mathsf{CLKID}^{\omega}$ of \cite{brotherston2006.thesis, Brotherston.2005, brotherston2011sequent}.
While there are approaches to extract Herbrand structures from cyclic proofs through higher-order recursion schemes \cite{afshari2025herbrand}, a translation of $\mathsf{CLKID}^{\omega}$ proofs into the proof schema formalism allows the direct extraction of Herbrand systems from inductive proofs without quantified cuts and the application of schematic CERES for inductive proofs with quantified cuts.
Due to the nature of proof schemata we consider a subclass of $\mathsf{CLKID}^{\omega}$ derivations restricted to proofs in which the variables governing the inductive behaviour of the cyclic proof are not bound by strong quantification; here, strong quantifiers refer to universal quantifiers of positive polarity and existential quantifiers of negative polarity. In addition, we limit our scope to systems where the inductive definitions are indeed recursive definitions. 
Despite this restriction, it is possible to formalize the Two-Hydra statement of Berardi and Tatsuta~\cite{BT.2017} in our formalism.
In case of recursive definitions that do not form a complete characterization we perform an extension of the calculus, which is justified by assuming the standard semantics of $\mathsf{CLKID}^{\omega}$.

In Section~\ref{sec.schema} we first define the calculus {\LK} (in the version we are using it). In Section~\ref {subsec.pts}  the main features of point transition systems are described (for a more detailed description we refer to~\cite{LCL.2026}). In Section~\ref{subsec.LKschema} we introduce a schematic logic and corresponding proof schemata based on the sequent calculus 
$\LKN$. The proof schemata in this paper differ from the schemata in~\cite{CLL.2021} and~\cite{LCL.2026} in various ways. They are simpler concerning the definition of terms and formulas as no recursive schemata are considered on this level. On the other hand $\LKN$ is more powerful than the schematic {\LK}-calculi defined so far as it allows weak quantification over parameters. Moreover, our schemata admit a more powerful proof recursion based on point transition systems (note that in~\cite{CLL.2021} point transition systems have been used only for schematic resolution calculi, not for the {\LK}-based calculi).  In Section~\ref{sec.herbrandsystems} we show that for all proof schemata (based on point transition systems) with at most quantifier-free cuts Herbrand systems and schematic Herbrand sequents can be computed. This result extends a former result about all proof schemata based on primitive recursion and containing at most quantifier-free cuts. Section~\ref{sec.translation} shows that a subclass of cyclic proofs can be algorithmically translated into proof schemata, thus making cyclic proofs candidates for the extraction of Herbrand systems. Finally, in Section~\ref{sec.2hydraexample}, we show that the 2-Hydra example can be translated into a proof schema and we construct the corresponding Herbrand system. Combining this construction with the result of Berardi and Tatsuta~\cite{BT.2017}, we obtain that proof schemata based on point transition systems can prove inductive statements that are not provable in $\mathsf{LKID}$.

\section{Proof Schemata}\label{sec.schema}

A proof schema is a formalism admitting the use of mathematical induction without applying an induction rule. Induction is replaced by an expression representing an infinite sequence of proofs. Proof schemata have been used successfully for cut-elimination in inductive proofs and in proof analysis of mathematical proofs (see e.g.~\cite{BHLRS.2008},\cite{LPW17} and~\cite{Thesis.Lolic.2020}). Unlike in the publications mentioned above we consider only proof schemata without recursive term- and formula expressions, but with more complex proof recursions. The basic calculus for all our investigations in this paper is the Gentzen calculus {\LK} to be defined below.

\begin{definition}[sequent calculus $\LK$] \label{def.LK}
Here, we will use a variant of Gentzen's version of \textbf{LK} \cite{Gen35}. Since we consider multi-sets of formulas, we do not need exchange or permutation rules. 
There are two groups of rules, the logical and the structural ones. All rules except the cut have left and right versions, denoted by $l$ and $r$, respectively. The binary rules are of multiplicative type, i.e. no auto-contraction of the context is applied. 
Below, $A$ and $B$ denote formulas whereas $\Gamma, \Delta, \Pi, \Lambda$ denote multi-sets of formulas. In the rules there are introducing or auxiliary formulas in the premises and introduced or principal formulas in the conclusion. We indicate these formulas for all rules, auxiliary formula occurrences are marked by $+$ and principal formula occurrences are marked by $*$. We frequently say auxiliary (main) formula instead of auxiliary (main) formula occurrence.
\\[1ex]
{\bf The logical rules:}

$\land$-introduction

	\begin{minipage}{0.4\linewidth}
	\centering
	\begin{prooftree}
	\AxiomC{$A^+, B^+, \Gamma \vdash \Delta$}
	\RightLabel{$\land_{l}$}
	\UnaryInfC{$(A \land B)^*, \Gamma \vdash \Delta$}
	\end{prooftree}
	\end{minipage}
	\begin{minipage}{0.4\linewidth}
	\centering
	\begin{prooftree}
	\AxiomC{$\Gamma_1 \vdash \Delta_1, A^+$}
	\AxiomC{$\Gamma_2 \vdash \Delta_2, B^+$}
	\RightLabel{$\land_{r}$}
	\BinaryInfC{$\Gamma_1 , \Gamma_2 \vdash \Delta_1 , \Delta_2 , (A \land B)^*$}
	\end{prooftree}
	\end{minipage}
	
$\lor$-introduction

	\begin{minipage}{0.4\linewidth}
	\centering
	\begin{prooftree}
	\AxiomC{$\Gamma \vdash \Delta, A^+, B^+$}
	\RightLabel{$\lor_{r}$}
	\UnaryInfC{$\Gamma \vdash \Delta, (A \lor B)^+$}
	\end{prooftree}
	\end{minipage}	
	\begin{minipage}{0.4\linewidth}
	\centering
	\begin{prooftree}
	\AxiomC{$A^+, \Gamma_1 \vdash \Delta_1$}
	\AxiomC{$B^+, \Gamma_2 \vdash \Delta_2$}
	\RightLabel{$\lor_{l}$}
	\BinaryInfC{$(A \lor B)^*, \Gamma_1, \Gamma_2 \vdash \Delta_1 , \Delta_2$}
	\end{prooftree}
	\end{minipage}
	
$\to$-introduction

	\begin{minipage}{0.4\linewidth}
	\centering
	\begin{prooftree}
	\AxiomC{$\Gamma_1 \vdash \Delta_1, A^+$}
	\AxiomC{$B^+, \Gamma_2 \vdash \Delta_2$}
	\RightLabel{$\to_{l}$}
	\BinaryInfC{$(A \to B)^*, \Gamma_1 , \Gamma_2 \vdash \Delta_1 , \Delta_2$}
	\end{prooftree}
	\end{minipage}
	\begin{minipage}{0.4\linewidth}
	\centering
	\begin{prooftree}
	\AxiomC{$A^+, \Gamma \vdash \Delta, B^+$}
	\RightLabel{$\to_{r}$}
	\UnaryInfC{$\Gamma \vdash \Delta, (A \to B)^*$}
	\end{prooftree}
	\end{minipage}
	
$\neg$-introduction

	\begin{minipage}{0.4\linewidth}
	\centering
	\begin{prooftree}
	\AxiomC{$\Gamma \vdash \Delta, A^+$}
	\RightLabel{$\neg_{l}$}
	\UnaryInfC{$\neg A^*, \Gamma \vdash \Delta$}
	\end{prooftree}
	\end{minipage}
	\begin{minipage}{0.4\linewidth}
	\centering
	\begin{prooftree}
	\AxiomC{$A^+, \Gamma \vdash \Delta$}
	\RightLabel{$\neg_{r}$}
	\UnaryInfC{$\Gamma \vdash \Delta, \neg A^*$}
	\end{prooftree}
	\end{minipage}
	
$\forall$-introduction

	\begin{minipage}{0.4\linewidth}
	\centering
	\begin{prooftree}
	\AxiomC{$A\{x \leftarrow t\}^+, \Gamma \vdash \Delta$}
	\RightLabel{$\forall_{l}$}
	\UnaryInfC{$\forall x A^*, \Gamma \vdash \Delta$}
	\end{prooftree}
	\end{minipage}
	\begin{minipage}{0.4\linewidth}
	\centering
	\begin{prooftree}
	\AxiomC{$\Gamma \vdash \Delta, A\{x \leftarrow \alpha\}^+$}
	\RightLabel{$\forall_{r}$}
	\UnaryInfC{$\Gamma \vdash \Delta, \forall x A^*$}
	\end{prooftree}
	\end{minipage}
	
where $t$ is an arbitrary term and $\alpha$ is a free variable which may not occur in $\Gamma, \Delta, A$. $\alpha$ is called an eigenvariable. 
	
$\exists$-introduction

	\begin{minipage}{0.4\linewidth}
	\centering
	\begin{prooftree}
	\AxiomC{$A\{x \leftarrow \alpha\}^+, \Gamma \vdash \Delta$}
	\RightLabel{$\exists_{l}$}
	\UnaryInfC{$\exists x A^*, \Gamma \vdash \Delta$}
	\end{prooftree}
	\end{minipage}
	\begin{minipage}{0.4\linewidth}
	\centering
	\begin{prooftree}
	\AxiomC{$\Gamma \vdash \Delta, A\{x \leftarrow t\}^+$}
	\RightLabel{$\exists_{r}$}
	\UnaryInfC{$\Gamma \vdash \Delta, \exists x A^*$}
	\end{prooftree}
	\end{minipage}
	
where the variable conditions for $\exists_{l}$ are the same as those for $\forall_r$ and similarly for $\exists_r$ and $\forall_l$. The quantifier-rules $\forall_l,\exists_r$ are called {\em weak}, the rules $\exists_l,\forall_r$ {\em strong}.\\[1ex]
{\bf The structural rules:}

weakening

	\begin{minipage}{0.4\linewidth}
	\centering
	\begin{prooftree}
	\AxiomC{$\Gamma \vdash \Delta$}
	\RightLabel{$w_{l}$}
	\UnaryInfC{$A^*, \Gamma \vdash \Delta$}
	\end{prooftree}
	\end{minipage}
	\begin{minipage}{0.4\linewidth}
	\centering
	\begin{prooftree}
	\AxiomC{$\Gamma \vdash \Delta$}
	\RightLabel{$w_{r}$}
	\UnaryInfC{$\Gamma \vdash \Delta, A^*$}
	\end{prooftree}
	\end{minipage}
	
contraction

	\begin{minipage}{0.4\linewidth}
	\centering
	\begin{prooftree}
	\AxiomC{$A^+, A^+, \Gamma \vdash \Delta$}
	\RightLabel{$c_{l}$}
	\UnaryInfC{$A^*, \Gamma \vdash \Delta$}
	\end{prooftree}
	\end{minipage}
	\begin{minipage}{0.4\linewidth}
	\centering
	\begin{prooftree}
	\AxiomC{$\Gamma \vdash \Delta, A^+, A^+$}
	\RightLabel{$c_{r}$}
	\UnaryInfC{$\Gamma \vdash \Delta, A^*$}
	\end{prooftree}
	\end{minipage}
	
cut

	\begin{prooftree}
	\AxiomC{$\Gamma_1 \vdash \Delta_1, A^+$}
	\AxiomC{$A^+, \Gamma_2 \vdash \Delta_2$}
	\RightLabel{${\it cut}$}
	\BinaryInfC{$\Gamma_1 , \Gamma_2 \vdash \Delta_1 , \Delta_2$}
	\end{prooftree}
The formula $A$ is the auxiliary formula of $\it{cut}$ (also called the cut-formula) and there are no principal ones.

An {\em axiom set} is a set of atomic sequents which are closed under substitution. An $\LK$-proof from an axiom set $\Acal$ is a tree formed according to the rules of $\LK$ such that all leaves are in $\Acal$. The axiom set $\Acal_0\colon \{A \seq A \mid A \mbox{ an atomic formula}\}$ is called the {\em  standard axiom set}.
\end{definition}

\subsection{Point Transition Systems}\label{subsec.pts}

Point transition systems (see~\cite{CLL.2021} and~\cite{LCL.2026}) are frameworks to describe the abstract form of recursions (in the form of recursive calls), independent of the objects to be defined. This formalism can be used to define recursive functions, schematic terms, schematic formulas and schematic proofs. Though we use point transitions for defining proof schemata we give first a motivating example by defining computable functions. A simple form of recursion is primitive recursion, where the base case $0$ and a step case $>0$ are distinguished. 
\begin{example}\label{ex.primrec}
Let us assume that the functions $p$ (predecessor) and $s$ (successor) are already defined. It is well-known that we can define the addition $a$ in the form 
$$a(x,y) \eqdef\ \AIf\ y=0\ \AThen\ x\ \AElse\ s(a(x,p(y))).$$ 
Here $a(x,y)$ can either call $x$ (which is already defined) or $a(x,p(y))$ (assumed to be already defined). This is a specific example of primitive recursion.
\end{example}
 However, in practice of recursive definitions, cyclic calls among different defined functions cannot be avoided as illustrated below. 

\begin{example}\label{ex.pts0}
To give just a flavor of the problem, we define functions $f,g$ and assume that $p,+,*$ are already defined:
\begin{eqnarray*}
f(x,y) &=& f(p(x),y) + g(x,p(y)) \mbox{ for } x>0 \land y>0,\\
f(x,y) &=& x \mbox{ for }x>0 \land y=0,\\
f(x,y) &=& y \mbox{ for }x=0 \land y>0,\\
f(x,y) &=& 0 \mbox{ for }x=0 \land y=0.\\[1ex]
g(x,y) &=& f(x,p(y)) + f(x,p(y))*x \mbox{ for } y>0,\\
g(x,y) &=& x*x \mbox{ for } y=0.
\end{eqnarray*}
When we are primarily interested in showing that $f$ and $g$ are well-defined - and not in their values - we may abstract from the function symbols $p,+,*$ from which we know that they are well-defined and specify total functions. The structure we are using here will be defined formally as {\em point transition system}. The points will be the arguments of the functions, the labels the names of the functions. We use the label  "end" for the cases where at most the functions $s,p,+,*$ are used  in the definitions. We give the point transition system for the definitions above (the formal definitions then follow below).
\[
\begin{array}{l}
(f,(x,y)) \to \{(f,(p(x),y)), (g,(x,p(y)))\} \colon x>0 \land y>0,\\
(f,(x,y)) \to \{(\mbox{end},(x,y))\} \colon x>0 \land y=0,\\
(f,(x,y)) \to \{(\mbox{end},(x,y))\} \colon x=0 \land y>0,\\
(f,(x,y)) \to \{(\mbox{end},(x,y))\} \colon x=0 \land y=0.\\[1ex]
(g,(x,y)) \to \{(f,(x,p(y)))\}\colon y>0,\\
(g,(x,y)) \to \{(\mbox{end},(x,y))\} \colon y=0.
\end{array}
\]
By using an adequate tuple ordering for the arguments of the functions $f$ and $g$ we can show that $f$ and $g$ are indeed well-defined. 
Moreover, it is easy to see that the conditions above define a {\em partition}, i.e. the cases cover any tuple $(x,y)$.
\end{example}
We are now developing the formal framework of point transition systems. Let $T_0$ be the set of terms defined over {\em parameters} (which are variables over the natural numbers), the functions $s$ (successor), $p$ (predecessor), and $0$ (zero), where $p$ is defined as $p(0) = 0$ and $p(s(x)) =x$; the subset of parameter-free terms in $T_0$ is denoted by $T^G_0$. It is easy to show that terms in $T^G_0$ are equivalent to {\em numerals}, i.e. expressions of the form $s^\alpha(0)$ for some $\alpha \in \N$. The semantics of point transition systems is based on {\em parameter assignments}, which are functions from parameters to numerals. For further examples and the formal definition of semantics of point transition systems we refer to~\cite{LCL.2026}.
\begin{definition}[point\index{point}]\label{def.point}
A {\em point} is an element of $T_0^\alpha$ for $\alpha \geq 1$ and $T_0$ as defined above.
\end{definition}
\begin{definition}[labeled point\index{labeled point}]\label{def.labeledpoint}
Let $\Delta$ be a countably infinite set of symbols called {\em labels}. A {\em labeled point} is a pair $(\delta,p)$ where $p$ is a point and $\delta$ is a label.
\end{definition}
\begin{definition}[condition\index{condition}]\label{def.condition}
An {\em atomic condition} is either $\top,\bot$ or of the form (where $\triangle \in \{=,<,>,\leq,\geq\}$) 
\begin{itemize}
\item $s^\alpha(v) \triangle s^\beta(\bar{0})$ for a parameter $v$,  $\alpha \geq 0,\beta \geq 0$, or 
\item $s^\alpha(v) \triangle s^\beta(w)$ for (different) parameters $v$ and $w$,  $\alpha \geq 0,\beta \geq 0$.
\end{itemize} 
 Let $\ATC$ be the set of atomic conditions; the set $\COND$ of all conditions is defined as follows:
\begin{itemize}
\item $\ATC \IN \COND$,
\item if $C \in \COND$ then $\neg C \in \COND$,
\item if $C_1,C_2 \in \COND$ then $C_1 \land C_2 \in \COND$ and $C_1 \lor C_2 \in \COND$.
\end{itemize}
\end{definition} 
For better legibility we frequently use $k+\alpha$ instead of $s^\alpha(k)$ though the terms in the conditions do not contain $+$.
\begin{example}
$n=0$, $m<k+1$ are atomic conditions, $n=0 \land \neg m <k+1$ and $(n=0 \land m<k+1) \lor m  > k +2$ are conditions.
\end{example}
The semantics of a condition $C$ is obtained by interpreting the variables by a parameter assignment $\sigma$ and by evaluating $C$ under $\sigma$ (where $<,>$ are interpreted via the standard interpretation). So, if $\sigma(n)=2,\ \sigma(m)=3$ and $C\eqdef n<m$ then $\sigma(C) = \top$, for $D \eqdef n=m$ $\sigma(D)=\bot$.
\begin{definition}[point transition\index{point transition}]\label{def.pt}
A {\em point transition} is an expression of the form 
$$(\delta,p) \to \Lcal\colon C$$
where $(\delta,p)$ is a labeled point, $\Lcal$ is a nonempty set of labeled points, and $C$ is a condition. To extract the different parts of a point transition we define the functions $l,r,c$ as follows: if $t = (\delta,p) \to \Lcal\colon C$ then $l(t) = (\delta,p)$, $r(t) = \Lcal$ and $c(t) = C$.
\end{definition}
To every label $\delta$ we assign points of the same arity; the result is a so-called point transition cluster.
\begin{definition}[point transition cluster\index{point transition cluster}]\label{def.ptc}
A finite set of point transitions $\Pbf$ is called a {\em point transition cluster} if for all $\delta \in \Delta$ and points $p,q$ such that 
$(\delta,p)$ and $(\delta,q)$ occur in $\Pbf$ (either as $l(t)$ or as $r(t)$ for a $t \in \Pbf$) there exists an $\alpha \geq 1$ such that $p,q \in \Pcal^\alpha$. We write $\Delta(\Pbf)$ for the set of all $\delta \in \Delta$ occurring in $\Pbf$; with $A(\delta)$ we denote the arity of the points corresponding to $\delta$.
\end{definition}
\begin{example}\label{ex.ptc}
Let 
\begin{eqnarray*}
\Pbf_1 &=& \{(\delta,(m+1,n)) \to \{(\delta,(m-1,n+1)), (\delta',(m,n,m))\}\colon m>0 \land n<m,\\
         & & (\delta,(m,n)) \to \{(\delta',(n,n,n))\}\colon m<n
\end{eqnarray*}
$\Pbf_1$ is a point transition cluster: we have $A(\delta)=2, A(\delta')=3$. The set $\Pbf_2$ defined as 
\begin{eqnarray*}
\Pbf_2 &=& \{(\delta,(m+1,n)) \to \{(\delta,(m-1,n+1)), (\delta',(m,n,m))\}\colon m>0 \land n<m,\\
         & & (\delta,(m,n)) \to \{(\delta',(n,n))\}\colon m<n
\end{eqnarray*}
is not a point transition cluster as $A(\delta')$ cannot be defined consistently. 
\end{example}
Consider the point transition cluster $\Pbf_1$ in Example~\ref{ex.ptc}. The conditions $m>0 \land n<m$ and $m<n$ exclude each other, but they are not exhaustive: we do not know what happens when $m=n$. In point transition systems (restricted forms of point transition clusters) the conditions must define a partition and the point transitions must be of a restricted form. 
\begin{definition}[partition\index{partition}]\label{def.partition}
Let $\Ccal = \{C_1,\ldots,C_\alpha\}$ be a set of conditions. $\Ccal$ is called a {\em partition} if $C_1 \lor \cdots \lor C_\alpha$ is valid and for all $i,j \in \{1,\ldots,\alpha\}$ such that $i \neq j$ we have that $C_i \land C_j$ is unsatisfiable.
\end{definition} 

\begin{example}\label{ex.if-then-else-part}
Take 
$$\Ccal = \{x>0 \land y>0,\ x>0 \land y=0,\ x=0 \land y>0,\ x=0 \land y=0 \}$$
from Example~\ref{ex.pts0}. Then  $\Ccal$ is a partition. 
\end{example}
\begin{definition}[regular point transition\index{regular point transition}]\label{def.regularpt}
Let $t\colon (\delta,p) \to \Lcal\colon C$  be a point transition in a point transition cluster and $A(\delta) = \alpha$. $t$ is called {\em regular} if 
\begin{itemize}
\item $p = (n_1,\ldots,n_\alpha)$ for distinct parameters $n_1,\ldots,n_\alpha$,
\item for all $(\delta',q) \in \Lcal$ $V(q) \IN \{n_1,\ldots,n_\alpha\}$,
\item $V(C) \IN  \{n_1,\ldots,n_\alpha\}$, where $V(C)$ denotes the set of parameters occurring in $C$.
\end{itemize}
\end{definition}
Note that the first point transition in the point transition cluster $\Pbf_1$ in Example~\ref{ex.ptc} is not regular as $p=(m+1,n)$, but it should be of the form $(m,n)$.
\begin{definition}[point transition system\index{point transition system}]\label{def.pts}
Let $\Pbf$ be a point transition cluster and $\delta_0 \in \Delta(\Pbf)$. Then the tuple 
$\Pbf^*\colon (\delta_0,\Delta^*,\Delta_e,\Pbf)$ is called a {\em point transition system} if the following conditions are fulfilled:
\begin{itemize}
\item $\delta_0 \in \Delta^*$.
\item $\Delta_e \IN \Delta^*$.
\item $\Delta^* = \Delta(\Pbf)$.
\item All point transitions in $\Pbf$ are regular.
\item Let $t_1,t_2 \in \Pbf$ such that $l(t_1) = (\delta,p)$ and $l(t_2) = (\delta,q)$; then $p=q$ - i.e. every left-hand side of a point transition corresponding to $\delta$ has the same point. We call $p$ the {\em source} of $\delta$.
\item Let $\Pbf(\delta) = \{(\delta,p) \to \Lcal_1\colon C_1,\ldots,(\delta,p) \to \Lcal_\alpha\colon C_\alpha\}$ be the set of all point transitions $t$ in $\Pbf$ with $l(t) = (\delta,p)$. Then, for $\Pbf(\delta) \neq \emptyset$, $\{C_1,\ldots,C_\alpha\}$ is a partition.
\item $\Pbf(\delta_0) \neq \emptyset$.
\item $\Pbf(\delta) = \emptyset$ for $\delta \in \Delta_e$.
\end{itemize}
The label $\delta_0$ is called the {\em start label} of $\Pbf^*$, any  $\delta \in \Delta_e$ is called an {\em end label} of $\Pbf$.
\end{definition}
\begin{example}\label{ex.pts}
Consider Example~\ref{ex.pts0}. We reformulate the point transition system given there such that it is closer to the formal definition above. For $f$ we choose $\delta_0$, $\delta_1$ for $g$ and $\delta_e$ for "end". Then $\Delta^* =\{\delta_0,\delta_1,\delta_e\}$, $\Delta_e = \{\delta_e\}$ and 
\begin{eqnarray*}
\Pbf(\delta_0) &=& \{(\delta_0,(x,y)) \to \{(\delta_0,(p(x),y)), (\delta_1,(x,p(y))\} \colon x>0 \land y>0,\\
                          & & (\delta_0,(x,y)) \to \{(\delta_e,(x,y)))\} \colon x>0 \land y=0,\\
                         & & (\delta_0,(x,y)) \to \{(\delta_e,(x,y))\} \colon x=0 \land y>0,\\
                         & & (\delta_0,(x,y)) \to \{(\delta_e,(x,y))\} \colon x=0 \land y=0 \}.\\[1ex]
\Pbf(\delta_1) &=& \{(\delta_1,(x,y)) \to \{(\delta_0,(x,p(y)))\}\colon y>0,\\
                           & & (\delta_1,(x,y)) \to \{(\delta_e,(x,y))\} \colon y=0\}.\\[1ex]
\Pbf(\delta_e) &=& \emptyset.
\end{eqnarray*}
Indeed, all point transitions are regular and the conditions belonging to $\Pbf(\delta_0)$ and $\Pbf(\delta_1)$ define a partition.
\end{example}
Every point transition system defines a computation with respect to parameter assignments. Such a computation can be described as a (possibly infinite) tree. If $\sigma$ is a parameter assignment and $\Pbf^*$ is a point transition system we write  $\Pbf^*[\sigma]$ for the computation tree defined by $\Pbf^*$ under the parameter assignment $\sigma$. For a formal definition of $\Pbf^*[\sigma]$ see~\cite{LCL.2026}.  $\Pbf^*[\sigma]$ may be finite or infinite; in the latter case the computation is nonterminating. Indeed, 
$\Pbf^*[\sigma]$ may be finite, but infinite for another $\sigma'$. The desirable property is the termination for all $\sigma \in \Scal$, where $\Scal$ denotes the set of all parameter assignments. 
\begin{definition}[terminating point transition system]\label{def.ptsterminate}
Let $\Pbf^*$ be a point transition system. $\Pbf^*$ is called {\em terminating} if, for all $\sigma \in \Scal$, $\Pbf^*[\sigma]$ is finite.
\end{definition}
In order to ensure the termination of a point transition system some syntactic criteria are necessary. To this aim we will define first an ordering on the labels and then a well-founded  ordering on points. These orderings will then be combined to a well-founded ordering on labeled points. 
\begin{definition}\label{def.label-relations}
Let $\Pbf^*\colon (\delta_0,\Delta^*,\Delta_e,\Pbf)$ be a point transition system and $\delta \in \Delta^*$ such that 
$$\Pbf(\delta) = \{(\delta,p) \to \Lcal_1\colon C_1,\ldots,(\delta,p) \to \Lcal_\alpha\colon C_\alpha\}.$$
We define $\Ccal(\delta) = \{\delta' \mid \delta' \mbox{ occurs in } \Lcal_1 \union \cdots \union \Lcal_\alpha\}$.\\[1ex]
$\Ccal(\delta)$ can be considered as the set of all labels which can be called from $\delta$. We define
\begin{itemize}
\item $\delta \redPbf \delta'$ if $\delta' \in \Ccal(\delta)$ ($\delta$ calls $\delta'$) and 
\item $\redPbfs$ as the transitive closure of $\redPbf$.
\item Assume $\delta \redPbfs \delta'$ and $\delta' \redPbfs \delta$ then we define $\delta \eqPbf \delta'$.
\end{itemize}
\end{definition}
$\delta \eqPbf \delta$ means that $\delta$ calls itself either directly (i.e. $\delta \redPbf \delta$) or over different $\delta_1,\ldots,\delta_m$ such that $\delta \redPbf \delta_1 \redPbf \cdots \redPbf \delta_m \redPbf \delta$ and we have cyclic calls. It is easy to see that for all these $\delta_m$ we have $\delta_m \eqPbf \delta$. Indeed, we have
\begin{proposition}\label{prop.eqPbf}
Let $\Pbf^*\colon (\delta_0,\Delta^*,\Delta_e,\Pbf)$ be a point transition system. Then 
\begin{itemize}
\item[1.] if $\delta \eqPbf \delta'$ then $\delta \eqPbf \delta$ and $\delta' \eqPbf \delta'$,
\item[2.] $\eqPbf$ is symmetric,
\item[3.] $\eqPbf$ is transitive.
\end{itemize}
\end{proposition}
\begin{proof}
\begin{itemize}
\item[1.] By $\delta \redPbfs \delta'$ and $\delta' \redPbfs \delta$ and by transitivity of $\redPbfs$.
\item[2.] Trivial by definition.
\item[3.] $\eqPbf$ is transitive: assume we have $\delta \eqPbf \delta'$ and $\delta' \eqPbf \delta''$; then we also have 
$\delta \redPbfs \delta'$ and $\delta' \redPbfs \delta''$ and $\redPbfs$ is transitive.
\end{itemize}
\end{proof}
Note that $\eqPbf$ is not an equivalence relation on $\Delta^*$ as, in case $\Delta_e \neq \emptyset$ - which is the only nontrivial case, 
$\delta \not \eqPbf \delta$ for $\delta \in \Delta_e$. Indeed, there is no $\delta'$ such that $\delta \redPbf \delta'$. \\[1ex]
We can now define a call-ordering on the labels of a point transition system:
\begin{definition}\label{def.ltPbf}
Let $\Pbf^*\colon (\delta_0,\Delta^*,\Delta_e,\Pbf)$ be a point transition system and $\delta,\delta' \in \Delta^*$. We define 
\begin{itemize}
\item $\delta' \ltPbf \delta$ if $\delta \redPbfs \delta'$ and $\delta' \not \redPbfs \delta$.
\end{itemize}
\end{definition} 
\begin{proposition}\label{prop.ltPbf}
Let $\Pbf^*$ be a point transition system. Then the relation $\ltPbf$ is transitive and irreflexive.
\end{proposition}
\begin{proof}
\begin{itemize}
\item[1.] $\ltPbf$ is transitive: Assume that $\delta \ltPbf \delta'$ and $\delta' \ltPbf \delta''$. Then by transitivity of $\redPbfs$ we obtain 
$\delta  \redPbfs \delta''$. We have to show that $\delta'' \not \redPbfs \delta$. So assume that $\delta'' \redPbfs \delta$, then we have 
$\delta'' \redPbfs \delta$ and $\delta \redPbfs \delta'$ and therefore $\delta'' \redPbfs \delta'$ by transitivity. But this contradicts 
$\delta' \ltPbf \delta''$. 
\item[2.] $\ltPbf$ is irreflexive: obvious as $\delta \redPbfs \delta$ and $\delta \not \redPbfs \delta$ is contradictory.
\end{itemize}
\end{proof}
\begin{example}\label{ex.label-relations}
Consider Example~\ref{ex.pts}. We have $\Delta^* =\{\delta_0,\delta_1,\delta_e\}$, $\Delta_e = \{\delta_e\}$ and the relations 
$$\delta_0 \eqPbf \delta_1,\ \delta_e \ltPbf \delta_0,\ \delta_e \ltPbf \delta_1.$$
\end{example}
We have defined an ordering on the labels of a point transition system, but in order to prove termination we need an ordering on labeled points. We first consider tuples over parameter-free points, i.e. tuples of numerals. On such sets of tuples there exist well-founded orderings which we denote as $<_l$.  An example  $<_l$ is, e.g., the lexicographic ordering (via the standard ordering on numerals). Combining the orderings $<_l$ and $\ltPbf$ we could define 
\begin{eqnarray*}
(\delta',p') < (\delta,p) &\mbox{ if either }& \delta' \ltPbf \delta\\ 
                                         &\mbox{ or } & \delta' \eqPbf \delta \mbox{ and   } p' <_l p.
\end{eqnarray*}
But for defining termination conditions for point transition systems we have to extend the ordering $<_l$ on ground tuples to orderings on tuples over $T_0$. As our point transitions are conditional we have to take into account the conditions.
\begin{definition}\label{def.C-order}
Let $\Pbf^*$ be a point transition system such that $(\delta,p)  \to \Lcal\colon C \in \Pbf$ and $(\delta',p') \in \Lcal$. We define 
$(\delta',p') <_C (\delta,p)$ if  
\begin{itemize}
\item either $\delta' \ltPbf \delta$ or 
\item $\delta' \eqPbf \delta$ and $\sigma(p')\Eval <_l \sigma(p)\Eval$ for all $\sigma \in \Scal$ such that $\sigma[C] = \top$,
\end{itemize}
where $<_l$ is a Noetherian ordering of tuples in $T^G_0$.
\end{definition}
We call a system which admits such an ordering {\em canonic}. 
\begin{definition}[canonic point transition systems]\label{def.pts-canonic} \mbox{ }\\
Consider a point transition system $\Pbf^*\colon (\delta_0,\Delta^*,\Delta_e,\Pbf)$.
$\Pbf^*$ is called {\em canonic} if 
$$\mbox{for all } (\delta,p)  \to \Pcal\colon C \in \Pbf \mbox{ we have } (\delta',p') <_C (\delta,p) \mbox{ for all }(\delta',p') \in \Pcal$$ 
for the ordering $<_C$ defined in Definition~\ref{def.C-order}.
\end{definition}
\begin{theorem}\label{theo.canonic-terminate}
Canonic point transition systems are terminating.
\end{theorem}
\begin{proof}
In~\cite{LCL.2026}.
\end{proof}
As a consequence a recursive definition having a corresponding canonic point transition is well defined for all parameter assignments, i.e. for any inputs. As an example consider the function definitions in Example~\ref{ex.pts0}; the corresponding point transition is canonical (by choosing $<_l$ as the lexicographic ordering) and thus the function $f$ is well-defined for all inputs, i.e. for all parameter assignments  $\sigma$.

\subsection{{\LK}-Schemata}\label{subsec.LKschema}

Proof schemata are syntactic expressions representing infinite sequences of proofs indexed by parameters, that is, free variables over the natural numbers that may be instantiated by numerals. In this paper, we consider a restricted form of proof schemata in which only proofs, but neither terms nor formulas, are defined recursively. This contrasts with the more general frameworks of~\cite{CLL.2021} and~\cite{LCL.2026}, where terms and formulas may also have inductive definitions. Moreover, throughout this paper we work in first-order logic over the natural numbers, rather than over arbitrary domains.
\begin{definition}\label{def.variables}
We consider two classes of free variables over $\omega$, the class of {\em ordinary free variables} $V_f$ and the set of {\em parameters} $\Ncal$; we assume $V_f$ and $\Ncal$ to be disjoint. The set of {\em bound variables} are described by $V_b$.
\end{definition}
On the basis of these variables we will define a schematic $\Ncal$-logic. As syntactic material serve $\bar{0}$ (the constant symbol for $0$), $s$ (the successor function) and $p$ (the predecessor function). Moreover, an infinite set of function symbols $\FS$ such that 
$\FS = \Union^\infty_{i=0}\FS_i$, $\FS_i\colon \omega^i \to \omega$ and $\FS \intrs \{\bar{0},s,p\} = \emptyset$. We also use an infinite set of predicate symbols $\PS$ such that $\PS = \Union^\infty_{i=0}\PS_i$ and $\PS_i\colon \omega^i \to o$.
\begin{definition}\label{def.N-terms}
We define the set $\TN$ of {\em $\Ncal$-semi-terms}:
\begin{itemize}
\item $\Ncal \union V_f \union V_b \IN \TN$,
\item $\bar{0} \in \TN$,
\item if $t \in \TN$ then $p(t),s(t) \in \TN$,
\item if $f \in \FS^i$ and $t_1,\ldots,t_i \in \TN$ then $f(t_1,\ldots,t_i) \in \TN$.
\end{itemize}
$\Ncal$-semi-terms not containing bound variables are called {\em $\Ncal$-terms}, and $\Ncal$-terms not containing parameters are simply called {\em terms} and are denoted by $T$. 
\end{definition}
The set of terms $T_0$ defined in Section~\ref{subsec.pts}, which is a subset of $\TN$, will mainly be used for defining proof recursion. Note that, as we subject $\bar{0},s,p$ to the standard semantics, every term in $T^G_0$ can be identified with a numeral.\\[1ex]
Now we have two types of free variables, $V_f$ and $\Ncal$. For the concept of substitution we take the set $V_f$, while substitutions on the domain $\Ncal$ will be defined as parameter assignments. 
\begin{definition}[$V_f$-substitution]\label{def.Vf-subst}
Let $x_1,\ldots,x_k$ be (different) variables in $V_f$ and $t_1,\ldots,t_k$ be terms in $\TN$. Then $\theta\colon \{x_1 \ass t_1,\ldots,x_k \ass t_k\}$ is called a $V_f$-substitution.
\end{definition}
The application of $V_f$-substitutions to terms in $\TN$ is defined in the usual way. 
\begin{definition}\label{def.N-formulas}
The set of {\em $\Ncal$-formulas} $\FN$ is defined as follows:
\begin{itemize}
\item If $P \in \PS_i$ and $t_1,\ldots,t_i$ are $\Ncal$-terms then $P(t_1,\ldots,t_i) \in \FN$,
\item if $F \in \FN$ then $\neg F \in \FN$,
\item if $F_1,F_2 \in \FN$ then $F_1 \land F_2,\ F_1 \lor F_2,\  F_1 \impl F_2 \in \FN$, 
\item if $x \in \Ncal \union V_f$, $F \in \FN$ and $y \in V_b$ then $\forall y.F\{x \ass y\},\ \exists y.F\{x \ass y\} \in \FN$.
\end{itemize}
An $\Ncal$-formula not containing parameters is simply called a {\em formula}.
\end{definition}
The semantics of $\TN$ and $\FN$ is based on so-called parameter assignments, a replacement of parameters by numerals. We denote the set of numerals by $\Num$.
\begin{definition}\label{def.par-assign}
A {\em parameter assignment} is a function $\Ncal \to \Num$. The set of all parameter assignments is denoted by $\Scal$. 
\end{definition}
Parameter assignments are usually denoted by $\sigma$. We define the application of $\sigma$ to terms in  $\TN$.
\begin{definition}\label{def.sigmaT}
Let $t \in \TN$ and $\sigma \in \Scal$; we define 
\begin{itemize}
\item for $n \in \Ncal$ $\sigma(n)$ is already defined. 
\item $\sigma(x) = x$ for $x \in V_f$, 
\item $\sigma(x) = x$ for $x \in V_b$,
\item $\sigma(\bar{0}) = \bar{0}$,
\item if $t \in \TN\setminus T^G_0$ then $\sigma(p(t)) = p(\sigma(t))$ ,$\sigma(s(t)) = s(\sigma(t))$,
\item if $t \in T^G_0$ then $\sigma(p(t)) = p(\sigma(t))\Eval$, $\sigma(s(t)) = s(\sigma(t))$, where $p(\sigma(t))\Eval$ is the numeral corresponding to the term $p(\sigma(t))$.
\item if $f \in \FS^i$ and $t_1,\ldots,t_i \in \TN$ then $\sigma(f(t_1,\ldots,t_i)) =  f(\sigma(t_1),\ldots,\sigma(t_i))$.
\end{itemize}
\end{definition}
\begin{example}
Let $f \in \FS_3$, $g \in \FS_1$, $n,m \in \Ncal$, $x \in V_f,\ y \in V_b$, $\sigma(n) = \bar{1}, \sigma(m) = \bar{3}$ (where we write $\bar{\alpha}$ for $s^\alpha(\bar{0})$). Then 
$$\sigma(f(g(x),p(g(n)),p(m))) = f(g(x),p(g(\bar{1})),\bar{2}).$$
\end{example}
It is evident that a parameter assignment maps an $\Ncal$-term ($\Ncal$-semi-term) to a term (semi-term) and thus defines a semantics on $\TN$.
\begin{proposition}\label{prop.sem-TN}
Let $\sigma$ be a parameter assignment and $t$ be an $\Ncal$-term ($\Ncal$-semi-term). Then $\sigma(t)$ is a term (semi-term).
\end{proposition}
\begin{proof}
A simple inductive argument shows that the right-hand-sides of the equations in Definition~\ref{def.sigmaT} are always semi-terms; 
if the term  $t$ to which $\sigma$ is applied is an $\Ncal$-term then $\sigma(t)$ is a term.
\end{proof} 
Parameter assignments on $\Ncal$-formulas behave homomorphically:
\begin{definition}\label{def.sigmaFN}
Let $\sigma$ be a parameter assignment, then 
\begin{itemize}
\item If $P \in \PS^i$ and $t_1,\ldots,t_i$ are $\Ncal$-terms then $\sigma(P(t_1,\ldots,t_i)) = P(\sigma(t_1),\ldots,\sigma(t_n))$,
\item $\sigma(\neg F) = \neg \sigma(F)$,
\item $\sigma(F_1 \circ F_2) = \sigma(F_1) \circ \sigma(F_2)$ for $\circ \in \{\land,\lor,\impl\}$, 
\item $\sigma(Q y.F) = Q y.\sigma(F)$ for $Q \in \{\forall, \exists\}$.
\end{itemize}
\end{definition}
\begin{example}
Let $m,n \in \Ncal$, $f \in \FS_1$ and $P,Q \in \PS_4$,  such that $\sigma(m) = \bar{0}$, $\sigma(n) = \bar{2}$ then 
\[
\begin{array}{l}
\sigma(\forall y.\exists z(P(f(y),z,s(m),p(n)) \impl Q(y,f(z),m,n))) =\\ 
\forall y.\exists z(P(f(y),z,\bar{1},\bar{1}) \impl Q(y,f(z),\bar{0},\bar{2})).
\end{array}
\]
\end{example}
\begin{proposition}\label{prop.semFN}
Let $F$ be an $\Ncal$-formula and $\sigma$ a parameter assignment. Then $\sigma(F)$ is a formula
\end{proposition}
\begin{proof}
trivial.
\end{proof}
For defining a version of {\LK} on this $\Ncal$-logic we need the extension of formulas to $\Ncal$-sequents. 
\begin{definition}\label{def.Nsequent}
Let $\Gamma,\Delta$ be multisets of $\Ncal$-formulas. Then $\Gamma \seq \Delta$ is an $\Ncal$-sequent.
\end{definition}
The extension of a parameter assignment to $\Ncal$-sequents is straightforward: Let $S$ be $A_1,\ldots,A_i \seq B_1,\ldots,B_j$. Then 
$$\sigma(S) = \sigma(A_1),\ldots,\sigma(A_i) \seq \sigma(B_1),\ldots,\sigma(B_j).$$
It follows that $\sigma(S)$ is an (ordinary) sequent. \\[1ex]
We have seen that the interpretation of $\Ncal$-terms and $\Ncal$-formulas via a parameter assignment is almost straightforward and behaves like a substitution. However, if we extend the application of $\sigma$ to proofs we face a problem as the following example shows. 
\begin{example}
Let $n \in \Ncal$ and $\varphi$ be the following {\LK}-proof (where $x \in V_b$, $P,Q \in \PS_1$)
\[
\infer[\forall\colon r]{\forall x.P(x) \seq \forall x(P(x) \lor Q(x))}
  { \infer[\forall\colon l]{\forall x.P(x) \seq P(n) \lor Q(n)}
   { \infer[\lor\colon r]{P(n) \seq P(n) \lor Q(n)}
     { P(n) \seq P(n)
      }
   }
  }
\]
Now assume that $\sigma(n) = \bar{0}$. Then we obtain $\sigma(\varphi)=$
\[
\infer[\forall\colon r]{\forall x.P(x) \seq \forall x(P(x) \lor Q(x))}
  { \infer[\forall\colon l]{\forall x.P(x) \seq P(\bar{0}) \lor Q(\bar{0})}
   { \infer[\lor\colon r]{P(\bar{0}) \seq P(\bar{0}) \lor Q(\bar{0})}
     { P(\bar{0}) \seq P(\bar{0})
      }
   }
  }
\]
Then, obviously, $\sigma(\varphi)$ is not an {\LK}-proof as the last inference is not a correct $\forall\colon r$-introduction. However, 
if we omit the last inference in $\varphi$ the application of $\sigma$ would indeed result in an {\LK}-proof. We see that we can quantify via $\forall\colon l$ over a parameter, but not via $\forall\colon r$. The case for $\exists\colon l$ corresponds to $\forall\colon r$. We see that strong quantification over parameters ($\forall\colon r$, $\exists\colon l$) is harmful, but weak quantification ($\forall\colon l$, $\exists\colon r$) is not.
\end{example}
\begin{definition}[$\LKN$]\label{def.LK-N}
$\LKN$ is a sequent calculus over $\Ncal$-sequents, where the rules are the same as for {\LK} in Definition~\ref{def.LK} with the exception that $\forall\colon r$ and $\exists\colon l$ have new definitions as shown below.
\[
\infer[\forall\colon r]{\Gamma \seq \Delta, \forall y.A(y)}
   { \Gamma \seq \Delta, A(x)}
\]
where $x \in V_f$, $y \in V_b$ such that $x$ does not occur in $\Gamma,\Delta$. Note that $x$ may not be in $\Ncal$.
\[
\infer[\exists\colon l]{\exists y.A(y), \Gamma \seq \Delta}
  {A(x),\Gamma \seq \Delta}
\]
where the conditions are the same as for $\forall\colon l$.\\[1ex]
$\Acal$ is called an $\Ncal$-axiom  set if $\Acal$ consists of atomic sequents which are closed under $V_f$-substitutions. The definition of $\LKN$-proofs is completely analogous to that of $\LK$-proofs, taking into account the different quantifier rules and the different concept of axiom set. As we consider predicate logic with equality we admit not only tautological atomic sequents as axioms but also all atomic sequents which are valid in predicate logic with equality; in particular we admit reflexivity, transitivity, symmetry and substitution axioms of $=$, i.e. axioms of  the form $s=t, P[s]_\Lambda \seq P[t]_\Lambda$ and  $s=t, P[t]_\Lambda \seq P[s]_\Lambda$, where $\Lambda$ is an occurrence of $s$ which is replaced by $t$. If $\varphi$ is an $\LKN$-proof and $(n_1,\ldots,n_k)$ is a list of all parameters occurring in $\varphi$ we denote $\varphi$ by $\varphi(n_1,\ldots,n_k)$. 
\end{definition}
We can now define the semantics of $\LKN$-proofs via parameter assignments.
\begin{definition}\label{def.sigma-LKN}
Let $\sigma$ be a parameter assignment and $\varphi(n_1,\ldots,n_k)$ be an $\LKN$-proof. Let us assume that $\sigma(n_i) =r_i$ for $
i =1,\ldots,k$ . Then $\sigma(\varphi(n_1,\ldots,n_k))$ (also written as $\varphi(r_1,\ldots,r_k)$) is defined inductively:
\begin{itemize}
\item $\varphi= S$ where $S$ is an $\Ncal$-sequent. Then $\sigma(\varphi) = \sigma(S)$.
\item $\varphi$ is of the form 
\[
\infer[\xi]{S}{\varphi'}
\]
where $\xi$ is a unary rule. Then $\sigma(\varphi)=$
\[
\infer[\xi]{\sigma(S)}{\sigma(\varphi')}
\]
\item $\varphi$ is of the form 
\[
\infer[\xi]{S}{ \varphi_1 & \varphi_2}
\]
where $\xi$ is a binary rule. Then $\sigma(\varphi)=$
\[
\infer[\xi]{\sigma(S)}{ \sigma(\varphi_1) & \sigma(\varphi_2)}
\]
\end{itemize}
\end{definition}
\begin{example}\label{ex.sigma-LKN}
Let $n \in \Ncal$, $y \in V_f$, $x,z \in V_b$ and $\varphi(n)$ be the $\LKN$-proof 
\[
\infer[\forall\colon l]{\forall z\forall x.P(z,x) \seq \forall x(P(n,x) \lor Q(x))}
{
\infer[\forall\colon r]{\forall x.P(n,x) \seq \forall x(P(n,x) \lor Q(x))}
  { \infer[\forall\colon l]{\forall x.P(n,x) \seq P(n,y) \lor Q(y)}
   { \infer[\lor\colon r]{P(n,y) \seq P(n,y) \lor Q(y)}
     { P(n,y) \seq P(n,y)
      }
   }
  }
}
\]
Let $\sigma(n) = \bar{1}$. Then $\sigma(\varphi)(n)= \varphi(\bar{1}) =$
\[
\infer[\forall\colon l]{\forall z\forall x.P(z,x) \seq \forall x(P(\bar{1},x) \lor Q(x))}
{
\infer[\forall\colon r]{\forall x.P(\bar{1},x) \seq \forall x(P(\bar{1},x) \lor Q(x))}
  { \infer[\forall\colon l]{\forall x.P(\bar{1},x) \seq P(\bar{1},y) \lor Q(y)}
   { \infer[\lor\colon r]{P(\bar{1},y) \seq P(\bar{1},y) \lor Q(y)}
     { P(\bar{1},y) \seq P(\bar{1},y)
      }
   }
  }
}
\]
Obviously, $\varphi(\bar{1})$ is an $\LK$-proof.
\end{example}
\begin{proposition}\label{prop.sem-LKN}
Let $\varphi(n_1,\ldots,n_k)$ be an $\LKN$-proof of an $\Ncal$-sequent $S$ (from arbitrary initial sequents) and $\sigma$ be a parameter assignment. Then $\sigma(\varphi(n_1,\ldots,n_k))$ is an $\LK$-proof of $\sigma(S)$.
\end{proposition}
\begin{proof}
We proceed by induction on the number of proof nodes in $\varphi$:
\begin{itemize}
\item Let $\varphi = S$ for some $\Ncal$-sequent $S$. Then, by Definition~\ref{def.sigma-LKN} $\sigma(\varphi) = \sigma(S)$ and 
$\sigma(\varphi)$ is an $\LK$-proof of $\sigma(S)$.
\item Assume that $\varphi=$
\[
\infer[\xi]{S}
  {\deduce{S'}{(\varphi')}
  }
\]
for a unary rule $\xi$. We have to show that $\sigma(\varphi)$, which by Definition~\ref{def.sigma-LKN} is equal to 
\[
\infer[\xi]{\sigma(S)}
  {\deduce{\sigma(S')}{(\sigma(\varphi'))}
  }
\]
is an $\LK$-proof of $\sigma(S)$. By induction hypothesis $\sigma(\varphi')$ is an $\LK$-proof of $\sigma(S')$. If the inference $\xi$ is either $\land\colon l$ or $\lor\colon r$ then, for $\circ \in \{\land,\lor\}$, $\sigma(A \circ B) = \sigma(A) \circ \sigma(B)$ (by Definition~\ref{def.sigmaFN}) and thus $\xi$ infers the principal formula $\sigma(A) \circ \sigma(B)$ from the auxiliary formulas 
$\sigma(A),\sigma(B)$. A similar argument holds for $\xi = \neg\colon l$ and $\xi = \neg\colon r$. The cases of contraction and weakening are trivial. Indeed, $\sigma$ behaves in a ``homomorphic'' way. For illustration we show the case $\land\colon l$. Here we have $\varphi=$
\[
\infer[\land\colon l]{A \land B,\Gamma \seq \Delta}
 { \deduce{A,B,\Gamma \seq \Delta}{(\varphi')}
 }
\]
and, by Definition~\ref{def.sigma-LKN}, $\sigma(\varphi)=$
\[
\infer[\land\colon l]{\sigma(A \land B), \sigma(\Gamma) \seq \sigma(\Delta)}
 { \deduce{\sigma(A),\sigma(B), \sigma(\Gamma) \seq \sigma(\Delta)}{(\sigma(\varphi'))}
 }
\]
By induction hypothesis and $\sigma(A \land B) = \sigma(A) \land \sigma(B)$ $\sigma(\varphi)$ is an $\LK$-proof the sequent 
$$\sigma(A \land B,\Gamma \seq \Delta).$$
It remains to investigate the case where $\xi$ is an introduction rule of a quantifier. We investigate the cases $\forall\colon l$ and 
$\forall\colon r$ (the existential introductions are symmetric). So let $\varphi=$
\[
\infer[\forall\colon l]{\forall x.A,\Gamma \seq \Delta}
  { \deduce{A\{x \ass t\}, \Gamma \seq \Delta}{(\varphi')}
   }
\]
By Definition~\ref{def.sigma-LKN} we get $\sigma(\varphi)=$
\[
\infer[\forall\colon l]{\sigma(\forall x.A),\sigma(\Gamma) \seq \sigma(\Delta)}
  { \deduce{\sigma(A\{x \ass t\}), \sigma(\Gamma) \seq \sigma(\Delta)}{(\sigma(\varphi'))}
   }
\]
By Definition~\ref{def.sigmaFN} $\sigma(\forall x.A) = \forall x.\sigma(A)$. Moreover, 
$$\sigma(A\{x \ass t\}) = \sigma(A)\{x \ass \sigma(t)\}$$
as $x \in V_b$ and thus $\sigma(x)=x$. Therefore $\sigma(\varphi)$ is an $\LK$-proof.\\[1ex]
Now consider $\varphi=$
\[
\infer[\forall\colon r]{\Gamma \seq \Delta, \forall x.A}
{ \deduce{\Gamma \seq \Delta, A\{x \ass y\}}{(\varphi')}
}
\]
$y$ is an eigenvariable and, by definition of $\LKN$, we have $x \in V_f$ (and $x \not \in \Ncal$). 
By Definition~\ref{def.sigma-LKN} we have $\sigma(\varphi)=$
\[
\infer[\forall\colon r]{\sigma(\Gamma) \seq \sigma(\Delta), \sigma(\forall x.A)}
{ \deduce{\sigma(\Gamma) \seq \sigma(\Delta), \sigma(A\{x \ass y\})}{(\sigma(\varphi'))}
}
\]
But $\sigma(\forall x. A) = \forall x.\sigma(A)$ and $\sigma(A\{x \ass y\}) = \sigma(A)\{x \ass y\}$ as $x,y$ are in $V_f$ and thus are not parameters, and therefore $\sigma(x) =x, \sigma(y) = y$. So $\sigma(\varphi)$ is an $\LK$-proof.
\item The last rule of $\varphi$ is a binary rule: trivial, as for logical rules introducing $A \circ B$ we have 
$\sigma(A \circ B) = \sigma(A) \circ \sigma(B)$. If the binary rule is a cut with formula $A$ we get a cut with formula $\sigma(A)$.
\end{itemize}
  
\end{proof}
$\LKN$-proofs can be plugged together via calls and parameter substitutions; the result of such operations is an inductive proof schema. In order to define these proof schemata we need parameter substitutions ($\Ncal$-substitutions) and  several syntax extensions.
\begin{definition}\label{def.par-subst}
Let $n_1,\ldots,n_k$ be (different) variables in $\Ncal$ and $t_1,\ldots,t_k \in T_0$. Then the substitution $\{n_1 \ass t_1,\ldots,n_k \ass t_k\}$ is called a {\em parameter substitution}. The set of all parameter substitutions is denoted by $\SubN$.
\end{definition}
Parameter substitutions should not be confused with parameter assignments; parameter assignments are mappings from $\Ncal$ to 
$\Num$ and not substitutions on $\Ncal$ with range in $T_0$. In fact, parameter assignments define the semantics, where parameter substitutions are part of the syntax. Parameter substitutions are ordinary substitutions and can be extended to formulas, sequents and proofs in an obvious way. For schemata we need expressions representing proofs, so-called proof terms.
\begin{definition}\label{def.proofterms}
Let $\Pi$ be an infinite set of symbols called {\em proof symbols}. To each symbol $\rho \in \Pi$ we can assign a finite set of parameters 
$\Ncal(\rho)$ and an $\Ncal$-sequent $\Seq(\rho)$ such that the set of parameters in $S$ (from now on denoted by $\Ncal(S)$) is $\Ncal(\rho)$. Let $\Ncal(\rho) = \{n_1,\ldots,n_k\}$ (the order of the parameters does not matter) then the expression $\rho(n_1,\ldots,n_k)$ is called a {\em proof term for $\rho$}.
\end{definition}
Like ordinary terms, proof terms can be subject to substitutions. So let $\rho(n_1,\ldots,n_k)$ be a proof term and $\theta$ be a parameter substitution with domain $\{n_1,\ldots,n_k\}$. Then $\rho(n_1,\ldots,n_k)\theta = \rho(n_1\theta,\ldots,n_k\theta)$ and 
also $\rho(n_1\theta,\ldots,n_k\theta)$ is a proof term for $\rho$.
We also define 
\begin{eqnarray*}
\Seq(\rho(n_1,\ldots,n_k)) &=& \Seq(\rho),\\
\Seq(\rho(n_1\theta,\ldots,n_k\theta) &=& \Seq(\rho)\theta
\end{eqnarray*}
We can now define derivations corresponding to proof terms $\rho(n_1,\ldots,n_k)$ having initial sequents which are either axioms or proof terms; we call these derivations $s$-proofs. 
\begin{definition}\label{def.sproof}
Let $\rho \in \Pi$, $\Ncal(\rho) =\{n_1,\ldots,n_k\}$ and $\Seq(\rho)$ be defined. Let $\varphi(\rho,n_1,\ldots,n_k)$ be an 
$\LKN$-derivation of $\Seq(\rho)$ with $\Ncal(\varphi(\rho,n_1,\ldots,n_k)) = \Ncal(\rho)$ ($\varphi(\rho,n_1,\ldots,n_k)$ need not be a proof but may have arbitrary initial $\Ncal$-sequents). Let $\varphi^*(\rho,n_1,\ldots,n_k)$ be the following modification of 
$\varphi(\rho,n_1,\ldots,n_k)$: 
initial sequents are either taken from $\varphi(\rho,n_1,\ldots,n_k)$ or initial sequents $S'$ in $\varphi(\rho,n_1,\ldots,n_k)$  are replaced by a (new type of) inference of the form 
\[
\infer[\rho'\theta]{S'}{\rho'(m_1,\ldots,m_l)\theta}
\]
where $\rho'(m_1,\ldots,m_l)\theta$ is a proof term (corresponding possibly to another proof symbol $\rho'$) 
such that $\Seq(\rho')\theta = S'$. Then $\varphi^*(\rho,n_1,\ldots,n_k)$ is called an {\em $s$-proof corresponding to $\rho$}. 
$\varphi(\rho,n_1,\ldots,n_k)$ is called the {\em core proof} of $\varphi^*(\rho,n_1,\ldots,n_k)$. 
\end{definition}
\begin{example}\label{ex.sproof}
Let $\rho$ be a proof symbol, $f \in \FS_1,\fhat \in \FS_2, P \in \PS_1$ and $c$ a constant symbol such that 
\begin{eqnarray*}
\Ncal(\rho) &=& \{n\} \mbox{ and }\\
\Seq(\rho) &=& {\it fdef}, \forall x(P(x) \impl P(f(x))) \seq P(c) \impl P(\fhat(c,n)) \mbox{ where }\\
{\it fdef}  &=& \forall x.\fhat(x,\bar{0}) = x, \forall x \forall z(\fhat(x,s(z)) = f (\fhat(x,z))).
\end{eqnarray*}
We define $\varphi(\rho,n)$; by definition $\varphi(\rho,n)$ must be an $\LKN$-derivation of $\Seq(\rho)$ and 
$\Ncal(\varphi(\rho,n)) = \{n\}$. In defining $\varphi(\rho,n)$  we assume $n>0$.  $\varphi(\rho,n)=$
\[
\infer[{\it cut}+c^*]{{\it fdef},\forall x(P(x) \impl P(f(x))) \seq P(c) \impl P(\fhat(c,n))}
{ {\it fdef},\forall x(P(x) \impl P(f(x))) \seq P(c) \impl P(\fhat(c,p(n))))
  &
 \deduce{S}{(\psi)}
 } 
\]
where 
$$S  = P(c) \impl P(\fhat(c,p(n))), {\it fdef},\forall x(P(x) \impl P(f(x))) \seq P(c) \impl P(\fhat(c,n)).$$
and $\psi=$
\[
\infer[\impl\colon r]{S\colon P(c) \impl P(\fhat(c,p(n))), {\it fdef},\forall x(P(x) \impl P(f(x))) \seq P(c) \impl P(\fhat(c,n))}
  { \infer[w\colon l]{P(c), P(c) \impl P(\fhat(c,p(n))), {\it fdef},\forall x(P(x) \impl P(f(x))) \seq P(\fhat(c,n))}
    { \infer[\impl\colon l]{P(c), P(c) \impl P(\fhat(c,p(n))),\forall x,z(\fhat(x,s(z)) = f (\fhat(x,z))),\forall x(P(x) \impl P(f(x))) \seq     
         P(\fhat(c,n))}
     { P(c) \seq P(c) 
                     &
         \deduce{S_1}{(\psi_1)}
      }
      }
   }
\]
where $\psi_1=$
\[
\infer[\forall\colon l]{S_1\colon P(\fhat(c,p(n))),\forall x,z(\fhat(x,s(z)) = f (\fhat(x,z))),\forall x(P(x) \impl P(f(x))) \seq     
         P(\fhat(c,n))}
  { \infer[\impl\colon l]{P(\fhat(c,p(n))),\forall x,z(\fhat(x,s(z)) = f (\fhat(x,z))), P(\fhat(c,p(n))) \impl P(f(\fhat(c,p(n)))) \seq 
	      P(\fhat(c,n))}
     {    P(\fhat(c,p(n))) \seq P(\fhat(c,p(n)))
            &
          \deduce{S_2}{(\psi_2)}
      }
    }
\]
where $\psi_2=$
\[
\infer[\forall\colon l^*]{S_2\colon \forall x,z(\fhat(x,s(z)) = f (\fhat(x,z))),P(f(\fhat(c,p(n)))) \seq P(\fhat(c,n))}
  { \infer[{\it cut}]{\fhat(c,s(p(n))) = f (\fhat(c,p(n))), P(f(\fhat(c,p(n)))) \seq P(\fhat(c,n))}
     { \fhat(c,s(p(n))) = f (\fhat(c,p(n))), P(f(\fhat(c,p(n)))) \seq P(\fhat(c,s(p(n))))
         &
        \deduce{S_3}{(\psi_3)}
     }
   }
\]
for $\psi_3=$
\[
\infer[{\it cut}]{S_3\colon P(\fhat(c,s(p(n)))) \seq P(\fhat(c,n))}
   { \seq s(p(n)) = n
       &
      s(p(n))=n, P(\fhat(c,s(p(n)))) \seq P(\fhat(c,n))
   }
\]
Note that we assumed $n>0$ and so the equation $s(p(n))=n$ is valid under $n>0$ and we can use $\seq s(p(n))=n$ as an axiom! Axioms of this type are not admitted in $\LKN$ as they are not valid, but we will define these kind of axioms (which depend on conditions) formally below.\\[1ex]
By replacing the leftmost leaf in $\varphi(\rho,n)$ (for $n>0$) we obtain $\varphi^*(\rho,n)=$
\[
\infer[{\it cut}]{{\it fdef},\forall x(P(x) \impl P(f(x))) \seq P(c) \impl P(\fhat(c,n))}
{ \infer[\rho\{n \ass p(n)\}]{{\it fdef},\forall x(P(x) \impl P(f(x))) \seq P(c) \impl P(\fhat(c,p(n)))}
   {\rho(p(n))
   }
  &
 \psi
 } 
\]
So, here, the $\rho'$ in Definition~\ref{def.sproof} is $\rho$ itself and represents a recursive call. 
\end{example}
Obviously, the proof $\varphi^*(\rho,n)=$ represents a recursive call only for $n> \bar{0}$ and the base case for $n=0$ is missing. If we take into account this case distinction we can obtain a proof schema.
\begin{example}\label{ex.proofschema}
Let $\varphi^*(\rho,n)$ be the $s$-proof of Example~\ref{ex.sproof}. We have to find a base case, i.e. an $\LKN$-proof of the sequent 
$${\it fdef}, \forall x(P(x) \impl P(f(x))) \seq P(c) \impl P(\fhat(c,\bar{0})).$$
The proof we give is even an $\LK$-proof as, by the replacement of $n$ by $\bar{0}$, our proof term becomes 
$\rho(\bar{0})$ and is parameter-free. So we obtain an $\LK$-proof $\chi$ corresponding to the proof term $\rho(\bar{0})$ for $\chi=$
\[
\infer[w^*]{{\it fdef}, \forall x(P(x) \impl P(f(x))) \seq P(c) \impl P(\fhat(c,\bar{0}))}
 { \infer[\forall\colon l]{\forall x.\fhat(x,\bar{0})=x \seq P(c) \impl P(\fhat(c,\bar{0}))}
    { \infer[\impl\colon r]{\fhat(c,\bar{0}) = c \seq P(c) \impl P(\fhat(c,\bar{0}))}
     { P(c), \fhat(c,\bar{0}) = c \seq  P(\fhat(c,\bar{0}))
      }
    }
  }
\]
Note that the leaf of $\chi$ is an equality axiom. Thus the full recursion can be described as 
$$\rho(n) \defeq \{n=0\colon \chi,\ n>0\colon \varphi^*(\rho,n)\}.$$
Here we have a partition as it is also required for point transition systems.
It remains to define the semantics for $s$-proofs $\varphi^*(\rho,n)$ (i.e. to define $\sigma(\varphi^*(\rho,n))\Eval$ and to show that for all $\sigma$ the run $\sigma(\rho_n)\Eval$ yields an {\LK}-proof. 
\end{example}
We first define a so-called {\em semi-proof schema} which (in case of termination under a parameter assignment) yields an {\LK}-proof. Proof schemata will be semi-proof schemata which terminate for every parameter assignment (and thus always yield {\LK}-proofs). 
\begin{definition}\label{def.semi-proof-schema}
A semi-proof schema $\Pbf$ is a tuple $(\rho_0,R,\Ncal_0,t,\Seq,\Pi)$ where
\begin{itemize}
\item $R$ is a set of proof symbols and $\rho_0 \in R$.
\item $\Ncal_0$ is a finite set of parameters (which we will describe by $\{n_1,\ldots,n_k\}$ - if not explicitly defined differently).
\item $t$ is a mapping assigning to every $\rho \in R$ a tuple of parameters  $t(\rho)\colon (m_1,\ldots,m_l)$ for 
$\{m_1,\ldots,m_l\} \IN \Ncal_0$.   
\item $\Seq$ is a mapping assigning to every $\rho \in R$ an $\Ncal$-sequent $\Seq(\rho)$ such  that $\Seq(\rho)$ contains the parameters in $\Ncal(\rho)$.
\end{itemize}
Now let $\rho \in R$ then $\rho(t(\rho))$ is called a {\em proof term in $\Pbf$}; by definition $\rho(t(\rho))$ is of the form $\rho(n_1,\ldots,n_l)$. \\[1ex]
$\Pi$ is a mapping assigning to every proof term $\rho(m_1,\ldots,m_l)$ in $\Pbf$ a definition $D(\rho)$ of the form   
\begin{eqnarray*}
\rho(m_1,\ldots,m_l) &\defeq& \{C_1\colon \xi_1(\rho)(m_1,\ldots,m_l), \ldots, C_k\colon \xi_k(\rho)(m_1,\ldots,m_l)\}
\end{eqnarray*}
where $(\C_1,\ldots,C_k)$ is a partition and the $\xi(\rho)(m_1,\ldots,m_l)$ are $s$-proofs corresponding to $\rho$ where the leaves of $\xi(\rho)(m_1,\ldots,m_l)$ which are not axioms are instances of proof terms in $\Pbf$.  We also extend the set of axioms. 
Besides the tautological and equality axioms (which are valid in predicate logic with equality) 
we admit in the proofs $\xi_j$ (defined under the condition $C_j$) axiom sequents of the form 

$$S\colon A_1,\ldots,A_n \seq B_1,\ldots,B_m$$
where the $A_i$ and $B_j$ are of the form $s \circ t$ where $\circ \in \{<,>,=,\leq,\geq\}$ such that $s,t,$ are terms in $T_0$ and 
$C_j \models_{\Mcal} S$ under the standard interpretation $\Mcal$ of the conditions extended by the standard interpretation of $p$.
\\[1ex]
The set of all defining equations in $\Pi$ is denoted as $\Dcal(\Pi)$. $\rho_0$ is called the {\em main proof symbol} of 
$\Pbf$.
\end{definition}
Note that the new axioms $S$ defined above do not necessarily consist of atomic {\em conditions} but may contain atoms $s \circ t$ where $s$ and $t$ may contain the predecessor $p$; that such an extension of the syntax is unavoidable is shown in Example~\ref{ex.sproof} where the sequent $\seq s(p(x)) = x$ appears as an axiom which is valid under the condition $x>0$. There the predecessor does not only appear in a proof-call but also in a nonrecursive side proof. The problem $C_j \models_{\Mcal} S$ above is decidable as it can be transformed into an equivalent set of problems not containing the predecessor $p$. By the elimination of $p$ we obtain problems in quantifier-free Presburger arithmetic where efficient solution methods exist (see, e.g.,~\cite{Haase.2018}). We illustrate this transformation for the sequent $S\colon \seq s=t$ where $s$ and $t$ may contain the predecessor. The cases for the other predicates are analogous.  For sequents with more than one atom the $A_i$, $B_j$ containing $p$ have to eliminated successively.  We consider now the different types of equations $s = t$ for terms $s,t \in T_0$. By $p(s(x)) = x$ it suffices to consider the following cases for the equation $s=t$: 
\begin{itemize}
\item[a.]$s^kp^n\bar{0} = s^lp^m\bar{0}$ for $k,n,l,m \geq 0$, $n \geq m$.
\item[b.] $s^kp^nx = s^lp^m\bar{0}$ for $k,n,l,m \geq 0$, $n \geq m$ and $x \in \Ncal$.
\item[c.]  $s^kp^n x = s^lp^m x$ for $k,n,l,m \geq 0$, $n \geq m$ and $x \in \Ncal$.
\item[d.] $s^kp^n x = s^lp^m y$ for $k,n,l,m \geq 0$, $n \geq m$ and $x,y \in \Ncal$, $x \neq y$.
\end{itemize}
We have to decide $C_j \models s=t$. The case a is trivial as $s^kp^n\bar{0}\Eval, s^lp^m\bar{0}\Eval$ are numerals which are either equal (then we replace the equation by $\top$) or different - in which case we replace the equation by $\bot$. We only consider case d in detail, the others are analogous.   \\[1ex]
If, in case d, we have $x \geq n$ and $y \geq m$, the equation $s^kp^nx = s^lp^my$ can be written as 
\begin{eqnarray*}
(x-n)+k &=& (y-m)+l \mbox{ as then }p\mbox{ acts as }-1, \mbox{ and so }\\
x+k &=& y+l +(n-m).
\end{eqnarray*}
Note that $k,l,m,n$ are constants and so the last equation is equivalent to $x+k = y+r$ for some $r \geq 0$. Now 
$$(1)\colon C_j  \land x \geq n \land y \geq m\models x+k = y+r$$ 
is decidable within quantifier-free Presburger arithmetic. We now consider the remaining cases 
\begin{itemize}
\item $x <n$ and $y \geq m$,
\item $x \geq n$ and $y <m$,
\item $x <n$ and $y <m$.
\end{itemize}
In case $x<n$ and $y \geq m$ we consider the problems
$$(2i)\colon C_j \land x = i \land y \geq m \models s^kp^n i\Eval + m = y+l$$ 
for $i=0,\ldots,n-1$.\\[1ex]
For $x \geq n$ and $y<m$ we consider 
$$(3i)\colon C_j \land x \geq n \land y =i \models x+k = s^lp^m i\Eval + n$$ 
for $i=0,\ldots,m-1$ and, finally, for $x<n$ and $y<m$ 
$$(4ij)\colon C_j \land x=i \land y =k \models s^kp^n i\Eval = s^lp^m k\Eval$$ 
for $i=0,\ldots,n-1$  and  $k =1,\ldots,m-1$.\\[1ex]
As the conditions $(1), (2i),(3i),(4ij)$ define a partition, $C_j \models  s=t$ (and therefore $C_j \models\ \seq s=t$) iff all all the conditions  $(1), (2i),(3i),(4ij)$ are valid.\\[1ex]
In practice we can avoid the application of expressions like $s^lp^k x$ to numerals smaller than $k$, thus avoiding a non-invertible behavior of $p$. Given the sequent  $S\colon A_1,\ldots,A_n \seq B_1,\ldots,B_m$ we identify, for every variable $x$ occurring in $S$, the term $s^lp^k x$ with maximal $k$ and postulate the condition $x \geq k$. This simplifies the application of the algorithm considerably and the transformation to a problem in quantifier-free Presburger arithmetic is polynomial.
\begin{example}
Consider the sequent  $\seq ssppx = y$ and let $C$ be a condition. It suffices to consider the cases $x \geq 2$ and $x<2$ as both $l$ and $m$ are $0$. So we obtain 
\begin{eqnarray*}
(1)\ C \land x \geq 2 &\models& ssppx = y,\\
(2-1)\ C \land x=0 &\models& 2=y,\\
(2-2)\ C \land x=1 &\models& 2 = y.
\end{eqnarray*}
If $C$ is the condition $x=y$ then, obviously (2-1) and (2-2) are false and so $C \not \models\ \seq ssppx = y$. However if $C$ is the condition  $x=y \land x > 3$ then $C \models\ \seq ssppx = y$; note that, under that condition $x \geq 2$ and thus also under $x>3 \land x=y$, $ssppx=y$ boils down to $x=y$.
\end{example}
We extend now the proofs in Example~\ref{ex.proofschema} to a more complex semi-proof schema.
\begin{example}\label{ex. semi-proof-schema}
Let $\Pbf$ be the semi-proof schema $(\rho_0,\{\rho_0,\rho\},\{n\},t,\Seq,\Pi)$ where 
$\rho$ is defined in Example~\ref{ex.proofschema}. In particular we have  
\begin{eqnarray*}
t(\rho) &=& (n),\\
\Seq(\rho) &=& {\it fdef},\forall x(P(x) \impl P(f(x))) \seq P(c) \impl P(\fhat(c,n)),\\
\Pi(\rho(n)) &=& \rho(n) \defeq \{n=0\colon \xi_1(\rho)(n),\ n>0\colon \xi_2(\rho)(n)\},
\end{eqnarray*}
where $\xi_1(\rho)(n) = \chi$ and $\xi_2(\rho)(n) = \varphi^*(\rho,n)$ where $\chi$ is defined in Example~\ref{ex.proofschema} and 
$\varphi^*(\rho)(n)$ in Example~\ref{ex.sproof}. We define 
\begin{eqnarray*}
\Seq(\rho_0) &=& P(c), {\it fdef}, \forall x(P(x) \impl P(f(x))) \seq P(\fhat(c,s(n))),\\
\Pi(\rho_0(n)) &=& \rho_0(n) \defeq \{n=0\colon \xi_1(\rho_0)(n),\ n>0\colon \xi_2(\rho_0)(n)\}
\end{eqnarray*}
where $\xi_1(\rho_0)(n)=$
\[
\infer[\forall\colon l^*]{P(c), {\it fdef}, \forall x(P(x) \impl P(f(x))) \seq P(\fhat(c,s(0)))}
{ \infer[\forall\colon l ]{ P(c),\fhat(c,0)=c,\fhat(c,s(0)) =f(\fhat(c,0)) , \forall x(P(x) \impl P(f(x))) \seq P(\fhat(c,s(0)))}
  { \infer[{\it cut}+c^*]{P(c),\fhat(c,0)=c, \fhat(c,s(0)) =f(\fhat(c,0)) ,  P(c) \impl P(f(c)) \seq P(\fhat(c,s(0)))}
     {\deduce{\Gamma \seq P(f(\fhat(c,0)))}{(\xi'_1)}
      &
       \infer[w^*]{P(f(\fhat(c,0))),\Gamma \seq P(\fhat(c,s(0)))}
	          {  P(f(\fhat(c,0))),  \fhat(c,s(0)) =f(\fhat(c,0))  \seq  P(\fhat(c,s(0)))
               }  
     }
  }
}
\]
For $\Gamma = P(c),\fhat(c,0)=c, \fhat(c,s(0)) =f(\fhat(c,0)) ,  P(c) \impl P(f(c))$ and $\xi'_1=$
\[
\infer[{\it cut}+c^*]{\Gamma \seq P(f(\fhat(c,0)))}
    { \infer[w^*]{\Gamma \seq P(f(c))}
       { \infer[\impl\colon l]{P(c), P(c) \impl P(f(c)) \seq P(f(c))}
          { P(c) \seq P(c)
            &
            P(f(c)) \seq P(f(c))
          }
      }
        &
     \infer[w^*]{ P(f(c)), \Gamma \seq P(f(\fhat(c,0)))}{P(f(c)), \fhat(c,0)=c \seq P(f(\fhat(c,0)))}
    }
\]
$\xi_2(\rho_0)(n)=$
\small
\[
\infer[{\it cut} ]{P(c), {\it fdef}, \forall x(P(x) \impl P(f(x))) \seq P(\fhat(c,s(n)))}
 { \infer[(\rho,\theta)]{{\it fdef}, \forall x(P(x) \impl P(f(x)))  \seq P(c) \impl  P(\fhat(c,s(n)))}
   { \rho(s(n))
    }
   & 
  \infer[\impl\colon l]{P(c), P(c) \impl  P(\fhat(c,s(n))) \seq P(\fhat(c,s(n)))}
   { P(c) \seq P(c)
     &
      P(\fhat(c,s(n))) \seq P(\fhat(c,s(n)))
   }
 }
\]
\normalsize  
for $\theta = \{n \ass s(n)\}$.
\end{example}
A semi-proof schema can be considered as a program, given a parameter assignment as input, producing an $\LK$- proof -- provided it terminates. We have already defined $\sigma(\varphi)$ for $\LKN$-proofs. We extend this semantics to semi-proof schemata. 
\begin{definition}[semantics of semi-proof schemata]\label{def.sem-proofterms}\mbox{ }\\
Let $\sigma$ a parameter assignment and $\Pbf\colon (\rho_0,R,\Ncal_0,t,\Seq,\Pi)$ be a proof schema such that $\Ncal_0 = \{n_1,\ldots,n_k\}$ and let 
$\rho(m_1,\ldots,m_l)$ be the proof term corresponding to $\rho$. Then $\Pi$ assigns to $\rho(m_1,\ldots,m_l)$ a  definition $D(\rho)$ of the form   
\begin{eqnarray*}
\rho(m_1,\ldots,m_l) &\defeq& \{C_1\colon \xi_1(\rho)(m_1,\ldots,m_l), \ldots, C_j\colon \xi_j(\rho)(m_1,\ldots,m_l)\}
\end{eqnarray*}
Now let us assume that $\sigma(m_i) = r_i$ for $i=1,\ldots,k$. Then, as $(C_1,\ldots,C_j)$ is a partition, there is exactly one $j \in \{1,\ldots,k\}$ such that  $\sigma(C_j) = \top$. So we define 
$$\sigma(\rho(m_1,\ldots,m_l))\Eval = \rho(r_1\ldots,r_l)\Eval = \sigma(\xi_j(\rho)(m_1,\ldots,m_l))\Eval.$$
We have to define $\sigma(\xi_j(\rho)(m_1,\ldots,m_l))\Eval$. We know that $\xi_j(\rho)(m_1,\ldots,m_l)$ is an s-proof. We consider two cases 
\begin{itemize}
\item[a.] $\xi_j(\rho)(r_1,\ldots,r_l)$ is an $\LKN$-proof. Then 
$\sigma(\xi_j(\rho)(m_1,\ldots,m_l))\Eval = \xi_j(\rho)(r_1,\ldots,r_l)$.
\item[b.] $\xi_j(\rho)(m_1,\ldots,m_l)$ is not an $\LKN$-proof. Then  $\xi_j(\rho)(m_1,\ldots,m_l)$ is an $s$-proof with a $\LKN$-core  $\zeta$.  In $\xi_j(\rho)(m_1,\ldots,m_l)$ there are leaves which are not axioms but proof terms. Let 
$\rho'(k_1,\ldots,k_s)\theta$ be such a proof term at a position $\lambda$ with 
$$\theta = \{k_1 \ass t_1,\ldots,k_s \ass t_s\}.$$ 
Let $\Lambda'$ be the set of all positions $\lambda'$ such that $\lambda'$ is a leaf position in $\zeta$ corresponding to $\lambda$ (i.e. the consequent of $\lambda$) and $S(\lambda')$ be  the sequent occurring at $\lambda'$. Then we replace in $\sigma(\zeta)$ the leaf 
$\sigma(S(\lambda'))$ by $\sigma'(\rho'(k_1,\ldots,k_s))\Eval$ for 
$$\sigma'(k_i) = \sigma(t_i) \mbox{ for } i=1,\ldots,s.$$

We do this for all proof terms occurring at the leaves and, in case of termination, obtain $\sigma(\xi_j(\rho)(m_1,\ldots,m_l))\Eval$. 
Note that by definition of $\sigma(t_i)$ for terms $t_i \in T_0$, $\sigma(t_i)$ is a numeral. If for all positions $\lambda' \in \Lambda'$ the evaluation $\sigma'(\rho'(k_1,\ldots,k_s))\Eval$ terminates with an $\LK$-proof $\chi(\lambda')$ then 
$$\sigma(\xi_j(\rho)(m_1,\ldots,m_l))\Eval = \sigma(\zeta)\{S(\lambda')\setminus \chi(\lambda') \mid \lambda' \in \Lambda'\}.$$
\end{itemize}
When we apply this definition to the the evaluation of the proof term $\rho_0(n_1,\ldots,n_k)$ and all resulting computations terminate 
then $\sigma(\rho_0(n_1,\ldots,n_k))\Eval$ is the $\LK$-proof of $\Pbf$ corresponding to $\sigma$.
\end{definition}
\begin{example}\label{ex.semantics}
Consider the parameter assignment $\sigma(n)= \bar{1}$ ($\sigma(m)$ for  $m \neq n$ may be arbitrary)  and the semi-proof schema from Example~\ref{ex. semi-proof-schema}. We had 
$\Pbf = (\rho_0,\{\rho_0,\rho\},\{n\},t,\Seq,\Pi)$ where 
$\rho$ is defined in Example~\ref{ex.proofschema} and 
\begin{eqnarray*}
t(\rho) &=& (n),\\
\Seq(\rho) &=& {\it fdef},\forall x(P(x) \impl P(f(x))) \seq P(c) \impl P(\fhat(c,n)),\\
\rho(n) &\defeq& \{n=0\colon \xi_1(\rho)(n),\ n>0\colon \xi_2(\rho)(n)\},
\end{eqnarray*}
where $\xi_1(\rho)(n) = \chi$ and $\xi_2(\rho)(n) = \varphi^*(\rho,n)$ where $\chi$ is defined in Example~\ref{ex.proofschema} and $\varphi^*(\rho,n)$ in Example~\ref{ex.sproof}. The main proof term $\rho_0$ was defined as 
\begin{eqnarray*}
\Seq(\rho_0) &=& P(c), {\it fdef}, \forall x(P(x) \impl P(f(x))) \seq P(\fhat(c,s(n))),\\
\rho_0(n) &\defeq& \{n=0\colon \xi_1(\rho_0)(n),\ n>0\colon \xi_2(\rho_0)(n)\}
\end{eqnarray*}
We have to compute $\sigma(\rho_0(n))\Eval$ which, by definition, is 
$\xi_2(\rho_0)(\bar{1})\Eval=$
\small
\[
\infer[{\it cut} ]{P(c), {\it fdef}, \forall x(P(x) \impl P(f(x))) \seq P(\fhat(c,\bar{2}))}
 { \infer[(\rho,\theta)]{{\it fdef}, \forall x(P(x) \impl P(f(x)))  \seq P(c) \impl  P(\fhat(c,\bar{2}))}
   { \sigma(\rho(s(n))\Eval
    }
   & 
  \infer[\impl\colon l]{P(c), P(c) \impl  P(\fhat(c,\bar{2})) \seq P(\fhat(c,\bar{2}))}
   { P(c) \seq P(c)
     &
      P(\fhat(c,\bar{2})) \seq P(\fhat(c,\bar{2}))
   }
 }
\]
We have to compute $\sigma(\rho(s(n))$ which is $\sigma'(\rho(n))\Eval$ for $\sigma'(n)$ where $\sigma'(n)=\bar{2}$. By definition of $\rho$ we have 
$$\sigma'(\rho(n))\Eval = \sigma'(\xi_2(\rho)(n))\Eval = \sigma'(\varphi^*(\rho,n))\Eval.$$
we first compute $\sigma'(\varphi(\rho,n))\Eval$ where $\varphi(\rho,n)$ is the core of  $\varphi^*(\rho,n))$.\\[1ex]
$\sigma'(\varphi(\rho,n))\Eval=$
\[
\infer[{\it cut}+c^*]{{\it fdef},\forall x(P(x) \impl P(f(x))) \seq P(c) \impl P(\fhat(c,\bar{2}))}
{ {\it fdef},\forall x(P(x) \impl P(f(x))) \seq P(c) \impl P(\fhat(c,\bar{1})))
  &
 \deduce{\sigma'(S)}{(\sigma'(\psi)\Eval)}
 } 
\]
where 
$$\sigma'(S)  = P(c) \impl P(\fhat(c,\bar{1})), {\it fdef},\forall x(P(x) \impl P(f(x))) \seq P(c) \impl P(\fhat(c,\bar{2}))$$
and $\sigma'(\psi)\Eval=$
\[
\infer[\impl\colon r]{\sigma'(S)\colon P(c) \impl P(\fhat(c,\bar{1})), {\it fdef},\forall x(P(x) \impl P(f(x))) \seq P(c) \impl P(\fhat(c,\bar{2}))}
  { \infer[w\colon l]{P(c), P(c) \impl P(\fhat(c,\bar{1})), {\it fdef},\forall x(P(x) \impl P(f(x))) \seq P(\fhat(c,\bar{2}))}
    { \infer[\impl\colon l]{P(c), P(c) \impl P(\fhat(c,\bar{1})),\forall x \forall z(\fhat(x,s(z)) = f (\fhat(x,z))),\forall x(P(x)     
        \impl P(f(x))) \seq  P(\fhat(c,\bar{2}))}
     { P(c) \seq P(c) 
                     &
         \deduce{\sigma'(S_1)\Eval}{(\sigma'(\psi_1)\Eval)}
      }
      }
   }
\]
where $\sigma'(\psi_1)\Eval=$
\[
\infer[\forall\colon l]{\sigma'(S_1)\colon P(\fhat(c,\bar{1})),\forall x \forall z(\fhat(x,s(z)) = f (\fhat(x,z))),\forall x(P(x) \impl P(f(x))) \seq     
         P(\fhat(c,\bar{2}))}
  { \infer[\impl\colon l]{P(\fhat(c,\bar{1})),\forall x \forall z(\fhat(x,s(z)) = f (\fhat(x,z))), P(\fhat(c,\bar{1})) \impl P(f(\fhat(c,\bar{1})))   
        \seq P(\fhat(c,\bar{2}))}
     {    P(\fhat(c,\bar{1})) \seq P(\fhat(c,\bar{1}))
            &
          \deduce{\sigma'(S_2)}{(\sigma'(\psi_2)\Eval)}
      }
    }
\]
where $\sigma'(\psi_2)\Eval=$
\[
\infer[\forall\colon l^*]{\sigma'(S_2)\colon \forall x,z(\fhat(x,s(z)) = f (\fhat(x,z))),P(f(\fhat(c,\bar{1}))) \seq P(\fhat(c,\bar{2}))}
  { \infer[{\it cut}]{\fhat(c,\bar{2}) = f (\fhat(c,\bar{1})), P(f(\fhat(c,\bar{1}))) \seq P(\fhat(c,\bar{2}))}
     { \fhat(c,\bar{2}) = f (\fhat(c,\bar{1})), P(f(\fhat(c,\bar{1}))) \seq P(\fhat(c,\bar{2}))
         &
        \deduce{\sigma'(S_3)}{(\sigma'(\psi_3))\Eval}
     }
   }
\]
for $\sigma'(\psi_3)\Eval=$
\[
\infer[{\it cut}]{\sigma'(S_3)\colon P(\fhat(c,\bar{2})) \seq P(\fhat(c,\bar{2}))}
   { \seq \bar{2} = \bar{2}
       &
      \bar{2} = \bar{2} , P(\fhat(c,\bar{2})) \seq P(\fhat(c,\bar{2}))
   }
\]
Note that $\sigma'(\psi_3)\Eval$ is indeed a redundant proof. The reason is that $\seq s(p(n)) = n$ normalizes under $\sigma'$ to $\seq \bar{2} = \bar{2}$ - which is a valid sequent in first-order logic. Obviously, 
$\sigma'(\varphi(\rho,n))\Eval$ is an $\LK$-proof. \\[1ex]
Remember that $\varphi^*(\rho,n)=$
\[
\infer[{\it cut}]{{\it fdef},\forall x(P(x) \impl P(f(x))) \seq P(c) \impl P(\fhat(c,n))}
{ \infer[\rho\{n \ass p(n)\}]{{\it fdef},\forall x(P(x) \impl P(f(x))) \seq P(c) \impl P(\fhat(c,p(n)))}
   {\rho(p(n))
   }
  &
 \psi
 } 
\]
Therefore $\sigma'(\varphi^*(\rho,n))\Eval=$
\[
\infer[{\it cut}]{{\it fdef},\forall x(P(x) \impl P(f(x))) \seq P(c) \impl P(\fhat(c,\bar{2}))}
{ \sigma(\rho(\bar{n}))\Eval
  &
 \sigma'(\psi)\Eval
 } 
\]
For (the original $\sigma$) $\sigma(n)=\bar{1}$. Now we compute $\rho(\bar{1})\Eval$ which is a similar computation as above and again yields an $\LK$-proof. We leave the details to the reader. Finally, we have to compute $\sigma(\rho_0(n))\Eval$. 
$\sigma(\rho_0(n))\Eval = \sigma(\xi_2(\rho_0)(n))\Eval=$
\small
\[
\infer[{\it cut} ]{P(c), {\it fdef}, \forall x(P(x) \impl P(f(x))) \seq P(\fhat(c,\bar{2}))}
 { \infer[(\rho,\theta)]{{\it fdef}, \forall x(P(x) \impl P(f(x)))  \seq P(c) \impl  P(\fhat(c,\bar{2}))}
   { \sigma'(\varphi^*(\rho,n))\Eval
    }
   & 
  \infer[\impl\colon l]{P(c), P(c) \impl  P(\fhat(c,\bar{2})) \seq P(\fhat(c,\bar{2}))}
   { P(c) \seq P(c)
     &
      P(\fhat(c,\bar{2})) \seq P(\fhat(c,\bar{2}))
   }
 }
\]
The obtained proof is an $\LK$-proof.
\end{example}
Due to nontermination, a semi-proof schema does not always evaluate to an {\LK}-proof under parameter assignments. We now define sufficient conditions on a semi-proof schema to terminate under all parameter assignments. The key is to extract a point transition system from the semi-proof schema which is canonical (i.e. terminating under every parameter assignment); then we will show that, from the termination of the point transition system that of the semi-proof schema follows. 
\begin{definition}[extraction of  a point transition system]\label{def.extract-pts}
Let $\Pbf\colon (\rho_0,R,\Ncal_0,t,\Seq,\Pi)$ be a semi-proof schema where, by definition, $\Pi$ is a mapping assigning to every proof term $\rho(m_1,\ldots,m_l)$ in $\Pbf$ a definition $D(\rho)$ of the form   
\begin{eqnarray*}
\rho(m_1,\ldots,m_l) &\defeq& \{C_1\colon \xi_1(\rho)(m_1,\ldots,m_l), \ldots, C_k\colon \xi_k(\rho)(m_1,\ldots,m_l)\}
\end{eqnarray*}
where $(\C_1,\ldots,C_k)$ is a partition and, for $\xi \in \{\xi_1,\ldots,\xi_k\}$, the $\xi(\rho)(m_1,\ldots,m_l)$ are $s$-proofs corresponding to $\rho$. Note  that the leaves of $\xi(\rho)(m_1,\ldots,m_l)$ which are not axioms are instances of proof terms in $\Pbf$. \\[1ex]
To $\Pbf$ we assign a point transition system $p(\Pbf)\colon (\rho_0,R \union \Delta_e,\Delta_e, p(\Pi))$ where 
\begin{itemize}
\item $\Delta_e = \{\delta_e\}$ where $\delta_e$ is not a proof symbol in $R$,
\item the sources of $\rho$ in $R$ are defined by $t(\rho)$.
\item Now let $D(\rho)$ as defined above and $p=(m_1,\ldots,m_l)$. We define 
$$p(\Pbf)(\rho) = \{(\rho,p) \to \Lcal_1\colon C_1,\ldots, (\rho,p) \to \Lcal_k\colon C_k\}$$
where the $\Lcal_j$ for $j=1,\ldots,k$ are defined as follows:
\item If $\xi_j(\rho)$ is an $\LKN$-proof (i.e. no leaves are proof terms) then $\Lcal_j = \{(\delta_e,p)\}$.
\item Assume that $\xi_j(\rho)$ is an $s$-proof which is not an $\LKN$-proof and let 
$$\rho_1(t^1_1,\ldots,t^1_{r(1)}), \ldots, \rho_\alpha(t^\alpha_1,\ldots,t^\alpha_{r(\alpha)})$$
be the proof terms appearing at the non-axiomatic leaves of $\xi_j(\rho)$. Then 
$$\Lcal_j = \{(\rho_1,(t^1_1,\ldots,t^1_{r(1)})), \ldots, (\rho_\alpha,(t^\alpha_1,\ldots,t^\alpha_{r(\alpha)})).$$
\end{itemize}
By definition $p(\Pbf)$ is a point transition system.
\end{definition} 
The point transition system $p(\Pbf)$ abstracts from the $\LKN$-proofs appearing in $\Pbf$ and only preserves the call-structure of 
$\Pbf$. If all proofs appearing in $\Pbf$ are $\LKN$-proofs then all point transitions are of the form  $(\rho,p) \to \{(\delta_e,p)\}$. 
\begin{example}\label{ex.schema-pts}
Let us consider the semi-proof schema $\Pbf\colon (\rho_0,\{\rho_0,\rho\},\{n\},t,\Seq,\Pi)$ from 
Example~\ref{ex. semi-proof-schema}. We had 
\begin{eqnarray*}
t(\rho) &=& (n),\\
\Seq(\rho) &=& {\it fdef},\forall x(P(x) \impl P(f(x))) \seq P(c) \impl P(\fhat(c,n)),\\
\Pi(\rho(n)) &=& \rho(n) \defeq \{n=0\colon \xi_1(\rho)(n),\ n>0\colon \xi_2(\rho)(n)\},
\end{eqnarray*}
where $\xi_1(\rho)(n) = \chi$ and $\xi_2(\rho)(n) = \varphi^*(\rho,n)$ where $\chi$ is defined in Example~\ref{ex.proofschema} and 
$\varphi^*(\rho)(n)$ in Example~\ref{ex.sproof}. For $\rho_0$ we had 
\begin{eqnarray*}
t(\rho_0) &=& (n),\\
\Seq(\rho_0) &=& P(c), {\it fdef}, \forall x(P(x) \impl P(f(x))) \seq P(\fhat(c,s(n))),\\
\Pi(\rho_0(n)) &=& \rho_0(n) \defeq \{n=0\colon \xi_1(\rho_0)(n),\ n>0\colon \xi_2(\rho_0)(n)\}
\end{eqnarray*}
where $\xi_1(\rho_0)(n)$ and $\xi_2(\rho_0)(n)$ are defined in Example~\ref{ex. semi-proof-schema}. The corresponding point transition system is 
$$p(\Pbf)\colon  (\rho_0,\{\rho_0,\rho,\delta_e\}, \{\delta_e\},p(\Pi))$$
where
\begin{eqnarray*}
p(\Pi)(\rho) &=& \{(\rho,n) \to \{(\delta_e,n)\}\colon n=0,\ (\rho,n) \to \{(\rho,p(n))\}\colon n>0\},\\
p(\Pi)(\rho_0) &=& \{(\rho_0,n) \to \{(\delta_e,n)\}\colon n=0,\ (\rho_0,n) \to \{(\rho,s(n))\}\colon n>0\}.
\end{eqnarray*}
We show that $p(\Pbf)$ is canonic: we have $\rho \ltPbf \rho_0$ as $\rho_0 \redPbfs \rho$ and 
$\rho \not \redPbfs \rho_0$, and $\rho \redPbf \rho$. Let us consider $<_l$, the lexicographic ordering of tuples (which on $1$-tuples reduces to the standard ordering on numerals) and let $C \equiv n>0$; then for 
\[
\begin{array}{l}
(\rho,n) \to \{(\rho,p(n))\}\colon n>0 \mbox{ we have } (\rho,p(n)) <_C (\rho,n) \mbox{ via }<_l,\\
\mbox{ and for } (\rho_0,n) \to \{(\rho,s(n))\}\colon n>0, (\rho,s(n)) <_C  (\rho_0,n) \mbox{ via } \rho \ltPbf \rho_0.
\end{array}
\]
\end{example}
\begin{definition}[proof schema]\label{def.proof-schema}
A semi-proof schema $\Pbf$ is called a {\em proof schema} if $p(\Pbf)$ is a canonic point transition system. 
\end{definition}
\begin{theorem}\label{the.proof-schema}
Let $\Pbf\colon (\rho_0,R,\Ncal_0,t,\Seq,\Pi)$ be a proof schema such that $t(\rho_0)= (n_1,\ldots,n_k)$. Then, for all parameter assignments $\sigma$, $\sigma(\rho_0(n_1,\ldots,n_k))\Eval$ is an $\LK$-proof.
\end{theorem}
\begin{proof}
Let $\sigma$ be an arbitrary parameter assignment. As $\Pbf$ is a proof schema the corresponding point transition system $p(\Pbf)$ is canonic. I.e. $p(\Pbf)$ is terminating under $\sigma$ w.r.t. an ordering on the proof symbols and on the tuples. If the computation $\sigma(\rho_0(n_1,\ldots,n_k))\Eval$ terminates then, by Definition~\ref{def.sem-proofterms}, 
$\sigma(\rho_0(n_1,\ldots,n_k))\Eval$ is an $\LK$-proof. Thus the only possibility of getting no $\LK$-proof is nontermination. The difference between the proof schema and the corresponding point transition system consists only in the kernels of the proofs $\xi_j$ which are missing in the point transition system. As nontermination cannot occur in the kernels  (these are $\LK$-proofs under $\sigma$)  we must obtain an infinite sequence of proof terms which are defined in the evaluation of $\rho_0(n_1,\ldots,n_k)$. This infinite sequence, in turn, defines an infinite sequence in $p(\Pbf)$ under 
$\sigma$ - contradicting the canonicity of $p(\Pbf)$.
\end{proof}

\section{Herbrand Systems}\label{sec.herbrandsystems}

In its simplest form, Herbrand's theorem \cite{herbrand1930recherches} states that a formula $\exists x A(x)$ is valid only if there exist terms $t_1, \ldots, t_n$ such that the Herbrand disjunction $A(t_1) \lor \ldots \lor A(t_n)$ is valid.
More generally, proof-theoretic analyses of first-order proofs extract finite collections of instances of quantified formulas in the end-sequent from which validity can be reconstructed.

In the presence of equality, the appropriate notion is that of an equational Herbrand sequent, a sequent obtained from the quantified formulas of the end-sequent by instantiation, which is valid in first-order logic with equality. In contrast to ordinary propositional Herbrand sequents, equational Herbrand sequents are inherently first-order.
\begin{definition}[Equational Herbrand sequent]
Let $S = \Gamma \vdash \Delta$ be a skolemized prenex sequent, where the formulas in $\Gamma$ and $\Delta$
may contain quantified subformulas
\[
Qx_1 \dots Qx_n\, A(x_1,\dots,x_n)
\]
with $Q \in \{\forall,\exists\}$ and $A$ quantifier-free. Moreover, let $S$ be provable from a set of equality axioms $\mathcal{A}_E$. An equational Herbrand sequent of $S$ is a sequent obtained by replacing each quantified formula
\[
Qx_1 \dots Qx_n\, A(x_1,\dots,x_n)
\]
by a finite set of quantifier-free instances
\[
A(t_1^1,\dots,t_n^1), \dots, A(t_1^k,\dots,t_n^k)
\]
such that the resulting quantifier-free sequent is equationally valid under $\mathcal{A}_E$. This means that the equational Herbrand sequent is valid in first-order logic with equality, and therefore valid in any interpretation of the equality predicate as equality over some domain.
\end{definition}
All Herbrand sequents considered in this paper are understood modulo equality.
Herbrand sequents can be extracted directly from proofs without quantified cuts of prenex end-sequents \cite{hetzl2008herbrand}. 
Moreover, it is possible to extract their schematic version, so-called Herbrand systems, from proof schemata \cite{LCL.2026}.

$\LK$-schemata are based on s-proofs, therefore, in this section, we will restrict s-proofs to derivations without quantified cuts of skolemized, prenex end-sequents. Moreover, s-proofs may contain initial sequents that correspond to (recursive) calls and may not be axioms. By admitting non-axiomatic initial sequents in s-proofs, we do not obtain full but merely partial Herbrand substitutions for quantified formulas in the end-sequents of s-proofs.
\begin{definition}[thread, cf \cite{Takeuti}]
Let $\varphi$ be an s-proof. A thread in $\varphi$ is a path of sequent occurrences in the proof tree beginning at the end-sequent and ending in an initial sequent.
\end{definition}
\begin{definition}[trace, cf \cite{Takeuti}] 
Let $\varphi$ be an s-proof and $\theta \colon \nu_0, \ldots, \nu_{\alpha}$ a thread in $\varphi$. 
Let $\mu_0$ be a formula occurrence in the sequent at occurrence $\nu_0$, denoted as $\Seq(\nu_0)$. 
For every $i \in \{0, \ldots, \alpha-1\}$ let $\mu_{i+1}$ be an occurrence in $\Seq(\nu_{i+1})$, which is an ancestor occurrence of $\mu_i$ in $Seq(\nu_i)$. Then the sequence 
$(\nu_0, \mu_0), \ldots, (\nu_{\alpha}, \mu_{\alpha})$ 
is called a trace of $\mu_0$ in $\theta$. 
A sequence $\mu_0, \ldots, \mu_{\alpha}$ is called a trace of $\mu_0$ in $\varphi$ if there exists a thread $\theta$ in $\varphi$ such that $(\nu_0, \mu_0), \ldots, (\nu_{\alpha}, \mu_{\alpha})$ is a trace of $\mu_0$ in $\theta$.
\end{definition}
\begin{example} \label{ex.trace}
Consider the proof schema from Example \ref{ex.proofschema}, where
$$\rho(n) \defeq \{n=0\colon \chi,\ n>0\colon \varphi^*(\rho,n)\},$$
and the derivation $\chi$=
\[
\infer[w^*]{\nu_0 \colon {\it fdef}, \forall x(P(x) \impl P(f(x))) \seq P(c) \impl P(\fhat(c,\bar{0}))}
 { \infer[\forall\colon l]{\nu_1 \colon \forall x.\fhat(x,\bar{0})=x \seq P(c) \impl P(\fhat(c,\bar{0}))}
    { \infer[\impl\colon r]{\nu_2 \colon \fhat(c,\bar{0}) = c \seq P(c) \impl P(\fhat(c,\bar{0}))}
     { \nu_3 \colon P(c), \fhat(c,\bar{0}) = c \seq  P(\fhat(c,\bar{0}))
      }
    }
  }
\]
Then $\theta \colon \nu_0,\nu_1,\nu_2,\nu_3$ is a thread in $\chi$.
Let $\mu_0$ be the occurrence of the formula $\forall x. \fhat(x, \overline{0}) = x$ (in ${\it fdef}$) in the end-sequent. Then $\tau \colon \mu_0, \mu_1, \mu_2, \mu_3$ is a trace of $\mu_0$ in $\chi$, where $\mu_1$ is the occurrence of $\forall x. \fhat(x, \overline{0}) = x$ in the sequent occurrence $\nu_1$, $\mu_2$ the occurrence of $\fhat(c,\overline{0})=c$ in the sequent occurrence $\nu_2$, and $\mu_3$ the occurrence of $\fhat(c,\overline{0})=c$ in the initial sequent.
\end{example}
\begin{definition}[yield] 
Let $\tau \colon \mu_0, \ldots, \mu_{\alpha}$ be a trace in an s-proof $\varphi$. Then the sequence of formulas 
$formula(\mu_0), \ldots , formula(\mu_{\alpha})$
is called the yield of $\tau$. 
\end{definition}
\begin{example}
Consider the trace $\tau \colon \mu_0, \mu_1, \mu_2, \mu_3$ of $\mu_0$ in $\chi$ from Example \ref{ex.trace}. Then 
$$\forall x. \fhat(x, \overline{0}) = x, \ \forall x. \fhat(x, \overline{0}) = x, \ \fhat(c,\overline{0})=c, \ \fhat(c,\overline{0})=c$$
is a yield of $\tau$.
\end{example}
Note that parameter substitutions can be applied to yields of traces. To every trace $\tau \colon \mu_0, \ldots, \mu_{\alpha}$ we can attach a parameter substitution $\theta$ and obtain $\mu_0 \theta, \ldots, \mu_{\alpha}\theta$. As the $\mu_i$ denote formula occurrences, no parameter substitution can be applied here. However, we can apply the parameter substitution to its yield. Let $formula(\mu_0), \ldots , formula(\mu_{\alpha})$ be the yield of $\tau$, and for each $formula(\mu_i)$ let $\vec{n_i}$ be the parameters in $formula(\mu_i)$. Then $formula(\mu_0) \theta, \ldots , formula(\mu_{\alpha}) \theta$ denotes the sequence of formulas $formula(\mu_i)$ where each $\vec{n_i}$ is replaced by $\vec{n_i} \theta$.
\begin{proposition}
Let $\varphi$ be an s-proof without quantified cuts of a skolemized, prenex end-sequent and $\mu_0$ an occurrence of
$$Q x_1 \ldots Q x_{\beta}. A(x_1, \ldots, x_{\beta})$$
in the end-sequent, where $Q \in \{\exists, \forall\}$ and $A$ is quantifier-free. 
Let $\tau$ be a trace of $\mu_0$ in $\varphi$ which ends in an axiom. Then there is a substitution $\sigma$ such that $A \sigma$ occurs in the yield of $\tau$. 
\end{proposition}
\begin{proof}
By induction on the number of quantifiers in $Q x_1 \ldots Q x_{\beta}. A(x_1, \ldots, x_{\beta})$.

Base case: There is one quantifier, therefore the formula is of the form $Q x A$. We follow the trace $\tau$ of the occurrence $\mu_0$ of $Q x A$ in the derivation. By the subformula property, and axioms being atomic, the quantified variable $x$ in $Q x A$ will be replaced by a term $t$ and thus, $\{x \leftarrow t \}$ is the desired substitution $\sigma$.

IH: For $n$ quantifiers in the formula $Q x_1 \ldots Q x_{n}. A(x_1, \ldots, x_{n})$ there exists a substitution $\sigma$ such that $A \sigma$ occurs in the yield of $\tau$.

Now consider the case of $n+1$ quantifiers in $Q x_1 \ldots Q x_{n+1}. A(x_1, \ldots , x_{n+1})$. We follow the trace $\tau$ of the occurrence of $\mu_0$ of $Q x_1 \ldots Q x_{n+1}. A(x_1, \ldots , x_{n+1})$ in the derivation. By the subformula property, and axioms being atomic, the first quantified variable that will be eliminated is $Q x_1$. Therefore, $x_1$ will be replaced by a term $t_1$ and a substitution $\sigma_1 = \{x_1 \leftarrow t_1\}$ is obtained. Thus, $Q x_1 \ldots Q x_{n+1}. A(x_1, \ldots , x_{n+1})$ will change to $A' \colon Q x_2 \ldots Q x_{n+1}. A\sigma_1$. In $A'$ there are only $n$ quantifiers left and we apply the induction hypothesis, i.e. there is a substitution $\sigma$ such that the formula $A\sigma_1 \sigma$ occurs in the yield of $\tau$. $\sigma_1 \sigma$ is the desired substitution.
\end{proof}
As $s$-proofs might have non-axiomatic initial sequents, we need a machinery to trace the substitutions in the different derivations.
\begin{definition} \label{def.hi1}
Let $\varphi$ be an s-proof without quantified cuts of a skolemized, prenex end-sequent $S$ and let $\mu_0$ be an occurrence of a formula
$$F \colon Qx_1 \ldots Qx_{\beta} F'(x_1, \ldots , x_{\beta})$$
in $S$ such that $Q \in \{\forall, \exists\}$ and $F'$ is quantifier-free. 
Let $\tau \colon \mu_0, \ldots, \mu_{\alpha}$ be a trace of $\mu_0$ in $\varphi$. We distinguish two cases: \\[1ex]
[A] $\mu_{\alpha}$ occurs in an axiom: Then there exist terms $t_1, \ldots, t_{\beta}$ such that formula$(\mu_{\alpha})$ is a subformula of $F'(t_1, \ldots, t_{\beta})$ and $F'(t_1, \ldots, t_{\beta})$ occurs in the yield of $\tau$. \\[1ex]
[B] $\mu_{\alpha}$ does not occur in an axiom: Then the sequent at node $\nu_{\alpha}$ is of the form 
\[
\infer[\rho'\theta]{S'}{\nu_{\alpha} \colon \rho'(m_1,\ldots,m_l)\theta}
\]
Then either
\begin{enumerate}
	\item formula$(\mu_{\alpha-1}) = F$, or
	\item there exists an $i \in \{1, \ldots, \beta-1\}$ and terms $t_1, \ldots, t_i$ such that 
	$$ \mbox{formula}(\mu_{\alpha-1}) = Qx_{i+1} \ldots Qx_{\beta} F'(t_1, \ldots, t_i, x_{i+1}, \ldots, x_{\beta}),$$
	or
	\item there exist terms $t_1, \ldots, t_{\beta}$ such that formula$(\mu_{\alpha-1})$ is a subformula of $F'(t_1, \ldots, t_{\beta})$ and $F'(t_1, \ldots, t_{\beta})$ occurs in the yield of $\tau$.
\end{enumerate}
\end{definition}
With the definition above, we can assign a substitution $hi(\tau)$ to each trace $\tau$ in an s-proof $\varphi$ of a skolemized prenex end-sequent, such that either $hi(\tau)$ is a Herbrand substitution or just a partial one.
\begin{definition} \label{def.hi2}
Let $\pi$ be an s-proof without quantified cuts of a skolemized, prenex end-sequent $S$ and let $\mu_0$ be an occurrence of a formula
$$F \colon Qx_1 \ldots Qx_{\beta} F'(x_1, \ldots , x_{\beta})$$
in $S$ such that $Q \in \{\forall, \exists\}$ and $F'$ is quantifier-free. 
Let $\tau \colon \mu_0, \ldots, \mu_{\alpha}$ be a trace of $\mu_0$ in $\pi$. For [A] in Definition \ref{def.hi1}, $hi(\tau) = \{x_1 \leftarrow t_1, \ldots, x_{\beta} \leftarrow t_{\beta}\}$. For [B], we consider the cases $1.$, $2.$, and $3.$:
\begin{itemize}
	\item case $1.$: $hi(\tau) = \emptyset$,
	\item case $2.$: $hi(\tau) = \{x_1 \leftarrow t_1, \ldots, x_i \leftarrow t_i\}$,
	\item case $3.$: $hi(\tau) = \{x_1 \leftarrow t_1, \ldots, x_{\beta} \leftarrow t_{\beta}\}$.
\end{itemize}
\end{definition}
Before we go on with the formal definitions, let us have a look at an example where we compute a set of (schematic) Herbrand substitutions for a quantified formula.
\begin{example}\label{ex.HerbrandSchemaSimple}
Consider the proof schema from Example \ref{ex.proofschema}, where
$$\rho(n) \defeq \{n=0\colon \chi,\ n>0\colon \varphi^*(\rho,n)\}.$$
The Herbrand system of $\rho(n)$ will be defined for each occurrence $\mu_0$ of a quantified formula in the end-sequent of $\rho(n)$. In this example, we will compute the Herbrand system for $\mu_0$, the occurrence of the quantified formula 
$$\forall x (P(x) \to P(f(x)))$$
in the end-sequent. Informally, the Herbrand system for $\mu_0$ is defined over the same partition as the proof schema:
$$\Theta^*(\rho, \mu_0, \lambda, n) \defeq \{n=0\colon \Theta(\chi, \mu_0, \lambda, n),\ n>0\colon \Theta(\varphi^*, \mu_0, \lambda, n)\},$$
where $\Theta(\chi, \mu_0, \lambda, n)$ and $\Theta(\varphi^*, \mu_0, \lambda,n)$ are the Herbrand systems for $\mu_0$ in $\chi$ and in $\varphi^*$, respectively. The substitution $\lambda$ can be ignored for now, but will become important later on (assume $\lambda = id$, the identical parameter substitution, for now).
\begin{enumerate}
\item We start with the partition $n = 0$ and $\Theta(\chi, \mu_0, \lambda,n)$. In $\chi$ there is no trace for $\mu_0$, because the formula $\forall x (P(x) \to P(f(x)))$ disappears by weakening, so $\Theta(\chi, \mu_0, \lambda,n) = \emptyset$.

\item For $n > 0$ we compute $\Theta(\varphi^*, \mu_0, \lambda, n)$. Here we see that there are two traces for $\mu_0$ in $\varphi^*$. The first trace $\tau_1$ proceeds on the left branch of the derivation and ends in the initial sequent $\rho(p(n))$ after an application of $\rho \theta$, for $\theta = \{n \leftarrow p(n)\}$.
The second trace $\tau_2$ proceeds on the right branch of the derivation ending in the sub-derivation $\psi_2$.
In the final Herbrand system for $\mu_0$ we have to take into account all substitution instances from all traces, thus, $\Theta(\varphi^*, \mu_0, \lambda,n)$ is defined as 
$$\Theta(\varphi^*, \tau_1, \lambda,n) \cup \Theta(\varphi^*, \tau_2, \lambda, n).$$ 
In both traces, Herbrand instances for $\mu_0$ are defined. It is easy to see that in $\tau_2$, $x$ is substituted with $\fhat(c,p(n))$, and we obtain a Herbrand instance for the formula $\forall x (P(x) \to P(f(x)))$ given by the substitution
$$\Theta(\varphi^*, \tau_2, \lambda,n) = \{x \leftarrow \fhat(c, p(n))\}\lambda.$$
The application of $\lambda$ to the substitution above can be ignored for this example, as $\lambda = id$. However, as we will explain later on, $\lambda$ might not always be $id$.

In $\tau_1$ we enter a recursion as $\varphi^*$, and hence also $\rho(n)$, is defined over $\rho(p(n))$ by a substitution $\theta = \{n \leftarrow p(n)\}$. Therefore, the Herbrand system for $\mu_0$ in $\rho(n)$ for $n>0$ will also be defined over the Herbrand system for $\mu_0$ in $\rho(p(n))$:
$$\Theta(\varphi^*, \tau_1, \lambda,n) = \Theta^*(\rho, \mu_0, \lambda \theta, p(n)).$$
Note that here, $\theta$ is indeed a substitution and not $id$ as for $\mu_0$ in the beginning.
So, 
$$\Theta(\varphi^*, \mu_0, \lambda,n) = \Theta^*(\rho, \mu_0, \lambda \theta, p(n)) \cup \{x \leftarrow \fhat(c, p(n))\}\lambda.$$
The proof schema from Example \ref{ex.proofschema} was used in Example \ref{ex. semi-proof-schema} to define a more complex proof schema. There, an initial sequent of the form $\rho(s(n))$ occurred after an application of $\rho \{n \leftarrow s(n)\}$. Therefore, the Herbrand system of the proof schema in Example \ref{ex. semi-proof-schema} will be defined over the Herbrand system we computed in this example. And now also the role of $\lambda$ (which was $id$ in this example) becomes obvious. Indeed, we will have to set $\lambda = \{n \leftarrow s(n)\}$ when using the Herbrand system computed in this example to compute the Herbrand system of Example \ref{ex. semi-proof-schema}.
\end{enumerate}
For some $n = \alpha$, we therefore obtain the set of Herbrand instances 
%
$$\fhat(c, \alpha), \ldots , \fhat(c, 0).$$
\end{example}
As indicated in the example above, to compute a Herbrand system for a formula occurrence $\mu_0$ we have to consider all traces of that formula in a proof schema. 
In $\LK$-schemata, s-proofs might define recursions by initial sequents linking to proof terms $\rho_i$. When we define the trace of a formula occurrence in an end-sequent of an s-proof, we need a machinery to track all traces of that formula that go into $\rho_i$. 
\begin{definition}
Let $\rho$ be an s-proof and $\mu_0$ a formula occurrence in its end-sequent. The set of all traces $T(\rho, \mu_0, \theta)$ for $\mu_0$, where $\theta$ is a parameter substitution, is defined as the union of all traces for $\mu_0$ in $\rho$, and, if the corresponding thread ends in an initial sequent of the form
\[
\infer[\rho'\theta']{S'}{\nu_{\alpha} \colon \rho'(m_1,\ldots,m_l)\theta'}
\]
the trace at node $\nu_{\alpha}$ is set to $T(\rho', \mu_0', \theta'\theta)$, where $\mu_0'$ is the corresponding ancestor occurrence of $\mu_0$ in $\rho'$.
\end{definition}
\begin{definition}[Herbrand system] \label{def.Herbrandschema}
Let $\Pbf \colon (\rho_0,R,\Ncal_0,t,\Seq,\Pi)$ be a proof schema, $\rho \in R$ such that the corresponding $s$-proof, of the form $\xi_i(\rho)(m_1,\ldots,m_l)$ for some condition $C_i$ and some parameter list $(m_1,\ldots,m_l)$, has quantifier-free cuts and is skolemized and prenex. Let $\mu_0$ be an occurrence of a formula
$$F \colon Qx_1 \ldots Qx_{\beta} F'(x_1, \ldots, x_{\beta})$$
in $\Seq(\rho)$ such that $Q \in \{\forall, \exists\}$ and $\tau \colon \mu_0, \ldots, \mu_{\alpha}$ a trace of $\mu_0$ in $\xi_i(\rho)$. \\[1ex]
We consider $\xi_i(\rho)(m_1, \ldots, m_l)$ under the condition $C$:

\begin{enumerate}
	\item formula$(\mu_{\alpha})$ is quantifier-free ($\lambda$ a parameter substitution). Then $\Theta(\xi_i(\rho),\tau,\lambda, m_1,\ldots,m_l) = \{hi(\tau)\lambda\}$. 
	
	\item $\mu_{\alpha}$ occurs in an initial sequent of the form $\rho'\theta[\tau]$ and formula$(\mu_{\alpha-1}) = F$. Then 
	$$\Theta(\xi_i(\rho), \tau,\lambda, m_1,\ldots,m_l) = \Theta^*(\rho',\mu_0(\rho'),\lambda\theta[\tau], m_1,\ldots,m_l),$$
	where $\mu_0(\rho')$ is the occurrence of $F$ in $\rho'$.
	
	\item $\mu_{\alpha}$ occurs in an initial sequent of the form $\rho'\theta[\tau]$ and formula$(\mu_{\alpha-1}) =$ $Qx_{i+1} \ldots$ $Qx_{\beta}$ $F'(t_1, \ldots ,$ $t_i, x_{i+1}, \ldots, x_{\beta})$. Then 
$$\Theta(\xi_i(\rho),\tau,\lambda,m_1,\ldots,m_l) = \{\{x_1 \leftarrow t_1, \ldots, x_i \leftarrow t_i\}\} \Theta^*(\rho',\mu_0(\rho'),\lambda\theta[\tau],m_1,\ldots,m_l).$$
\end{enumerate}
We define
$$\Theta(\xi_i(\rho), \mu_0,\lambda,m_1,\ldots,m_l) = \bigcup \{\Theta(\xi_i(\rho), \tau,\lambda, m_1,\ldots,m_l) \ | \ \tau \in T(\rho, \mu_0,\lambda)\}.$$
Let $D(\rho)$ be given as 
\begin{eqnarray*}
\rho(m_1,\ldots,m_l) &\defeq& \{C_1\colon \xi_1(\rho)(m_1,\ldots,m_l), \ldots, C_k\colon \xi_k(\rho)(m_1,\ldots,m_l)\}
\end{eqnarray*}
Then the Herbrand system of $\rho$ for $\mu_0$ can be defined as
\begin{eqnarray*}
\Theta^*(\rho,\mu_0,\lambda,m_1,\ldots,m_l) &\defeq& \{C_1(\rho)\colon \Theta(\xi_1(\rho),\mu_0,\lambda,m_1,\ldots,m_l), \ldots, \\
 & & \mbox{ } \  C_k(\rho)\colon \Theta(\xi_k(\rho),\mu_0,\lambda,m_1,\ldots,m_l)\},
\end{eqnarray*}
Finally, the Herbrand system for the proof schema $\Pbf$ is defined as $D(\Theta^*)=$
\begin{eqnarray*}
\Theta^*(\rho_0,\mu_0,id,m_1,\ldots,m_l) &\defeq& \{C_1(\rho_0)\colon \Theta(\xi_1(\rho_0),\mu_0,id,m_1,\ldots,m_l), \ldots, \\
 & & \mbox{ } \  C_k(\rho_0)\colon \Theta(\xi_k(\rho_0),\mu_0,id,m_1,\ldots,m_l)\},
\end{eqnarray*}
where $id$ is the identical parameter substitution on $\{m_1, \ldots, m_l\}$.

\end{definition}
\begin{example} \label{ex.Herbrandschemapartial}
Consider the proof schema $(\rho_0,\{\rho_0,\rho\},\{n\},t,\Seq,\Pi)$ from Example \ref{ex. semi-proof-schema}, where
$$\Pi(\rho_0(n)) = \rho_0(n) \defeq \{n=0\colon \xi_1(\rho_0)(n),\ n>0\colon \xi_2(\rho_0)(n)\}.$$
Let $\mu_0$ be the occurrence of the formula $\forall x (P(x) \to P(f(x)))$ in the end-sequent. We define the Herbrand system for $\mu_0$ as
$$\Theta^*(\rho_0, \mu_0, id, n) \defeq \{n=0\colon \Theta(\xi_1(\rho_0), \mu_0, id,n),\ n>0\colon \Theta(\xi_2(\rho_0), \mu_0, id, n)\}.$$
Let us start with $\Theta(\xi_1(\rho_0), \mu_0, id,n)$. In $\xi_1(\rho_0)(n)$ there is only one trace $\tau$ for $\mu_0$ ending in an initial sequent with formula$(\mu_{\alpha})$ quantifier-free (case 1. in the definition above), so we define the set of Herbrand substitutions 
$$\Theta(\xi_1(\rho_0), \tau, id,n) = hi(\tau) = \{x \leftarrow c\}.$$
Let us proceed with $\Theta(\xi_2(\rho_0), \mu_0, id, n)$. In $\xi_2(\rho_0)(n)$ there is again only one trace $\tau$ for $\mu_0$, which ends in the initial sequent $\rho(s(n))$ after the application of $\rho \theta$ for $\theta = \{n \leftarrow s(n)\}$. We are in case 2. in the definition above, and obtain
$$\Theta(\xi_2(\rho_0), \tau, id, n) = \Theta^*(\rho, \mu_0(\rho), \theta, n),$$
where $\Theta^*(\rho, \mu_0(\rho), \theta, n)$ is as in Example \ref{ex.HerbrandSchemaSimple}:
$$\Theta^*(\rho, \mu_0(\rho), \theta, n) \defeq \{n=0\colon \emptyset,\ n>0\colon \Theta(\varphi^*, \mu_0(\rho), \theta, n)\},$$
where 
$\Theta(\varphi^*, \mu_0(\rho), \theta, n) = \Theta^*(\rho, \mu_0(\rho), \theta\theta', n) \cup \{ \{x \leftarrow \fhat(c, p(n))\}\theta \}$, and $\theta' = \{n \leftarrow p(n)\}.$
Therefore, 
$$\Theta(\varphi^*, \mu_0(\rho), \theta, n) = \Theta^*(\rho, \mu_0, \theta\theta', n) \cup \{ \{x \leftarrow \fhat(c, n)\}\}.$$
It is easy to see that in this example, for some numeral $\alpha$, we obtain the sequence $\fhat(c, \alpha), \ldots, c$. In general, Herbrand systems can be evaluated for any given parameter assignment.
\end{example}
\begin{definition} \label{def.semantics}
Given a Herbrand system $\Theta^*(\rho,\mu,\lambda,m_1, \ldots, m_l)$ and a parameter assignment $\sigma$, we can evaluate the Herbrand system for any $\rho$ and $\mu$, denoted as $\sigma(\Theta^*(\rho,\mu,\lambda,m_1, \ldots, m_l))\Eval$.
Let 
\begin{eqnarray*}
\Theta^*(\rho,\mu_0,\lambda,m_1,\ldots,m_l) &\defeq& \{C_1(\rho)\colon \Theta(\xi_1(\rho),\mu_0,\lambda,m_1,\ldots,m_l), \ldots, \\
 & & \mbox{ } \  C_k(\rho)\colon \Theta(\xi_k(\rho),\mu_0,\lambda,m_1,\ldots,m_l)\},
\end{eqnarray*}
$\sigma'(m_i) = \sigma(\lambda(m_i)) = r_i$ for numerals $r_1 , \ldots, r_l$. Then there exists exactly one $i$ such that 
$\sigma'(C_i(\rho)) = \top$ for $\sigma'(m_i)=r_i$, $i = 1,\ldots,l$, and
$$\sigma(\Theta^*(\rho,\mu_0,\lambda,m_1,\ldots,m_l))\Eval \ = \sigma'(\Theta(\xi_i(\rho),\mu_0,id,m_1,\ldots,m_l))\Eval.$$
We have
$$\sigma'(\Theta(\xi_i(\rho), \mu_0,id,m_1,\ldots,m_l))\Eval \ = \bigcup \{\sigma'(\Theta(\xi_i(\rho), \tau,id, m_1,\ldots,m_l))\Eval \ | \ \tau \in T(\rho, \mu_0,id)\}.$$
We evaluate $\sigma'(\Theta(\xi_i(\rho), \tau,id, m_1,\ldots,m_l))\Eval$ and consider the cases in Definition \ref{def.Herbrandschema}:
\begin{enumerate}
	\item 
	$$\sigma'(\Theta(\xi_i(\rho),\tau,id,m_1,\ldots,m_l))\Eval \ = \{\sigma'(hi(\tau))\}.$$
	
	\item
	$$\sigma'(\Theta(\xi_i(\rho),\tau,id,m_1,\ldots,m_l))\Eval \ = \sigma'(\Theta^*(\rho', \mu_0(\rho'),\theta[\tau],m_1, \ldots, m_l))\Eval.$$
	
	\item
	\begin{eqnarray*}
	\sigma'(\Theta(\xi_i(\rho),\tau,id,m_1,\ldots,m_l))\Eval &=& \sigma'(\{\{x_1 \leftarrow t_1, \ldots, x_i \leftarrow t_i\}\}) \\
	 & & \   \sigma'(\Theta^*(\rho', \mu_0(\rho'), \theta[\tau], m_1, \ldots, m_l))\Eval.
	\end{eqnarray*}
\end{enumerate}
\end{definition}
\begin{theorem} \label{th.HerbrandsystemSoundComplete}
The Herbrand system of a proof schema $\Pbf \colon (\rho_0,R,\Ncal_0,t,\Seq,\Pi)$ is sound and complete. 
\end{theorem}
\begin{proof}
We want to show that every substitution produced by the Herbrand system corresponds to an actual trace of the quantified formula in the proof schema and therefore yields a Herbrand instance occurring in the evaluated proof (Soundness), and that every Herbrand instance occurring in the evaluated proof $\sigma(\rho_0)\Eval$ is generated by the Herbrand system (Completeness). \\[1ex]
Completeness follows almost immediately. 
Let $\sigma$ be a parameter assignment. By Theorem \ref{the.proof-schema}, $\sigma(\rho_0)\Eval$ is an $\LK$-proof. Completeness follows from the definition of $\Theta^*$: every trace of a quantified formula occurrence in the evaluated proof is considered, and the recursive clauses of Definitions \ref{def.Herbrandschema} and \ref{def.semantics} exhaust all possibilities for the endpoint of a trace. Hence, every Herbrand instance occurring in $\sigma(\rho_0)\Eval$ is generated by $\sigma(\Theta^*)\Eval$.\\[1ex]
%
%
Soundness can be shown by induction on the evaluation of $\sigma(\Theta^*(\rho,\mu,\lambda,m_1,\ldots,m_l))\Eval$. Since the underlying proof schema is canonical, recursion terminates. Thus the induction is well-founded.\\[1ex]
Base case: Case 1 of Definition \ref{def.semantics}
$$\sigma(\Theta^*(\rho,\mu,\lambda,m_1,\ldots,m_l))\Eval = \{\sigma(h_i(\tau))\}.$$
Here the trace ends in an axiom or quantifier-free occurrence. By Definition \ref{def.hi2} $h_i(\tau)$ is obtained from the actual instantiation terms occurring on the trace. Therefore, the resulting substitution is a Herbrand substitution.\\[1ex]
Assume as the IH that for every recursive call strictly below the current one, $\sigma(\Theta^*(\rho,\mu,\lambda,m_1,\ldots,m_l))\Eval$ produces valid Herbrand substitutions for the corresponding proof instance. \\[1ex]
Now consider the recursive cases. In case 2, 
$$\sigma(\Theta(\xi_i(\rho),\tau,id,m_1, \ldots, m_l)\Eval \ = \sigma(\Theta^*(\rho', \mu_0(\rho'), \theta[\tau], m_1,\ldots,m_l))\Eval.$$
By the IH, the recursive evaluation already returns correct Herbrand substitutions for the called proof term.  Hence the current evaluation is sound. \\[1ex]
In case 3, we have
$$\sigma(\Theta(\xi_i(\rho),\tau,id,m_1, \ldots, m_l)\Eval \ = \sigma(\{\{x_1 \leftarrow t_1, \ldots , x_i \leftarrow t_i \}\}) \sigma(\Theta^*(\rho', \mu_0(\rho'), \theta[\tau], m_1,\ldots,m_l))\Eval.$$
The first component comes from the instantiations already performed along the trace (Definition \ref{def.hi2} case 2 or 3). The second component is sound by the induction hypothesis. Therefore their application is sound.\\[1ex]
Since Definition \ref{def.semantics} takes the union over all traces, every generated substitution corresponds to an actual trace in the proof and therefore to a valid Herbrand instance, and soundness follows.
\end{proof}
Given a schematic end-sequent $S$ together with the Herbrand systems for the occurrences of quantified formulas in $S$, a schematic Herbrand sequent can be constructed in a straightforward way.
\begin{definition}
Let $\Pbf$ be a proof schema as in Definition \ref{def.Herbrandschema}, and let $\mu_0, \ldots , \mu_n$ be all the occurrences of quantified formulas
$$F_i \colon Qx_1 \ldots Qx_{\beta_i} F_i'(x_1, \ldots, x_{\beta_i})$$
in its end-sequent $Seq(\rho)$. Let
\begin{eqnarray*}
\Theta^*(\rho_0,\mu_i,id,m_1,\ldots,m_l) &\defeq& \{C_1(\rho_0)\colon \Theta(\xi_1(\rho_0),\mu_i,id,m_1,\ldots,m_l), \ldots, \\
 & & \mbox{ } \  C_k(\rho_0)\colon \Theta(\xi_k(\rho_0),\mu_i,id,m_1,\ldots,m_l)\},
\end{eqnarray*}
be the Herbrand systems for $\mu_i$ for $0 \leq i \leq n$.
Then the sequent $H(Seq(\rho))$ is obtained from $Seq(\rho)$ by replacing all occurrences of $\mu_i$ for $0 \leq i \leq n$ by 
$$F_i'(x_1, \ldots, x_{\beta_i}) \Theta^*(\rho_0,\mu_i,id,m_1,\ldots,m_l),$$
where $F \Theta^*(\rho_0,\mu_i,id,m_1,\ldots,m_l)$ denotes the sequence of formulas $F \theta_1, \ldots , F \theta_n$,  and $\theta_1, \ldots , \theta_n$ are all the substitutions in $\Theta^*(\rho_0,\mu_0,id,m_1,\ldots,m_l)$.
$H(Seq(\rho))$ is called the schematic Herbrand sequent of $\Pbf$.
\end{definition}
\begin{example}
Consider the proof schema $(\rho_0,\{\rho_0,\rho\},\{n\},t,\Seq,\Pi)$ as in Example \ref{ex.Herbrandschemapartial}, where
$$\Pi(\rho_0(n)) = \rho_0(n) \defeq \{n=0\colon \xi_1(\rho_0)(n),\ n>0\colon \xi_2(\rho_0)(n)\}.$$
The Herbrand system for $\mu_0$ is already computed in Example \ref{ex.Herbrandschemapartial}, but to compute the schematic Herbrand sequent of the proof schema, we have to consider the Herbrand systems of the other quantified formulas in the end-sequent
$$\Seq(\rho_0) = P(c), {\it fdef}, \forall x(P(x) \impl P(f(x))) \seq P(\fhat(c,s(n)))$$
as well. Let $\mu_1$ be the occurrence of $\forall x.\fhat(x,\bar{0}) = x$, and $\mu_2$ the occurrence of $\forall x \forall z(\fhat(x,s(z)) = f (\fhat(x,z)))$ in ${\it fdef}$.
The Herbrand schema for $\mu_1$ is 
$$\Theta^*(\mu_1) = \Theta^*(\rho_0, \mu_1, id, n) \defeq \{n=0 \colon \emptyset, \ n >0 \colon \Theta^*(\rho, \mu_1, \{n \leftarrow s(n)\}, n)\},$$
where 
$$\Theta^*(\rho, \mu_1, \{n \leftarrow s(n)\}, n) \defeq \{n=0 \colon \{ x \leftarrow c\}, \ n >0 \colon \Theta^*(\rho, \mu_1, \{n \leftarrow s(n)\}\{n \leftarrow p(n)\}, n )\}.$$
Therefore, we obtain the sequence of substitutions $\{x \leftarrow c\}$ independent of $n$. The Herbrand schema for $\mu_2$ is
$$\Theta^*(\mu_2) = \Theta^*(\rho_0, \mu_2, id, n) \defeq \{n=0 \colon \{x \leftarrow c, z \leftarrow 0\}, \ n >0 \colon \Theta^*(\rho, \mu_2, \{n \leftarrow s(n)\}, n)\},$$
where $\Theta^*(\rho, \mu_2, \{n \leftarrow s(n)\}, n) \defeq$
$$\{n=0 \colon \emptyset, \ n >0 \colon \{\{x \leftarrow c, z \leftarrow p(n)\}\} \cup \Theta^*(\rho, \mu_2, \{n \leftarrow p(n)\}\{n \leftarrow s(n)\}, n )\}.$$
Therefore, for a numeral $\alpha$, we obtain the sequence of substitutions $\{x \leftarrow c, z \leftarrow \alpha\}, \ldots , \{x \leftarrow c, z \leftarrow 0\}$.
The schematic Herbrand sequent $H(Seq(\rho))$ is 
$$P(c), (\fhat(x,\bar{0}) = x) \Theta^*(\mu_1), (\fhat(x,s(z)) = f(\fhat(x,z)))\Theta^*(\mu_2), (P(x) \impl P(f(x)))\Theta^*(\mu_0) \seq P(\fhat(c,s(n))),$$
where $\Theta^*(\mu_0) = \Theta^*(\rho_0, \mu_0, id, n)$ from Example \ref{ex.Herbrandschemapartial}. 
For $\sigma(n)=0$ we obtain the Herbrand sequent by replacing each formula in the end-sequent by its evaluation:
\begin{eqnarray*}
\sigma(P(c))\Eval &=& P(c),\\
\sigma((\fhat(x,\bar{0}) = x)) \sigma(\Theta^*(\mu_1))\Eval &=& \fhat(c,\bar{0}) = c,\\
\sigma((\fhat(x,s(z)) = f(\fhat(x,z))))\sigma(\Theta^*(\mu_2))\Eval &=& \fhat(c,s(0)) = f(\fhat(c,0)),\\
\sigma(P(x) \impl P(f(x)))\sigma(\Theta^*(\mu_0))\Eval &=& P(c) \impl P(f(c)),\\
\sigma(P(\fhat(c,s(n))))\Eval &=& P(\fhat(c,s(0))),
\end{eqnarray*}
$$P(c), \fhat(c,\bar{0}) = c, \fhat(c,s(0)) = f(\fhat(c,0)), P(c) \impl P(f(c)) \seq P(\fhat(c,s(0))).$$
For $\sigma(n)=1$ we obtain
\begin{eqnarray*}
P(c), & & \\
\fhat(c,\bar{0}) = c, & & \\
\fhat(c,s(0)) = f(\fhat(c,0)), \fhat(c,s(1)) = f(\fhat(c,1)), & & \\
P(c) \impl P(f(c)), P(\fhat(c,1)) \impl P(f(\fhat(c,1))) &\seq& P(\fhat(c,s(1)))
\end{eqnarray*}
\end{example}
\section{Cyclic Proofs}\label{sec.cyclicproofs}

To make this paper self-contained, we provide a brief overview, following \cite{brotherston2006.thesis, brotherston2011sequent}, of classical first-order logic with (mutually) inductively defined predicates, denoted $\mathsf{FOL_{ID}}$, in the style of Martin-Löf \cite{martin2000iterated}. We also recall the infinitary proof system $\mathsf{LKID}^{\omega}$, which formalises reasoning by infinite descent, and its restriction $\mathsf{CLKID}^{\omega}$, presented as a cyclic proof system. Since this is only a concise summary, we refer the reader to \cite{brotherston2006.thesis, brotherston2011sequent} for full details.

\vspace{0.5em}

We fix a standard (countable) first-order language $\Sigma$ with equality that includes predicate symbols classified as either \emph{ordinary} or \emph{inductive}. Ordinary predicates, terms, and formulas are defined in the usual way. Each inductive predicate symbol is specified by a finite set of production rules (axioms) that constitute its inductive definition.

\begin{definition}[Inductive definition set]\label{def:inductiveDefinitionSet}
An \emph{inductive definition set} $\Phi$ for a signature $\Sigma$ is a finite set of productions, where each production is a rule of the form
\begin{prooftree}
    \AxiomC{$Q_1(\mathbf{u_1}), \dots, Q_h(\mathbf{u_h}),
     P_{j1}(\mathbf{t_1}), \dots, P_{jm}(\mathbf{t_m})$}
    \UnaryInfC{$P_i(\mathbf{t})$}
\end{prooftree}
where $Q_1, \dots, Q_h$ are \emph{ordinary} predicate symbols and 
$P_{j1}, \dots, P_{jm}, P_i$ are \emph{inductive} predicate symbols. The bold symbols $\mathbf{u_1}, \dots, \mathbf{u_h}, \mathbf{t_1}, \dots, \mathbf{t_m}, \mathbf{t}$ denote vectors of terms of appropriate length.

We refer to the formulas above the line as the \emph{premises} of the production, and to the formula below the line as its \emph{conclusion}. 
For brevity, we often omit parentheses, writing, for instance, $Q_1\mathbf{u_1}$ instead of $Q_1(\mathbf{u_1})$. 
When it is useful to make the variables occurring in a production explicit, we may write formulas such as $Q_1\mathbf{u_1}(\mathbf{x})$, where $\mathbf{x}$ denotes a vector of variables.
\end{definition}

\begin{example}\label{ex.productions}
Let $\Sigma = \{ 0, s \}$ be a language with constant $0$ and unary function symbol $s$ (successor). 
The inductive predicates for natural numbers $N$, even numbers $E$, and odd numbers $O$ can be specified by the inductive definition set $\Phi$ as follows:

\begin{center}
\setlength{\tabcolsep}{6pt} 
\renewcommand{\arraystretch}{1.1} 
\begin{tabular*}{0.96\textwidth}{@{\extracolsep{\fill}}ccccc@{}}
\bottomAlignProof
\AxiomC{}
\UnaryInfC{$N0$}
\DisplayProof
&
\bottomAlignProof
\AxiomC{$N x$}
\UnaryInfC{$N s(x)$}
\DisplayProof
&
\bottomAlignProof
\AxiomC{}
\UnaryInfC{$E0$}
\DisplayProof
&
\bottomAlignProof
\AxiomC{$O x$}
\UnaryInfC{$E s(x)$}
\DisplayProof
&
\bottomAlignProof
\AxiomC{$E x$}
\UnaryInfC{$O s(x)$}
\DisplayProof
\end{tabular*}
\end{center}

\noindent
Note that, to obtain structures isomorphic to the natural numbers and to their even and odd subsets, we must include additional axioms ensuring the injectivity of $s$.

\end{example}

We briefly introduce $\mathsf{LKID}^{\omega}$, an infinitary proof system that formalises a version of proof by infinite descent in $\mathsf{FOL_{ID}}$. This system serves as the foundation for defining the cyclic proof system $\mathsf{CLKID}^{\omega}$.

The proof rules of $\mathsf{LKID}^{\omega}$ consist of those of the sequent calculus $\LK$ recalled in Definition~\ref{def.LK}, together with the rules for equality and substitution shown in Figure~\ref{fig:CLKID_subst_eq}. In addition, we define \emph{unfold rules} (right-introduction) and \emph{case-split rules} (left-introduction) for each inductive predicate.
Informally, the unfold rules correspond to the closure conditions in the definitions of inductively defined predicates, whereas the case-split rules capture the associated principles of natural induction.

\begin{figure}[b]
    \centering
    \bottomAlignProof
        \AxiomC{$\Gamma \vdash \Delta$}
        \RightLabel{\scriptsize $\mathsf{subst}$}
        \UnaryInfC{$\Gamma[\Theta] \vdash \Delta[\Theta]$}
    \DisplayProof
    \hspace{2em}
    \bottomAlignProof
        \AxiomC{$\Gamma[u/x,\, t/y] \vdash \Delta[u/x,\, t/y]$}
        \RightLabel{\scriptsize $=_l$}
        \UnaryInfC{$\Gamma[t/x,\, u/y],\; t = u \vdash \Delta[t/x,\, u/y]$}
    \DisplayProof
    \hspace{2em}
    \bottomAlignProof
        \AxiomC{}
        \RightLabel{\scriptsize $=_r$}
        \UnaryInfC{$\Gamma \vdash t = t,\, \Delta$}
    \DisplayProof

    \caption{
        Equality and substitution rules. 
        The notation $\Gamma[\Theta]$ denotes the result of applying the substitution $\Theta$ 
        (a mapping of terms to free variables) to all formulas in~$\Gamma$.
    }
    \label{fig:CLKID_subst_eq}
\end{figure}

\vspace{0.5em}
For each production $\phi \in \Phi$:

\begin{center}
\AxiomC{$Q_1(\mathbf{u_1}), \dots, Q_h(\mathbf{u_h}),
     P_{j_1}(\mathbf{t_1}), \dots,  P_{j_m}(\mathbf{t_m})$}
\UnaryInfC{$P_i(\mathbf{t})$}
\DisplayProof
\end{center}

the corresponding \emph{unfold} (right-introduction) rule is of the form:

\begin{center}
\AxiomC{$
\Gamma \vdash \Delta, Q_1(\mathbf{u_1}(\mathbf{u})), \dots, 
\Gamma \vdash \Delta, Q_h(\mathbf{u_h}(\mathbf{u})), 
\Gamma \vdash \Delta, P_{j_1}(\mathbf{t_1}(\mathbf{u})), \dots,
\Gamma \vdash \Delta, P_{j_m}(\mathbf{t_m}(\mathbf{u}))
$}
\RightLabel{$\phi_r$}
\UnaryInfC{$\Gamma \vdash \Delta, P_i(\mathbf{t}(\mathbf{u}))$}
\DisplayProof
\end{center}

Here, for each $r \in \{ \mathbf{u_1}, \dots, \mathbf{u_h}, \mathbf{t_1}, \dots, \mathbf{t_m}, \mathbf{t} \}$, 
the expression $r(\mathbf{u})$ represents the term $r$ after substituting the explicitly identified variables in the production $\phi$ with the corresponding terms in $\mathbf{u}$.

\begin{example}\label{ex.rightRules}
The unfold (right-introduction) rules for the inductive predicates $N$, $E$, and $O$ from Example~\ref{ex.productions} are:

\begin{center}
    \setlength{\tabcolsep}{0pt} 
    \renewcommand{\arraystretch}{1.1} 
    \begin{tabular*}{\textwidth}{@{\extracolsep{\fill}}ccccc}
    \bottomAlignProof
        \AxiomC{}
        \RightLabel{$N_{1r}$}
        \UnaryInfC{$\vdash N0$}
    \DisplayProof
    &
    \bottomAlignProof
        \AxiomC{$\Gamma \vdash Nx, \Delta$}
        \RightLabel{$N_{2r}$}
        \UnaryInfC{$\Gamma \vdash Ns(x), \Delta$}
    \DisplayProof
    &
    \bottomAlignProof
        \AxiomC{}
        \RightLabel{$E_{1r}$}
        \UnaryInfC{$\vdash E0$}
    \DisplayProof
    &
    \bottomAlignProof
        \AxiomC{$\Gamma \vdash Ox, \Delta$}
        \RightLabel{$E_{2r}$}
        \UnaryInfC{$\Gamma \vdash Es(x), \Delta$}
    \DisplayProof
    &
    \bottomAlignProof
        \AxiomC{$\Gamma \vdash Ex, \Delta$}
        \RightLabel{$O_{r}$}
        \UnaryInfC{$\Gamma \vdash Os(x), \Delta$}
    \DisplayProof
    \end{tabular*}
\end{center}
\end{example}

The \emph{case-split} rule serves as a left-introduction rule for inductive predicates. 
It formalises simple case distinctions that correspond directly to the productions of the predicate.
For each inductive predicate $P$, there is a corresponding case-split rule $(\mathsf{Case}\,P)$ of the form:

\begin{center}
\AxiomC{$\text{case distinctions}$}
\RightLabel{$(\mathsf{Case}\,P)$}
\UnaryInfC{$\Gamma, P(\mathbf{u}) \vdash \Delta$}
\DisplayProof
\end{center}

Specifically, for every production $\phi \in \Phi$ whose conclusion is $P$, such as

\begin{center}
\AxiomC{$Q_1(\mathbf{u_1}(\mathbf{x})), \dots, Q_h(\mathbf{u_h}(\mathbf{x})), P_{j_1}(\mathbf{t_1}(\mathbf{x})), \dots, P_{j_m}(\mathbf{t_m}(\mathbf{x}))$}
\UnaryInfC{$P(\mathbf{t}(\mathbf{x}))$}
\DisplayProof
\end{center}

there is a corresponding case distinction:

\begin{center}
$\Gamma, \mathbf{u} = \mathbf{t}(\mathbf{y}), Q_1(\mathbf{u_1}(\mathbf{y})), \dots, Q_h(\mathbf{u_h}(\mathbf{y})), P_{j_1}(\mathbf{t_1}(\mathbf{y})), \dots, P_{j_m}(\mathbf{t_m}(\mathbf{y})) \vdash \Delta$
\end{center}

where $\mathbf{y}$ is a vector of pairwise distinct eigenvariables of appropriate length satisfying $y \notin FV(\Gamma \cup \Delta \cup \{ P(\mathbf{u}) \})$ for all $y \in \mathbf{y}$. The formulas $P_{j_1}(\mathbf{t_1}(\mathbf{y})), \dots, P_{j_m}(\mathbf{t_m}(\mathbf{y}))$ are called the \emph{case-descendants} of the principal formula $P(\mathbf{u})$.


\begin{example}\label{ex.case-splitRules}
The case-split (left-introduction) rules for the inductive predicates $N$, $E$, and $O$ from Example~\ref{ex.productions} are:

\begin{center}
    \bottomAlignProof
        \AxiomC{$\Gamma, t = 0 \vdash \Delta$}
        \AxiomC{$\Gamma, t = s(x), Nx \vdash \Delta$}
        \RightLabel{$(\mathsf{Case}\, N)$}
    \BinaryInfC{$\Gamma, Nt \vdash \Delta$}
    \DisplayProof

    \vspace{1em}

    \begin{tabular}{c@{\hspace{3em}}c}
        \bottomAlignProof
            \AxiomC{$\Gamma, t = 0 \vdash \Delta$}
            \AxiomC{$\Gamma, t = s(x), Ox \vdash \Delta$}
            \RightLabel{$(\mathsf{Case}\, E)$}
        \BinaryInfC{$\Gamma, Et \vdash \Delta$}
        \DisplayProof
        &
        \bottomAlignProof
            \AxiomC{$\Gamma, t = s(x), Ex \vdash \Delta$}
            \RightLabel{$(\mathsf{Case}\, O)$}
        \UnaryInfC{$\Gamma, Ot \vdash \Delta$}
        \DisplayProof
    \end{tabular}
\end{center}

\end{example}

$\mathsf{LKID}^{\omega}$ is based on infinite derivation trees, where we distinguish between \emph{leaves} and \emph{buds}. 
Leaves are nodes that conclude inference rules without premises, i.e., axioms, whereas buds are non-leaf nodes that do not arise as the conclusion of any inference rule.

\begin{definition}[$\mathsf{LKID}^{\omega}$ pre-proof] 
An $\mathsf{LKID}^{\omega}$ \emph{pre-proof} of a sequent $S$ is a (possibly infinite) derivation tree built using the inference rules of $\mathsf{LKID}^{\omega}$, with root $S$, and containing no buds.
\end{definition}

To ensure the soundness of $\mathsf{LKID}^{\omega}$ pre-proofs, a \emph{global trace condition} is imposed, ensuring the continuous progression of arguments along infinite branches of a derivation. Informally, this condition requires that on every infinite branch of the proof tree, at least one inductively defined predicate on the left-hand side of the sequents is unfolded infinitely often.
A \emph{path} $\pi$ in a derivation tree is a finite or infinite sequence of sequents $\pi = (S_i)_{0 \le i < \alpha} $
with $\alpha \in \mathbb{N} \cup \{\infty\}$, such that $S_{i+1}$ is a child of $S_i$ for all $i + 1 < \alpha$.

\begin{definition}[Traces]
Let $(\Gamma_i \vdash \Delta_i)_{i \geq 0}$ be a path in a $\mathsf{LKID}^{\omega}$ pre-proof. A \emph{trace} along this path is a sequence of formulas $(\tau_i)_{i \geq 0}$ satisfying for all $i$:

\begin{itemize}
    \item $\tau_i = P_{j_i}(\mathbf{t_i}) \in \Gamma_i$, where $P_{j_i}$ is an inductive predicate appearing in $\Gamma_i$;
    \item if $(\Gamma_i \vdash \Delta_i)$ is the conclusion of a $(\mathsf{subst})$ rule with substitution $\Theta$, then $\tau_i = \tau_{i+1}[\Theta]$;
    \item if $(\Gamma_i \vdash \Delta_i)$ is the conclusion of a $(=_l)$ rule with principal formula $t = u$, then there is a formula $F$ and variables $x,y$ such that $\tau_i = F[t/x, u/y]$ and $\tau_{i+1} = F[u/x, t/y]$;
    \item if $(\Gamma_i \vdash \Delta_i)$ is the conclusion of a case-split rule, then either $\tau_{i+1} = \tau_i$ or $\tau_{i}$ is its principal formula and $\tau_{i+1}$ is a case-descendant of $\tau_i$. In this case, $i$ is called a \emph{progress point};
    \item if $(\Gamma_i \vdash \Delta_i)$ is the conclusion of any other rule, then $\tau_{i+1} = \tau_i$.
\end{itemize}

A \emph{progressing trace} is a trace that contains infinitely many progress points.
\end{definition}

\begin{definition}[$\mathsf{LKID}^{\omega}$ proof]
An $\mathsf{LKID}^{\omega}$ pre-proof $\Pi$ is an \emph{$\mathsf{LKID}^{\omega}$ proof} if it satisfies the \emph{global trace condition} -- that is, every infinite path $(\Gamma_i \vdash \Delta_i)_{i \ge 0}$ in $\Pi$ possesses an infinitely progressing trace along some tail $(\Gamma_i \vdash \Delta_i)_{i \ge k}$ for some $k \ge 0$.
\end{definition}

The cyclic proof system $\mathsf{CLKID}^{\omega}$ can be viewed as a restricted fragment of $\mathsf{LKID}^{\omega}$, obtained by limiting proofs to \emph{regular} derivation trees -- possibly infinite trees that contain only finitely many distinct subtrees. 
Proofs in $\mathsf{CLKID}^{\omega}$ are thus represented as finite graphs, where \emph{buds} are linked back to previously occurring sequents, called their \emph{companions}. 
Intuitively, a cyclic proof corresponds to an $\mathsf{LKID}^{\omega}$ derivation that can be seen as the result of repeatedly unfolding the buds of a finite proof graph into their companion nodes.

\begin{definition}[Companion]
Let $B$ be a bud in an $\mathsf{LKID}^{\omega}$ derivation tree $\Pi$. 
An internal node $C$ of $\Pi$ is called a \emph{companion} of $B$ if $B$ and $C$ correspond to identical sequents.
\end{definition}

\begin{definition}[$\mathsf{CLKID}^{\omega}$ pre-proof]
A $\mathsf{CLKID}^{\omega}$ pre-proof of a sequent $\Gamma \vdash \Delta$ is a pair $(\Pi, R)$, 
where $\Pi$ is a finite derivation tree constructed according to the inference rules of $\mathsf{LKID}^{\omega}$ 
with end-sequent $\Gamma \vdash \Delta$, and $R$ is a mapping that assigns to each bud in $\Pi$ a corresponding companion node.
\end{definition}

\begin{definition}[$\mathsf{CLKID}^{\omega}$ proof]
A $\mathsf{CLKID}^{\omega}$ proof is a $\mathsf{CLKID}^{\omega}$ pre-proof whose unfolding satisfies the global trace condition.
\end{definition}

\begin{example}\label{ex.even_odd_cyclic}
Consider the language $\Sigma=\{0,s\}$ together with the productions from Example~\ref{ex.productions} and the unfold and case-split rules given in Examples~\ref{ex.rightRules} and~\ref{ex.case-splitRules}.  
Below we present a $\mathsf{CLKID}^{\omega}$ proof of the sequent $E x \lor O x \vdash N x$.
The proof contains two bud–companion pairs, which are labelled (a) and (b) respectively.

\begin{center}

    \AxiomC{}
    \RightLabel{$N_{1r}$}
    \UnaryInfC{$\vdash N0$}
    \RightLabel{$=_l$}
    \UnaryInfC{$x = 0 \vdash Nx$}

    \AxiomC{$Ox \vdash Nx$ (a)}
    \RightLabel{\scriptsize$ \mathsf{subst}$}
    \UnaryInfC{$Oy \vdash Ny$}
    \RightLabel{$N_{2r}$}
    \UnaryInfC{$Oy \vdash Ns(y)$}
    \RightLabel{$=_l$}
    \UnaryInfC{$x = s(y), Oy \vdash Nx$}
    
    \RightLabel{$(\mathsf{Case}\, E)$}
    \BinaryInfC{$Ex \vdash Nx$ (b)}

    \AxiomC{$Ex \vdash Nx$ (b)}
    \RightLabel{\scriptsize $ \mathsf{subst}$}
    \UnaryInfC{$Ey \vdash Ny$}
    \RightLabel{$N_{2r}$}
    \UnaryInfC{$Ey \vdash Ns(y)$}
    \RightLabel{$=_l$}
    \UnaryInfC{$x = s(y), Ey \vdash Nx$}
    \RightLabel{$(\mathsf{Case}\, O)$}
    \UnaryInfC{$Ox \vdash Nx$ (a)}
    
    \RightLabel{$\lor_l$}
    \BinaryInfC{$Ex \lor Ox \vdash Nx$}
    \DisplayProof
\end{center}

\end{example}

\section{From Cyclic Proofs to Proof Schemata}\label{sec.translation}

In this section, we present an algorithmic translation of a subclass of $\mathsf{CLKID}^{\omega}$ proofs into the proof schema formalism. We consider $\mathsf{CLKID}^{\omega}$ derivations using a first-order language with equality over the domain of natural numbers with constant $0$ and the successor $s$. The resulting proof schemata additionally make use of the predecessor function $p$.

In addition, we restrict our attention to $\mathsf{CLKID}^{\omega}$ derivations in which no variable occurring in the guard of a case-split rule is strongly quantified\footnote{Universal quantifiers of positive polarity and existential quantifiers of negative polarity are called \emph{strong quantifiers.}} elsewhere in the proof. Intuitively, this ensures that the variables that are actually involved in representing the induction in a cyclic proof are not strongly quantified.

The cyclic framework recalled in the previous section and the proof schema formalism represent two fundamentally different approaches to inductive definitions. Consequently, they are not translatable into each other in full generality. 
Hence, we restrict attention to those inductive definition sets that can be reformulated in a way compatible with the proof schema framework.
Informally speaking, we can only consider cyclic proofs whose inductive definition sets behave recursively in the following sense: premises of productions do not introduce fresh variables, i.e.\ every free variable occurring in a premise also occurs in the conclusion, and productions defining the same inductive predicate are pairwise non-overlapping, so that at most one production applies to a given parameter instance.

This motivates the following class of inductive definitions.


\begin{definition}[Inductive predicates of definition type]\label{def.inductive_predicates_of_definition_type}
A set of production rules is of \emph{inductive predicate definition type} if, for every inductive predicate symbol $P$ of arity $n$, all its production rules are of the form
\[
\infer{P(x_1, \dots, x_n)}
      {C^i \colon P^i_1(\mathbf{r}^i_1) \land \cdots \land P^i_{k_i}(\mathbf{r}^i_{k_i})}
\]
where: 
\begin{itemize}
\item each $C^i$ is a condition over the signature $\{\bar{0}, s, p\}$ on variables $x_1, \dots, x_n$,
\item the conditions $C^i$ are pairwise disjoint,
\item $FV ( \{  P^i_1(\mathbf{r}^i_1), \dots , P^i_{k}(\mathbf{r^i_k}  \} ) \IN \{x_1, \dots , x_n \}$, 
\item and there are no other production rules for $P$.
\end{itemize}
\end{definition}

This definition ensures that inductive predicates behave like recursive definitions: for each parameter instance, at most one production applies, and no production introduces additional recursion variables through its premises. The definitions need neither be terminating nor complete.

\begin{example}
Consider the inductive definition set $\Phi$ from Example~\ref{ex.productions}. All productions can be transformed into the form required by Definition~\ref{def.inductive_predicates_of_definition_type}.

\begin{center}
\setlength{\tabcolsep}{6pt}
\renewcommand{\arraystretch}{1.1}
\begin{tabular*}{0.96\textwidth}{@{\extracolsep{\fill}}ccccc@{}}
\bottomAlignProof
\AxiomC{$x=0$}
\UnaryInfC{$Nx$}
\DisplayProof
&
\bottomAlignProof
\AxiomC{$x>0: N p(x)$}
\UnaryInfC{$N x$}
\DisplayProof
&
\bottomAlignProof
\AxiomC{$x=0$}
\UnaryInfC{$Ex$}
\DisplayProof
&
\bottomAlignProof
\AxiomC{$x>0: O p(x)$}
\UnaryInfC{$E x$}
\DisplayProof
&
\bottomAlignProof
\AxiomC{$x>0: E p(x)$}
\UnaryInfC{$O x$}
\DisplayProof
\end{tabular*}
\end{center}
\end{example}

\begin{example}
Consider the inductive definition set $\Phi_{R^+}$ from Example 2.2.7 in \cite{brotherston2006.thesis}, where $R$ is an ordinary predicate symbol of arity 2 and $R^+$ is an inductive predicate symbol of the same arity with the following production rules:
\begin{center}
\setlength{\tabcolsep}{6pt}
\renewcommand{\arraystretch}{1.1}
\begin{tabular*}{0.3\textwidth}{@{\extracolsep{\fill}}cc@{}}
\bottomAlignProof
\AxiomC{$R(x,y)$}
\UnaryInfC{$R^+(x,y)$}
\DisplayProof
&
\bottomAlignProof
\AxiomC{$R^+(x,y)$}
\AxiomC{$R^+(y,z)$}
\BinaryInfC{$R^+(x,z)$}
\DisplayProof
\end{tabular*}
\end{center}

$\Phi_{R^+}$ cannot be translated into a set of inductive predicates of definition type, since the second production introduces the variable $y$ in the premises that does not occur in the conclusion, violating the requirement that no fresh variables may be introduced by premises. Furthermore, the conditions induced by the productions are not pairwise disjoint according to Definition~\ref{def.inductive_predicates_of_definition_type}.
\end{example}
Note that Example~\ref{ex.even_odd_cyclic} from the previous section remains expressible within our framework. Moreover, numerous examples from Brotherston's thesis~\cite{brotherston2006.thesis} satisfy these conditions as well. Hence, despite the restrictions imposed by the translation, a substantial class of cyclic proofs remains within the scope of our approach.\\[1ex]
The core translation is performed by the algorithm \textsc{transformIntoSchema} (Alg.~\ref{alg:transformIntoSchema}). Given a cyclic proof $\Pi_C = (\Pi, R)$ with end-sequent $\Gamma \vdash \Delta$ and its associated inductive definition set $\Phi$, the algorithm constructs a proof schema $\mathbf{P} = (\rho_0, \mathcal{R}_P, \mathcal{N}_0, t, \mathit{Seq}, \Pi_P)$.
The resulting schema is rooted in the initial proof symbol $\rho_0$, which derives the sequent $\mathit{pdef}, \Gamma \vdash \Delta$. Here, \pdef\ denotes the first-order encoding of the inductive definition set $\Phi$, which will be explained later.
The transformation proceeds through five primary phases:
\begin{description}

    \item[Step 1: Initialization] 
The inductive definition set $\Phi$ is converted into a set of first-order axioms \pdef. The algorithm identifies all companion nodes in $\Pi$, which are subsequently used to partition the cyclic proof into its constituent subderivations in Step~2.
    
    \item[Step 2: Cycle Decomposition] 
    The structural unfolding of the cyclic proof $\Pi$ is simulated through the linking of modular $s$-proofs. In this step, each companion node identifies the root of a specific subderivation within $\Pi$. The algorithm partitions the global proof into these constituent subderivations, assigning a unique proof symbol $\rho_i$ to each.

\medskip
\noindent

\noindent
\begin{minipage}[t]{0.48\textwidth}
A cyclic proof $T$:
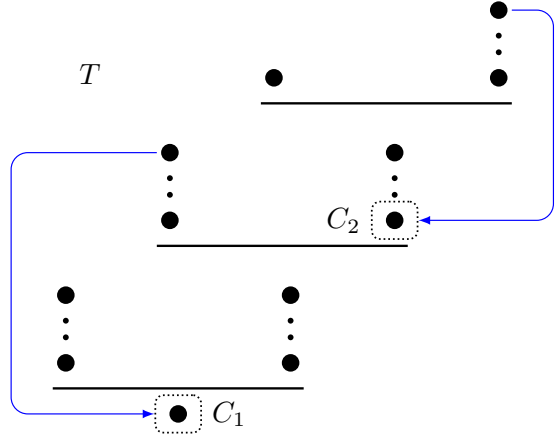
\begin{center}
\scalebox{1.1}{

\tikz[baseline=(current bounding box.north)]{
\node[minimum width=3cm, minimum height=.2cm, inner sep=0pt] (l1) {};
\draw[thick] (l1.west)--(l1.east);

\node[minimum width=.3cm, minimum height=.8cm, inner sep=0pt, anchor=south east] (b1) at ($(l1.north east)+(0,.2)$) {};
\fill[black] (b1.north) circle (3pt);
\fill[black] (b1.south) circle (3pt);
\fill[black] ($(b1.center)+(0,.1)$) circle (1pt);
\fill[black] ($(b1.center)+(0,-.1)$) circle (1pt);

\node[minimum width=.3cm, minimum height=.8cm, inner sep=0pt, anchor=south west] (b2) at ($(l1.north west)+(0,.2)$) {};
\fill[black] (b2.south) circle (3pt);
\node[minimum width=3cm, minimum height=.2cm, inner sep=0pt, below=1.5cm of l1, xshift=-1.25cm] (l2) {};
\draw[thick] (l2.west)--(l2.east);

\node[minimum width=.3cm, minimum height=.8cm, inner sep=0pt, anchor=south east] (b3) at ($(l2.north east)+(0,.2)$) {};
\fill[black] (b3.north) circle (3pt);
\fill[black] (b3.south) circle (3pt);
\fill[black] ($(b3.center)+(0,.1)$) circle (1pt);
\fill[black] ($(b3.center)+(0,-.1)$) circle (1pt);

\node[minimum width=.3cm, minimum height=.8cm, inner sep=0pt, anchor=south west] (b4) at ($(l2.north west)+(0,.2)$) {};
\fill[black] (b4.north) circle (3pt);
\fill[black] (b4.south) circle (3pt);
\fill[black] ($(b4.center)+(0,.1)$) circle (1pt);
\fill[black] ($(b4.center)+(0,-.1)$) circle (1pt);
\node[minimum width=3cm, minimum height=.2cm, inner sep=0pt, below=1.5cm of l2, xshift=-1.25cm] (l3) {};
\draw[thick] (l3.west)--(l3.east);

\node[minimum width=.3cm, minimum height=.8cm, inner sep=0pt, anchor=south east] (b5) at ($(l3.north east)+(0,.2)$) {};
\fill[black] (b5.north) circle (3pt);
\fill[black] (b5.south) circle (3pt);
\fill[black] ($(b5.center)+(0,.1)$) circle (1pt);
\fill[black] ($(b5.center)+(0,-.1)$) circle (1pt);

\node[minimum width=.3cm, minimum height=.8cm, inner sep=0pt, anchor=south west] (b6) at ($(l3.north west)+(0,.2)$) {};
\fill[black] (b6.north) circle (3pt);
\fill[black] (b6.south) circle (3pt);
\fill[black] ($(b6.center)+(0,.1)$) circle (1pt);
\fill[black] ($(b6.center)+(0,-.1)$) circle (1pt);

\node[minimum width=.3cm, minimum height=.1cm, inner sep=0pt, anchor=north] (b7) at ($(l3.south)+(0,-.2)$) {};
\fill[black] (b7.north) circle (3pt);
\node[draw, semithick, densely dotted, rounded corners=1mm, minimum width=.55cm, minimum height=.45cm, inner sep=0pt] (d1) at (b3.south) {};

\node[draw, semithick, densely dotted, rounded corners=1mm, minimum width=.55cm, minimum height=.45cm, inner sep=0pt] (d2) at (b7.north) {};

\draw[-latex, blue, rounded corners=2mm] (b1.north east)--++(.5,0) |- (d1.east);
\draw[-latex, blue, rounded corners=2mm] (b4.north west)--++(-1.75,0) |- (d2.west);

\node[left] at (d1.west) {$C_2$};
\node[right] at (d2.east) {$C_1$};
\node[above left=1cm of b4.north] {$T$};
}
}
\end{center}
\end{minipage}%
\hfill
\begin{minipage}[t]{0.48\textwidth}
$T$ gets decomposed in derivations $T_1$ and $T_2$:
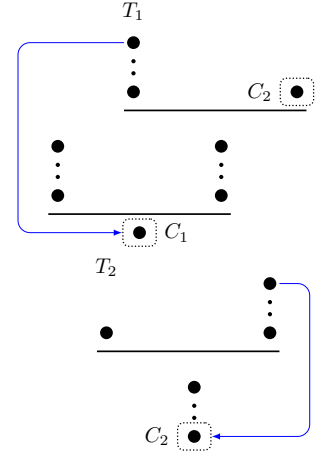
\begin{minipage}[t]{\textwidth}
\begin{center}
\scalebox{.8}{
\tikz[baseline=(current bounding box.north)]{
\node[minimum width=3cm, minimum height=.2cm, inner sep=0pt] (l2) {};
\draw[thick] (l2.west)--(l2.east);

\node[minimum width=.3cm, minimum height=.8cm, inner sep=0pt, anchor=south east] (b3) at ($(l2.north east)+(0,.2)$) {};
\fill[black] (b3.south) circle (3pt);

\node[minimum width=.3cm, minimum height=.8cm, inner sep=0pt, anchor=south west] (b4) at ($(l2.north west)+(0,.2)$) {};
\fill[black] (b4.north) circle (3pt);
\fill[black] (b4.south) circle (3pt);
\fill[black] ($(b4.center)+(0,.1)$) circle (1pt);
\fill[black] ($(b4.center)+(0,-.1)$) circle (1pt);
\node[minimum width=3cm, minimum height=.2cm, inner sep=0pt, below=1.5cm of l2, xshift=-1.25cm] (l3) {};
\draw[thick] (l3.west)--(l3.east);

\node[minimum width=.3cm, minimum height=.8cm, inner sep=0pt, anchor=south east] (b5) at ($(l3.north east)+(0,.2)$) {};
\fill[black] (b5.north) circle (3pt);
\fill[black] (b5.south) circle (3pt);
\fill[black] ($(b5.center)+(0,.1)$) circle (1pt);
\fill[black] ($(b5.center)+(0,-.1)$) circle (1pt);

\node[minimum width=.3cm, minimum height=.8cm, inner sep=0pt, anchor=south west] (b6) at ($(l3.north west)+(0,.2)$) {};
\fill[black] (b6.north) circle (3pt);
\fill[black] (b6.south) circle (3pt);
\fill[black] ($(b6.center)+(0,.1)$) circle (1pt);
\fill[black] ($(b6.center)+(0,-.1)$) circle (1pt);

\node[minimum width=.3cm, minimum height=.1cm, inner sep=0pt, anchor=north] (b7) at ($(l3.south)+(0,-.2)$) {};
\fill[black] (b7.north) circle (3pt);
\node[draw, semithick, densely dotted, rounded corners=1mm, minimum width=.55cm, minimum height=.45cm, inner sep=0pt] (d1) at (b3.south) {};

\node[draw, semithick, densely dotted, rounded corners=1mm, minimum width=.55cm, minimum height=.45cm, inner sep=0pt] (d2) at (b7.north) {};

\draw[-latex, blue, rounded corners=2mm] (b4.north west)--++(-1.75,0) |- (d2.west);

\node[left] at (d1.west) {$C_2$};
\node[right] at (d2.east) {$C_1$};
\node[above=.25cm of b4.north] {$T_1$};
}
}
\end{center}
\end{minipage}
\vspace{10pt}
\begin{minipage}[t]{\textwidth}
\begin{center}
\scalebox{.8}{
\tikz[baseline=(current bounding box.north)]{
\node[minimum width=3cm, minimum height=.2cm, inner sep=0pt] (l1) {};
\draw[thick] (l1.west)--(l1.east);

\node[minimum width=.3cm, minimum height=.8cm, inner sep=0pt, anchor=south east] (b1) at ($(l1.north east)+(0,.2)$) {};
\fill[black] (b1.north) circle (3pt);
\fill[black] (b1.south) circle (3pt);
\fill[black] ($(b1.center)+(0,.1)$) circle (1pt);
\fill[black] ($(b1.center)+(0,-.1)$) circle (1pt);

\node[minimum width=.3cm, minimum height=.8cm, inner sep=0pt, anchor=south west] (b2) at ($(l1.north west)+(0,.2)$) {};
\fill[black] (b2.south) circle (3pt);
\node[minimum width=3cm, minimum height=.2cm, inner sep=0pt, below=1.5cm of l1, xshift=-1.25cm] (l2) {};

\node[minimum width=.3cm, minimum height=.8cm, inner sep=0pt, anchor=south east] (b3) at ($(l2.north east)+(0,.2)$) {};
\fill[black] (b3.north) circle (3pt);
\fill[black] (b3.south) circle (3pt);
\fill[black] ($(b3.center)+(0,.1)$) circle (1pt);
\fill[black] ($(b3.center)+(0,-.1)$) circle (1pt);

\node[draw, semithick, densely dotted, rounded corners=1mm, minimum width=.55cm, minimum height=.45cm, inner sep=0pt] (d1) at (b3.south) {};

\draw[-latex, blue, rounded corners=2mm] (b1.north east)--++(.5,0) |- (d1.east);

\node[left] at (d1.west) {$C_2$};
\node[above=0cm of b2.north] {$T_2$};
}
}
\end{center}

\end{minipage}

\end{minipage}

For every resulting subderivation $T_i$, a local parameter set $\mathcal{N}_i$ is obtained from the variables occurring in the guards of case-split rules throughout the cyclic proof. More precisely, the algorithm first collects all variables occurring on the left-hand side of guards and then restricts this set to those variables that also occur in the end-sequent of $T_i$.

    \item[Step 3: Case Partitioning] 
    Derivations obtained in the previous step are further decomposed according to the branches of their case-split rules. This produces a set of derivations representing every possible trace.

\medskip
\noindent
A derivation $T$ gets split into three derivations:

    \begin{center}
   \scalebox{1}
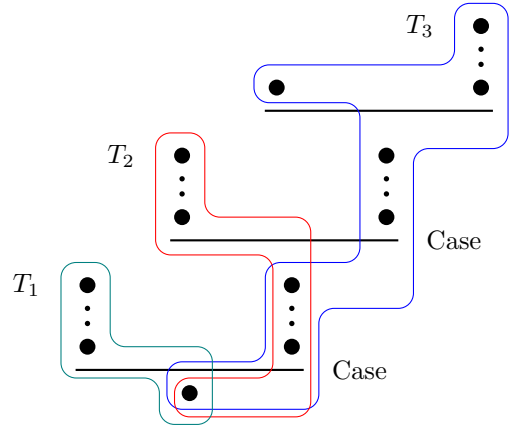
{
\tikz{

\node[minimum width=3cm, minimum height=.2cm, inner sep=0pt] (l1) {};
\draw[thick] (l1.west)--(l1.east);

\node[minimum width=.3cm, minimum height=.8cm, inner sep=0pt, anchor=south east] (b1) at ($(l1.north east)+(0,.2)$) {};
\fill[black] (b1.north) circle (3pt);
\fill[black] (b1.south) circle (3pt);
\fill[black] ($(b1.center)+(0,.1)$) circle (1pt);
\fill[black] ($(b1.center)+(0,-.1)$) circle (1pt);

\node[minimum width=.3cm, minimum height=.8cm, inner sep=0pt, anchor=south west] (b2) at ($(l1.north west)+(0,.2)$) {};
\fill[black] (b2.south) circle (3pt);
\node[minimum width=3cm, minimum height=.2cm, inner sep=0pt, below=1.5cm of l1, xshift=-1.25cm] (l2) {};
\draw[thick] (l2.west)--(l2.east);

\node[minimum width=.3cm, minimum height=.8cm, inner sep=0pt, anchor=south east] (b3) at ($(l2.north east)+(0,.2)$) {};
\fill[black] (b3.north) circle (3pt);
\fill[black] (b3.south) circle (3pt);
\fill[black] ($(b3.center)+(0,.1)$) circle (1pt);
\fill[black] ($(b3.center)+(0,-.1)$) circle (1pt);

\node[minimum width=.3cm, minimum height=.8cm, inner sep=0pt, anchor=south west] (b4) at ($(l2.north west)+(0,.2)$) {};
\fill[black] (b4.north) circle (3pt);
\fill[black] (b4.south) circle (3pt);
\fill[black] ($(b4.center)+(0,.1)$) circle (1pt);
\fill[black] ($(b4.center)+(0,-.1)$) circle (1pt);
\node[minimum width=3cm, minimum height=.2cm, inner sep=0pt, below=1.5cm of l2, xshift=-1.25cm] (l3) {};
\draw[thick] (l3.west)--(l3.east);

\node[minimum width=.3cm, minimum height=.8cm, inner sep=0pt, anchor=south east] (b5) at ($(l3.north east)+(0,.2)$) {};
\fill[black] (b5.north) circle (3pt);
\fill[black] (b5.south) circle (3pt);
\fill[black] ($(b5.center)+(0,.1)$) circle (1pt);
\fill[black] ($(b5.center)+(0,-.1)$) circle (1pt);

\node[minimum width=.3cm, minimum height=.8cm, inner sep=0pt, anchor=south west] (b6) at ($(l3.north west)+(0,.2)$) {};
\fill[black] (b6.north) circle (3pt);
\fill[black] (b6.south) circle (3pt);
\fill[black] ($(b6.center)+(0,.1)$) circle (1pt);
\fill[black] ($(b6.center)+(0,-.1)$) circle (1pt);

\node[minimum width=.3cm, minimum height=.1cm, inner sep=0pt, anchor=north] (b7) at ($(l3.south)+(0,-.2)$) {};
\fill[black] (b7.north) circle (3pt);
\node[left=.5cm of b1.north] {$T_3$};
\node[left=.5cm of b4.north] {$T_2$};
\node[left=.5cm of b6.north] {$T_1$};

\node[right=.25cm of l2.east] {Case};
\node[right=.25cm of l3.east] {Case};

\draw[blue, rounded corners=2mm] ($(b1.north)+(0,.3)$) -- ++(-.35,0) |- ($(b2.south)+(-.3,.3)$) --++(0,-.5) -| ($(b3.south)+(-.35,-.6)$) -| ($(b5.south)+(-.35,-.2)$) -| ($(b7.south)+(-.3,-.1)$) -| ($(b5.east)+(.2,.1)$) -| ($(b3.east)+(.2,.5)$) -| ($(b1.east)+(.2,0)$) |- ($(b1.north)+(0,.3)$);

\draw[red, rounded corners=2mm] ($(b4.north)+(0,.3)$) -- ++(-.35,0) |- ($(b5.north)+(-.25,.4)$) |- ($(b7.north)+(-.2,.2)$) |- ($(b7.south)+(.2,-.2)$) -| ($(b5.north)+(.25,.3)$) |- ($(b4.south)+(.3,0)$) |- ($(b4.north)+(0,.3)$);

\draw[teal, rounded corners=2mm] ($(b6.north)+(0,.3)$) -- ++(-.35,0) |- ($(b7.north)+(-.4,.2)$) |- ($(b7.south)+(.3,-.3)$) |- ($(b6.south)+(.3,0)$) |- ($(b6.north)+(0,.3)$);
}
   }
    \end{center}
    
    \item[Step 4: Subderivation Transformation] 
    Each partitioned derivation is transformed into a definition component for its respective proof symbol. This involves: (i) extracting logical conditions for the parameter set through the normalization of guards in case-split rules; and (ii) transforming the derivation into a valid s-proof $\zeta$.

\item[Step 5: Schema Assembly] 
The final components are aggregated into the proof schema $\mathbf{P}$. The set of proof symbols $\mathcal{R}_P$ is derived from the companion nodes. The global parameter set $\mathcal{N}_0$ is obtained as the union of all local parameter sets, while $t$ records the local parameter information by mapping each proof symbol $\rho_i$ to its parameter set $\mathcal{N}_i$. These local parameter sets are obtained from the variables occurring in the guards of case-split rules, restricted to those variables that also occur in the end-sequent of the corresponding subderivation $T_i$.
The conditional definition of each proof symbol, structured as a partition, represents every possible trace from the original cyclic proof.

\end{description}

\begin{algorithm}[H]
\caption{\textsc{transformIntoSchema}($\Pi_C = (\Pi, R), \Phi$)\label{alg:transformIntoSchema}}
\begin{algorithmic}[1]
\Require 
Cyclic proof $\Pi_C = (\Pi, R)$; Inductive definition set $\Phi$
\Ensure 
Proof schema $\mathbf{P} = (\rho_0, \mathcal{R}_P, \mathcal{N}_0, t, \mathit{Seq}, \Pi_P)$

\Statex \Comment{\textit{Step 1: Initialization}}
\State $\mathit{pdef} \gets \Call{createInductivePredicateAxiomSet}{\Phi}$
\State $C \gets \Call{getCompanionNodes}{R}$

\Statex \Comment{\textit{Step 2: Cycle Decomposition}}
\State $T \gets \Call{splitAtCompanions}{\Pi, C}$ 
\Comment{Returns dictionary $\{\rho_i \mapsto T_i\}$}

\State $\mathcal{N}_{\mathrm{loc}} \gets \Call{collectParameters}{\Pi,T}$
\Comment{Returns dictionary $\{\rho_i \mapsto \mathcal{N}_i\}$}

\State $\Pi_P \gets \text{new Dictionary}$

\Statex \Comment{\textit{Step 3: Case Partitioning and Transformation}}
\ForAll{$(\rho_i, T_i) \in T$}
    \State $\mathcal{N}_i \gets \mathcal{N}_{\mathrm{loc}}[\rho_i]$
   
    \State $\mathcal{T}_{i}^C \gets \Call{splitCases}{T_i}$ 
    \State $D(\rho_i) \gets \text{new Dictionary}$

    \Statex \Comment{\textit{Step 4: Subderivation Transformation}}
    \ForAll{$T_i^{C_j} \in \mathcal{T}_{i}^C$}
        \State $\mathit{guards}_i^{C_j} \gets \Call{collectGuards}{T_i^{C_j}}$
        \State $\mathit{cond}_i^{C_j} \gets \Call{normalizeGuards}{\mathit{guards}_i^{C_j}, \mathcal{N}_i}$
        
        \State $\mathit{subst}_i^{C_j} \gets \Call{extractSubstitutions}{\mathit{guards}_i^{C_j}, \mathcal{N}_i}$
  
        \State $\zeta_i^{C_j} \gets 
        \Call{transformSubderivation}
        {T_i^{C_j}, \mathit{cond}_i^{C_j}, \mathit{subst}_i^{C_j}, \mathit{pdef}, R, T, \mathcal{N}_{\mathrm{loc}}}$
   
        \State $D(\rho_i)[\mathit{cond}_i^{C_j}] \gets \zeta_i^{C_j}$
    \EndFor

    \If{$\neg \Call{isCompletePartition}{\text{keys}(D(\rho_i))}$}
        \State $\zeta_{else} \gets \Call{generateNegationProof}{\mathit{pdef}, \text{end-sequent}(\rho_i)}$
        \State $D(\rho_i)[\text{"else"}] \gets \zeta_{else}$
    \EndIf
    
    \State $\Pi_P[\rho_i] \gets D(\rho_i)$
\EndFor

\Statex \Comment{\textit{Step 5: Schema Assembly}}
\State $\mathcal{R}_P \gets \text{dom}(T)$
\State $t \gets \mathcal{N}_{\mathrm{loc}}$
\State $\mathcal{N}_0 \gets \bigcup_{\rho_i \in \mathcal{R}_P} \mathcal{N}_{\mathrm{loc}}[\rho_i]$
\Comment{Global parameter set}
\State $\mathit{Seq} \gets \{ \rho_i \mapsto \text{end-sequent}(\rho_i) \mid \rho_i \in \mathcal{R}_P \}$

\State \Return $(\rho_0, \mathcal{R}_P, \mathcal{N}_0, t, \mathit{Seq}, \Pi_P)$
\end{algorithmic}
\end{algorithm}

We now proceed by detailing each stage of the translation process. For each step, we provide an explanation followed by an example and the corresponding algorithm. 
The auxiliary procedures used in the algorithms are routine syntactic operations and are therefore not expanded further.

\subsection{Step 1: Initialization}\label{sub.step1}

\subsubsection{Encoding Inductive Definitions: \pdef}

The first requirement is to represent the inductive definition set $\Phi$ as a set of first-order formulas, denoted by \pdef. In $\mathsf{CLKID}^{\omega}$, inductive predicates are defined via production rules that often utilize successor symbols in the conclusion to represent the inductive step. To facilitate translation into a proof schema, these structural rules must be converted into guarded logical axioms.

For each production $\phi \in \Phi$, we apply a transformation $T$ that simplifies the conclusion by removing successor symbols and aligning premises using predecessor symbols. In addition, the transformation extracts matching constraints from the conclusion pattern, capturing cases where a variable occurs multiple times in the conclusion and thereby preserving dependencies between variables. This transformation is detailed in Algorithm~\ref{alg.pdef}.

\begin{enumerate}
    \item \textbf{Conclusion Simplification:}  
    Let $\mathcal{V}$ be the set of variables occurring in the conclusion term $\mathbf{t}$. For each variable $x \in \mathcal{V}$, let $k_x$ be the maximum number of successor applications to $x$ in $\mathbf{t}$ (i.e., $s^{k_x}(x)$ is a subterm of $\mathbf{t}$). All occurrences of maximal depth $s^{k_x}(x)$ in the conclusion are replaced by $x$ itself, and all shallower occurrences $s^j(x)$ (where $j<k_x$) are replaced by $p^{k_x-j}(x)$. Let the resulting simplified conclusion be $\mathbf{t}'$.

    \item \textbf{Premise Alignment:}  
    Every occurrence of a variable $x$ in the premises is replaced by $p^{k_x}(x)$. If a premise contains a successor application $s^m(x)$, the substitution yields $s^m(p^{k_x}(x))$, which simplifies to $p^{k_x-m}(x)$ or $s^{m-k_x}(x)$ via cancellation. Denote the resulting set of premises as $\mathcal{P}'$.

    \item \textbf{Domain Guards:}  
    Introduce a guard to ensure that each variable has sufficient predecessor depth:
    \[
    \Gamma_{\mathbf{t}} := \bigwedge_{x \in \mathcal{V}} (x \geq k_x).
    \]
   For zero-premise productions, $\Gamma_{\mathbf{t}} = \top$.
\end{enumerate}

\paragraph{Matching Constraints}
Independently of the premises, we extract the matching constraint from the conclusion:
\[
\Delta_{\mathbf{t}} := \bigwedge_{i=1}^n (x_i = t'_i),
\]
where $\mathbf{x} = (x_1,\dots,x_n)$ is the argument vector of $P$. This captures all structural dependencies between arguments that arise when variables occur multiple times in the conclusion pattern.

\begin{itemize}
    \item \textbf{Empty-premise productions:}  
    For a production with no premises, we add the universally closed axiom
    \[
    \Delta_{\mathbf{t}} \rightarrow P(\mathbf{x})
    \]
    to \pdef.

\item \textbf{Productions with premises:}  
For a production with premises, the transformation yields two universally closed guarded implications; both directions are required in order to simulate the case-split rules:

    \[
    \forall \mathbf{x} \bigl(
    (\Gamma_{\mathbf{t}} \land \Delta_{\mathbf{t}}) \to \bigl(\bigwedge \mathcal{P}'
    \rightarrow P(\mathbf{x})
    \bigr)\bigr),
    \]
    \[
    \forall \mathbf{x} \bigl( (
    \Gamma_{\mathbf{t}} \land \Delta_{\mathbf{t}}) \to \bigl( P(\mathbf{x})
    \rightarrow \bigwedge \mathcal{P}'
    \bigr)\bigr).
    \]
\end{itemize}
The use of both directions is justified under the standard semantics of inductive predicates over the natural numbers, where $0$, $s$, and $p$ are interpreted as usual, and an inductive predicate $P(\mathbf{x})$ holds exactly if it is derivable from the corresponding productions.

\paragraph{Pattern Exhaustion}
To provide a complete characterization, \pdef\ must specify that an inductive predicate $P$ is false when applied to arguments that do not match any production pattern. While this is, strictly speaking, an extension of the calculus, it is justified under the standard semantics, where such negative predicate instances are semantically false whenever they are not derivable from the productions. This extension is necessary for our formalism, as we depend on a complete characterization of the predicates.

\begin{enumerate}
    \item \textbf{Defining the Allowed Predicate:}  
    Let $P$ be an inductive predicate of arity $n$ with argument vector $\mathbf{x}$. For each production in $\Phi$ whose conclusion is $P(\mathbf{t})$, the transformation yields a domain guard $\Gamma_{\mathbf{t}}$ and a matching constraint $\Delta_{\mathbf{t}}$ (with $\Gamma_{\mathbf{t}} = \top$ for zero-premise productions).

    The allowed predicate is defined as the disjunction over all productions:
    \[
    \text{Allowed}_P(\mathbf{x}) :=
    \bigvee_{P(\mathbf{t}) \in \text{Concl}(P)}
    \left( \Gamma_{\mathbf{t}} \land \Delta_{\mathbf{t}} \right).
    \]

    \item \textbf{Generating the Negation Axiom:}  
    We add the axiom
    \[
    \forall \mathbf{x} \left( \neg \text{Allowed}_P(\mathbf{x}) \rightarrow \neg P(\mathbf{x}) \right),
    \]
    where the negation of $\text{Allowed}_P$ is obtained by pushing negation inward using standard propositional transformations.
\end{enumerate}

\begin{example}\label{ex.pdef_translation}
Applying the transformation $T$ to the definitions in Example~\ref{ex.productions}, we obtain the following inductive predicate axiom set. For predicates $N$ and $E$, the transformation follows the standard simplification and alignment:

\[
N^* = \Bigl\{ N(0), \quad 
\forall x \bigl( (x \ge \bar{1}) \to (N(px) \to N(x)) \bigr), \quad
\forall x \bigl( (x \ge \bar{1}) \to (N(x) \to N(px)) \bigr) \Bigr\}
\]

\[
E^* = \Bigl\{ E(0), \quad
\forall x \bigl( (x \ge \bar{1}) \to (O(px) \to E(x)) \bigr), \quad
\forall x \bigl( (x \ge \bar{1}) \to (E(x) \to O(px)) \bigr) \Bigr\}
\]

\noindent
For $O$, the construction is slightly more involved. The only allowed pattern for $O$ in $\Phi$ is $s(y)$, so we first derive the exhaustion axiom using the Allowed predicate:
\[
\text{Allowed}_O(x) := x \ge \bar{1}, \qquad
\forall x (\neg \text{Allowed}_O(x) \to \neg O(x))
\]
In the signature $\Sigma = \{0, s\}$, this effectively implies $\forall x (x = 0 \to \neg O(x))$. Combining this with the aligned inductive production, we then define:
\[
O^* = \Bigl\{ 
\forall x ( (x \ge \bar{1}) \to (E(px) \to O(x)) ), \quad
\forall x ( (x \ge \bar{1}) \to (O(x) \to E(px)) ), \quad
\forall x (x = 0 \to \neg O(x))
\Bigr\}
\]

Thus, the complete set is:
\[
\mathit{pdef} = N^* \cup E^* \cup O^*
\]

Note that for all these productions, the matching constraints $\Delta_{\mathbf{t}}$ are empty, as there are no dependencies between variables in the conclusions.
\end{example}

\begin{algorithm}[H]
\caption{\textsc{createInductivePredicateAxiomSet}($\Phi$)\label{alg.pdef}}
\begin{algorithmic}[1]
\Require Inductive definition set $\Phi$
\Ensure Inductive predicate axiom set \pdef
\State $\mathit{pdef} \gets \emptyset$
\State $Preds \gets \{ P \mid P \text{ is an inductive predicate in } \Phi \}$

\Statex \Comment{Store matching information per production}
\State $\mathcal{M} \gets \emptyset$ \Comment{will store tuples $(P,\Gamma_{\mathbf{t}},\Delta_{\mathbf{t}})$}

\ForAll{production $\phi \in \Phi$ of the form $\mathcal{P} \Rightarrow P(\mathbf{t})$}

    \State $\mathcal{V} \gets$ variables occurring in $\mathbf{t}$
    
    \Comment{Step 1: Compute maximal successor depth}
    \ForAll{$x \in \mathcal{V}$}
        \State $k_x \gets \max \{ k \mid s^k(x) \text{ is a subterm of } \mathbf{t} \}$
    \EndFor
    
    \Comment{Step 2: Domain guards}
    \If{$\mathcal{P} = \emptyset$}
        \State $\Gamma_{\mathbf{t}} \gets \top$
    \Else
        \State $\Gamma_{\mathbf{t}} \gets \bigwedge_{x \in \mathcal{V}} (x \ge k_x)$
    \EndIf

    \Comment{Step 3: Simplify conclusion and extract matching constraints}
    \State $\mathbf{t}' \gets \mathbf{t}$ with each maximal-depth $s^{k_x}(x)$ replaced by $x$ and each shallower $s^j(x)$ replaced by $p^{k_x-j}(x)$
    \State $\Delta_{\mathbf{t}} \gets \bigwedge_i (x_i = t'_i)$

    \State $\mathcal{M} \gets \mathcal{M} \cup \{ (P,\Gamma_{\mathbf{t}},\Delta_{\mathbf{t}}) \}$

    \If{$\mathcal{P} = \emptyset$} \Comment{Empty-premise production}
        \State $\mathit{pdef} \gets \mathit{pdef} \cup \{ \text{UnivClosure}(\Delta_{\mathbf{t}} \rightarrow P(\mathbf{x})) \}$
    \Else \Comment{Productions with premises}

        \Comment{Step 4: Align and simplify premises}
        \State $\mathcal{P}' \gets \{ \text{Simplify}(Q[x \mapsto p^{k_x}(x)]) \mid Q \in \mathcal{P} \}$

        \Comment{Step 5: Add guarded implications}
        \State $Axiom_1 \gets \text{UnivClosure}((\Gamma_{\mathbf{t}} \land \Delta_{\mathbf{t}}) \to (\bigwedge \mathcal{P}' \rightarrow P(\mathbf{x})))$
        \State $Axiom_2 \gets \text{UnivClosure}((\Gamma_{\mathbf{t}} \land \Delta_{\mathbf{t}}) \to (P(\mathbf{x}) \rightarrow \bigwedge \mathcal{P}'))$
        \State $\mathit{pdef} \gets \mathit{pdef} \cup \{Axiom_1, Axiom_2\}$
    \EndIf

\EndFor

\Statex \Comment{Step 6: Pattern exhaustion}
\ForAll{$P \in Preds$}
    \State $\mathbf{x} \gets$ fresh variables $(x_1,\dots,x_{\text{arity}(P)})$
    \State $\text{Allowed}_P(\mathbf{x}) \gets
    \bigvee_{(P,\Gamma_{\mathbf{t}},\Delta_{\mathbf{t}}) \in \mathcal{M}}
    (\Gamma_{\mathbf{t}} \land \Delta_{\mathbf{t}})$
    \State $\mathit{pdef} \gets \mathit{pdef} \cup \{ \forall \mathbf{x} (\neg \text{Allowed}_P(\mathbf{x}) \rightarrow \neg P(\mathbf{x})) \}$
\EndFor

\State \Return $\mathit{pdef}$
\end{algorithmic}
\end{algorithm}

\subsubsection{Identification of Companion Nodes}\label{sub.companion_nodes}

To determine the recursive entry points for the proof schema, we identify the set of nodes in the cyclic proof $\Pi_C= (\Pi, R)$ that serve as targets for back-links. Let $R$ be a mapping that associates each bud node $B_i$ (a leaf node) with its corresponding companion node $C_j$ (an internal node with an identical sequent). 

In the resulting proof schema, each $C_j$ in the set of companion nodes $C$ will later be assigned a unique proof symbol. This set is extracted by projecting the range of the mapping: $\textsc{getCompanionNodes}(\mathcal{R}) := \{ c \mid (b, c) \in \mathcal{R} \}$.

\begin{example}\label{ex:getCompanionNodes}
For the cyclic proof in Example~\ref{ex.even_odd_cyclic}, there are two bud-companion pairs with two different companion nodes. Hence, the algorithm extracts both companion nodes containing the sequents $Ex \vdash Nx$ and $Ox \vdash  Nx$ as the set $C$.
\end{example}


\subsection{Step 2: Cycle Decomposition}\label{sub.step2}

\subsubsection{Companion-Based Decomposition}
To translate a cyclic derivation $\Pi$ into a proof schema, we must decompose the global cyclic structure into a collection of finite sub-derivations. Each component corresponds to a unique proof symbol $\rho_i$, which serves as the identifier for a specific proof fragment in the resulting schema.

The procedure \textsc{splitAtCompanions} (Algorithm~\ref{alg.split_at_companions}) executes this decomposition by isolating sub-derivations based on the set of companion nodes $C$. It first assigns the root of the cyclic derivation to the main proof symbol $\rho_0$, regardless of whether this root node is a companion node or not. The algorithm then extracts a sub-derivation for every other companion node $C_i \in C$, defined as the sub-derivation of $\Pi$ rooted at $C_i$.

Once these sub-derivations are extracted, they are processed by the \textsc{pruneAtCompanions} procedure (Algorithm~\ref{alg.prune_at_companions}). This algorithm performs a depth-first traversal of each sub-derivation and "prunes" every branch at the first point it encounters a companion node (excluding the root of the sub-derivation itself). By removing the sub-derivations above these points, we effectively transform the cyclic back-links (companions) into leaves. These leaves act as placeholders that will be transformed into proof calls to the appropriate proof symbols in the final schema construction. The final output of the decomposition is a dictionary mapping each proof symbol $\rho_i$ to its corresponding finite, pruned sub-derivation $T_i$.

\begin{example}\label{ex.split_application}
Applying \textsc{splitAtCompanions} to the derivation in Example~\ref{ex.even_odd_cyclic} with the companion set $C$ from Example~\ref{ex:getCompanionNodes} yields the dictionary $\mathcal{D} = \{ \rho_0 \mapsto T_0, \rho_1 \mapsto T_1, \rho_2 \mapsto T_2 \}$, where each $T_i$ is a finite derivation tree:

\begin{itemize}
    \item \textbf{$T_0 $:} Rooted at $Ex \lor Ox \vdash Nx$.
    \begin{center}
    \small
    \AxiomC{$Ex \vdash Nx$ ($b$)}
    \AxiomC{$Ox \vdash Nx$ ($a$)}
    \RightLabel{$\lor_l$}
    \BinaryInfC{$Ex \lor Ox \vdash Nx$}
    \DisplayProof
    \end{center}
    The branches are pruned immediately as both are companions.

    \item \textbf{$T_1$ :} Rooted at $Ex \vdash Nx$.
    \begin{center}
    \small
    \AxiomC{}
    \RightLabel{$N_{1r}$}
    \UnaryInfC{$\vdash N0$}
    \RightLabel{$w$}
    \UnaryInfC{$x= 0 \vdash N0$}
    \RightLabel{$=_l$}
    \UnaryInfC{$x = 0 \vdash Nx$}

    \AxiomC{$Ox \vdash Nx$ ($a$)}
    \RightLabel{\scriptsize$\mathsf{subst}$}
    \UnaryInfC{$Oy \vdash Ny$}
    \RightLabel{$N_{2r}$}
    \UnaryInfC{$Oy \vdash Ns(y)$}
    \RightLabel{$w$}
    \UnaryInfC{$x = s(y),Oy \vdash Ns(y)$}
    \RightLabel{$=_l$}
    \UnaryInfC{$x = s(y), Oy \vdash Nx$}
    
    \RightLabel{$(\mathsf{Case}\, E)$}
    \BinaryInfC{$Ex \vdash Nx$}
    \DisplayProof
    \end{center}

\newpage
    \item \textbf{$T_2$ :} Rooted at $Ox \vdash Nx$.
    \begin{center}
    \small
    \AxiomC{$Ex \vdash Nx$ ($b$)}
    \RightLabel{\scriptsize $\mathsf{subst}$}
    \UnaryInfC{$Ey \vdash Ny$}
    \RightLabel{$N_{2r}$}
    \UnaryInfC{$Ey \vdash Ns(y)$}
    \RightLabel{$w$}
    \UnaryInfC{$x = s(y),Ey \vdash Ns(y)$}
    \RightLabel{$=_l$}
    \UnaryInfC{$x = s(y), Ey \vdash Nx$}
    \RightLabel{$(\mathsf{Case}\, O)$}
    \UnaryInfC{$Ox \vdash Nx$}
    \DisplayProof
    \end{center}
    
\end{itemize}
\end{example}

\begin{algorithm}[H]
\caption{\textsc{splitAtCompanions}($\Pi, C$)}\label{alg.split_at_companions}
\begin{algorithmic}[1]
\Require Derivation $\Pi$, set of companion nodes $C = \{C_1, \dots, C_k\}$
\Ensure A dictionary $\mathcal{D}$ mapping proof symbols $\rho_i$ to pruned sub-derivations $T_i$
\State $\mathcal{D} \gets \text{empty dictionary}$
\State $r \gets \text{root node of } \Pi$
\State $L \gets [r]$ \Comment{Initialize list with the root of the proof}
\ForAll{$c \in C$}
    \If{$c \neq r$} 
        \State add $c$ to the end of $L$ \Comment{Avoid duplicate entry if root is a companion}
    \EndIf
\EndFor

\For{$i = 0$ \textbf{to} $\text{length}(L) - 1$}
    \State $n \gets L[i]$
    \State $T_{sub} \gets \text{extract sub-derivation of } \Pi \text{ rooted at } n$
    \State $T_i \gets \textsc{pruneAtCompanions}(T_{sub}, C \setminus \{n\})$
    \State $\mathcal{D}[\rho_i] \gets T_i$ \Comment{Map symbol $\rho_i$ to the finite tree $T_i$}
\EndFor
\State \Return $\mathcal{D}$
\end{algorithmic}
\end{algorithm}

\begin{algorithm}[H]
\caption{\textsc{pruneAtCompanions}($T, C_{prune}$)}\label{alg.prune_at_companions}
\begin{algorithmic}[1]
\Require A derivation $T$ and a set of companion nodes $C_{prune}$
\Ensure A derivation $T'$ truncated at the first encounter of any $c \in C_{prune}$ along any branch
\State $T' \gets T$
\State \textbf{procedure} \textsc{TraverseAndPrune}($n$):
    \If{$n \in C_{prune}$}
        \State Remove all sub-derivations above $n$ \Comment{$n$ becomes a leaf}
        \State \Return \Comment{Stop exploring this specific branch}
    \Else
        \State Identify the premises $p_1, \dots, p_m$ of the rule applied at $n$
        \ForAll{$p_j \in \{p_1, \dots, p_m\}$}
            \State \textsc{TraverseAndPrune}($p_j$) \Comment{Continue searching other branches}
        \EndFor
    \EndIf
\State \textsc{TraverseAndPrune}(root of $T'$)
\State \Return $T'$
\end{algorithmic}
\end{algorithm}

\subsubsection{Identification of Parameters}\label{sub.collect_parameters}

After the cyclic proof has been decomposed into finite sub-derivations, the parameter sets of the resulting proof symbols are determined from the variables that control case distinctions in the original cyclic proof. The procedure \textsc{collectParameters} first traverses the whole cyclic proof $\Pi$ and collects all variables occurring on the left-hand side of guards in case-split rules. This yields a global candidate set of parameters.

The algorithm then restricts this candidate set separately for each sub-derivation $T_i$. Since the proof symbol $\rho_i$ represents the component $T_i$, only those candidate parameters that occur in the end-sequent of $T_i$ are kept for $\rho_i$. Thus, the output is a dictionary
$\mathcal{N}_{\mathrm{loc}} = \{\rho_i \mapsto \mathcal{N}_i\}$,
where each $\mathcal{N}_i$ contains precisely those global candidate parameters that occur in the end-sequent of $T_i$.

\begin{example}
Consider the dictionary
$T = \{\rho_0 \mapsto T_0,\rho_1 \mapsto T_1,\rho_2 \mapsto T_2\}$ 
obtained from \textsc{splitAtCompanions} in Example~\ref{ex.split_application}.  
The procedure \textsc{collectParameters} first traverses the whole cyclic proof $\Pi$ and collects all variables occurring on the left-hand side of guards in case-split rules. In the running example, the relevant guards are of the form $x=0$ and $x=s(y)$, so the global candidate set is
\[
L = \{x\}.
\]

The algorithm then restricts this set to the variables occurring in the end-sequent of each component $T_i$. Since the end-sequents of $T_0$, $T_1$, and $T_2$ all contain $x$, we obtain
\[
\mathcal{N}_{\mathrm{loc}}
=
\{\rho_0 \mapsto \{x\},
  \rho_1 \mapsto \{x\},
  \rho_2 \mapsto \{x\}\}.
\]
Thus, each proof symbol receives the local parameter set consisting of the global candidate parameters that occur in its own end-sequent.
\end{example}

\begin{algorithm}[H]
\caption{\textsc{collectParameters}($\Pi,T$)}\label{alg.collect_params}
\begin{algorithmic}[1]
\Require Cyclic proof $\Pi$; dictionary $T = \{\rho_i \mapsto T_i\}$ mapping proof symbols to pruned sub-derivations
\Ensure A dictionary $\mathcal{N}_{\mathrm{loc}} = \{\rho_i \mapsto \mathcal{N}_i\}$ mapping proof symbols to local parameter sets

\State $L \gets \emptyset$
\Comment{Variables occurring on the left-hand side of guards in $\Pi$}
\State $\mathcal{N}_{\mathrm{loc}} \gets \text{empty dictionary}$

\Statex \Comment{\textit{Collect global parameter candidates}}
\State \textbf{traverse} each rule $R'$ in $\Pi$:
\If{$R'$ is a case-split rule}
    \State $\mathcal{G} \gets \text{set of guards in the premises of } R'$
\ForAll{$(t = t') \in \mathcal{G}$} \State $L \gets L \cup \text{free variables occurring in } t$
    \EndFor
\EndIf

\Statex \Comment{\textit{Restrict candidates to the end-sequent of each component}}
\ForAll{$(\rho_i,T_i) \in T$}
    \State $V_i \gets \text{free variables occurring in } \text{end-sequent}(T_i)$
    \State $\mathcal{N}_i \gets L \cap V_i$
    \State $\mathcal{N}_{\mathrm{loc}}[\rho_i] \gets \mathcal{N}_i$
\EndFor

\State \Return $\mathcal{N}_{\mathrm{loc}}$
\end{algorithmic}
\end{algorithm}


\subsection{Step 3: Case Partitioning}\label{sub.step3}

After decomposing the cyclic proof into finite sub-derivations based on the set of companion nodes, we must address the branching introduced by case-split rules. In a proof schema, each proof symbol $\rho_i$ is assigned a definition $D(\rho_i)$. This definition consists of multiple components, where each component pairs a specific condition on the parameters with a corresponding s-proof.

To facilitate this, any derivation extracted in the previous step that contains a case-split rule must be partitioned into multiple \textit{case-linear} derivations. In this context, \textit{linear} refers strictly to the \textit{elimination} of case-split branching; any other inference rules, whether unary or binary, remain untouched at this transformational stage. Each of these case-linear derivations will eventually correspond to a single guarded component in the definition $D(\rho_i)$.

The procedure \textsc{splitCases} (Algorithm~\ref{alg.split_cases}) transforms a derivation $T$ into a set of derivations $\mathcal{S}$ such that each member of $\mathcal{S}$ follows exactly one premise branch for every case-split rule encountered. The algorithm operates bottom-up: it identifies the case-split rule closest to the root and duplicates the entire derivation for each premise of that rule. Because this process is applied recursively, it successfully handles nested case-splits, producing one derivation for every possible combination of cases across all splits in the tree.

\begin{example}\label{ex.split_cases_application}
Consider the pruned sub-derivation $T_1$ from Example~\ref{ex.split_application}, which is rooted at the sequent $Ex \vdash Nx$. This derivation contains one case-split rule, $(\mathsf{Case}\, E)$, which partitions the proof into the cases ($x=0$) and ($x=s(y)$).

Applying \textsc{splitCases}($T_1$) results in the set $\mathcal{S}_{T_1} = \{T_{1,1}, T_{1,2}\}$, where:

\begin{itemize}
    \item \textbf{$T_{1,1}$:}
    \begin{center}
    \small
    \AxiomC{}
    \RightLabel{$N_{1r}$}
    \UnaryInfC{$\vdash N0$}
    \RightLabel{$w$}
    \UnaryInfC{$x= 0 \vdash N0$}
    \RightLabel{$=_l$}
    \UnaryInfC{$x = 0 \vdash Nx$}
    \RightLabel{$(\mathsf{Case}\, E)$}
    \UnaryInfC{$Ex \vdash Nx$}
    \DisplayProof
    \end{center}

    \item \textbf{$T_{1,2}$:}
    \begin{center}
    \small
    \AxiomC{$Ox \vdash Nx$ ($n_a$)}
    \RightLabel{\scriptsize$\mathsf{subst}$}
    \UnaryInfC{$Oy \vdash Ny$}
    \RightLabel{$N_{2r}$}
    \UnaryInfC{$Oy \vdash Ns(y)$}
    \RightLabel{$w$}
    \UnaryInfC{$x = s(y),Oy \vdash Ns(y)$}
    \RightLabel{$=_l$}
    \UnaryInfC{$x = s(y), Oy \vdash Nx$}
    \RightLabel{$(\mathsf{Case}\, E)$}
    \UnaryInfC{$Ex \vdash Nx$}
    \DisplayProof
    \end{center}
\end{itemize}
In $T_{1,1}$, the second premise of $(\mathsf{Case}\, E)$ has been removed, while in $T_{1,2}$, the first premise has been removed. Both derivations maintain the original root.
\end{example}

\begin{algorithm}[H]
\caption{\textsc{splitCases}($T$)}\label{alg.split_cases}
\begin{algorithmic}[1]
\Require A derivation $T$
\Ensure A set of derivations $\mathcal{S}$ where each case-split rule is restricted to a single premise
\If{$T$ contains no case-split rules}
    \State \Return $\{T\}$
\EndIf

\State Let $R$ be the case-split rule in $T$ that is closest to the root
\State Let $p_1, \dots, p_k$ be the premises of $R$
\State $\mathcal{S} \gets \emptyset$

\For{$i = 1$ \textbf{to} $k$}
    \State $T_i \gets$ copy of $T$
    \State In $T_i$, remove all premise branches of $R$ except the one starting with $p_i$
    \State \Comment{All inference steps below $R$ and above $p_i$ are preserved}
    \State $\mathcal{S}_i \gets \textsc{splitCases}(T_i)$ \Comment{Recursively resolve any remaining case-splits in $T_i$}
    \State $\mathcal{S} \gets \mathcal{S} \cup \mathcal{S}_i$
\EndFor

\State \Return $\mathcal{S}$
\end{algorithmic}
\end{algorithm}

\subsection{Step 4: Subderivation Transformation}\label{sub.step4}
This stage operates on the case-linear derivations obtained from the previous step to construct the mapping $\Pi_P$, which assigns a definition $D(\rho)$ to every proof symbol $\rho$. We encode $\Pi_P$ as a dictionary mapping proof symbols to a collection of pairs, each consisting of a condition on parameters and a corresponding s-proof.

\subsubsection{Normalization of Guards, Conditions and Substitutions}

The transformation of cyclic derivations into proof schemata requires a systematic way to handle the conditions under which different proof paths are valid. This process begins with the identification of \textit{guards}, the equation systems found in case-split rules. In our case-linear derivations, every case-split rule is unary, representing one specific branch of the original split.

However, these guards often introduce auxiliary variables that are not part of the parameter set $\mathcal{N}_i$ associated with the current proof symbol $\rho_i$. For example, a guard might state $x = s(y)$, where $x \in \mathcal{N}_i$ is a parameter of the current component but $y$ is a new auxiliary variable. To obtain a proper s-proof, these additional variables must be eliminated or re-expressed. We achieve this in two ways:
\begin{itemize}
    \item \textbf{\textsc{normalizeGuards}}: Used for the partition. It transforms guards into a single conjunctive formula imposing conditions only on parameters in $\mathcal{N}_i$ (e.g., $x > 0$).
    \item \textbf{\textsc{extractSubstitutions}}: Used for the derivation itself. It rewrites auxiliary variables in terms of the parameters in $\mathcal{N}_i$, together with $s$ and $p$ (e.g., $\{y \xleftarrow{} p(x) \}$).
\end{itemize}

\begin{example}\label{ex.collect_guards_even_odd}
Consider the case-linear derivation $T_{1,2}$ from Example~\ref{ex.split_cases_application}. The $(\mathsf{Case}\, E)$ rule at the root of this sub-derivation contains a premise defined by the guard $x = s(y)$.
\begin{center}
    \small
    \AxiomC{$\vdots$}
    \noLine
    \UnaryInfC{$x = s(y), Oy \vdash Nx$}
    \RightLabel{$(\mathsf{Case}\, E)$}
    \UnaryInfC{$Ex \vdash Nx$}
    \DisplayProof
\end{center}
In this instance, $x \in \mathcal{N}_1$ is a parameter of the component corresponding to $\rho_1$, but $y$ is an auxiliary variable. The procedure \textsc{collectGuards} isolates this equation so that subsequent normalization steps can eventually eliminate $y$. For $T_{1,2}$, the procedure yields the set $\mathcal{G} = \{ x = s(y) \}$.

Following the decomposition in Example~\ref{ex.split_application} and the partitioning in Example~\ref{ex.split_cases_application}, we apply \textsc{collectGuards} to each case-linear derivation. For the components considered here, the relevant local parameter sets contain $x$. The results are summarized below:

\begin{table}[H]
\centering
\renewcommand{\arraystretch}{1.3} 
\begin{tabular}{|l|l|l|}
\hline
\textbf{Derivation} & \textbf{Root Sequent} & \textbf{Output $\mathcal{G}$} \\ \hline \hline
$T_0$               & $Ex \lor Ox \vdash Nx$ & $\emptyset$                   \\ \hline
$T_{1,1}$           & $Ex \vdash Nx$        & $\{ x = 0 \}$                 \\ \hline
$T_{1,2}$           & $Ex \vdash Nx$        & $\{ x = s(y) \}$              \\ \hline
$T_{2}$             & $Ox \vdash Nx$        & $\{ x = s(y) \}$              \\ \hline
\end{tabular}
\caption{Guards collected from the Even/Odd case-linear sub-derivations.}
\end{table}

\noindent
In $T_{1,1}$, the variable $x$ is matched directly to a ground term $0$. In $T_{1,2}$ and $T_{2}$, the guard introduces the auxiliary variable $y$, which is not a member of the corresponding local parameter set. The existence of these $y$ variables is exactly why the subsequent normalization and substitution steps are necessary.
\end{example}

The formalization of this gathering step is provided in Algorithm~\ref{alg.collect_guards}.

\begin{algorithm}[H]
\caption{\textsc{collectGuards}($T$)}\label{alg.collect_guards}
\begin{algorithmic}[1]
\Require Derivation $T$ where all occurring case-split rules are unary
\Ensure Set of guard equations $\mathcal{G}$ extracted from the case-split rules in $T$
\State $\mathcal{G} \gets \emptyset$

\ForAll{nodes $N$ in $T$ traversed from root to leaves}
    \If{$N$ is the premise of a case-split rule}
        \State $g \gets$ set of equations $\mathbf{u} = \mathbf{t}(\mathbf{y})$ found in $N$
        \State $\mathcal{G} \gets \mathcal{G} \cup g$ \Comment{Flatten systems into a single set}
    \EndIf
\EndFor
\State \Return $\mathcal{G}$
\end{algorithmic}
\end{algorithm}


\paragraph{\textsc{normalizeGuards}} 
To convert the guards extracted from case-split rules into a standardized format for the proof schema's partition, we define the procedure \textsc{normalizeGuards}. This algorithm processes the set of equations $\mathcal{G}$ relative to the local parameter set $\mathcal{N}_i$ through three stages:

\begin{itemize}
    \item \textbf{Auxiliary Substitution}: The algorithm first eliminates all auxiliary variables $v \notin \mathcal{N}_i$ by propagating their definitions through the set of equations. These substitutions collapse chains of dependencies until the equations consist solely of relationships between parameters, constants, or irreducible auxiliary terms.
    
    \item \textbf{Depth Extraction}: Each side of an equation is analyzed using a helper function \textsc{ExtractDepth}. This function identifies the base term (a parameter $v$, the constant $0$, or an auxiliary variable) and the number of successor applications $n$ surrounding it. For example, the term $s(s(x))$ is represented as the pair $(x, 2)$.
    
    \item \textbf{Arithmetic Translation}: The algorithm maps these depths to arithmetic constraints based on the nature of the base terms:
    \begin{itemize}
        \item \textbf{Parameter-to-Parameter}: If both sides involve parameters $v_L, v_R \in \mathcal{N}_i$, the \textsc{Align} function simplifies the successor applications by removing common successors from both sides (e.g., $s(s(y)) = s(z) \to s(y) = z$).
        \item \textbf{Parameter-to-Constant}: If a parameter $v_L$ is found within a term $s^{n_L} v_L$ that is equated to a ground term $s^{n_R} 0$, the integer value is calculated by offsetting the depths: $(v_L = n_R - n_L)$. For instance, $s(x) = s(s(0))$ yields $x = \bar{1}$.
        \item \textbf{Parameter-to-Auxiliary}: If a parameter $v_L$ in $s^{n_L} v_L$ is equated to an irreducible auxiliary variable $v_R$ in $s^{n_R} v_R$ (where $v_R \notin \mathcal{N}_i \cup \{0\}$), it implies $v_L$ must be at least large enough to satisfy the successor chain. This is translated into an inequality $(v_L > n_L + n_R - 1)$.
    \end{itemize}
\end{itemize}

\paragraph{Disjointness and Completeness} 
A crucial property of the conditions generated by \textsc{NormalizeGuards} is their mutual exclusivity. Because the sets $\mathcal{G}$ are extracted from the branches of case-split rules in the original cyclic proof, and these rules are constructed to be non-overlapping (e.g., $x=0$ and $x=s(y)$), the resulting arithmetic formulas $\Phi_i$ for different case-linear derivations will also be disjoint. 

However, these conditions might not be \textit{complete}; they only cover the specific paths present in the original proof. To satisfy the proof schema formalism's requirement for a total partition of the parameter space, we will later introduce an \texttt{else} case to handle any missing assignments.

\begin{example}\label{ex.normalize_guards_results}
Applying \textsc{NormalizeGuards} to the results of Example~\ref{ex.collect_guards_even_odd} yields the following arithmetic conditions relative to the corresponding local parameter sets:
\begin{itemize}
    \item \textbf{ $T_{1,1}$}: $\mathcal{G} = \{x = 0\} \implies \Phi_{1,1} \equiv (x = 0)$.
    \item \textbf{ $T_{1,2}$}: $\mathcal{G} = \{x = s(y)\} \implies \Phi_{1,2} \equiv (x > 0)$.
    \item \textbf{ $T_{2}$}: $\mathcal{G} = \{x = s(y)\} \implies \Phi_{2} \equiv (x > 0)$.
\end{itemize}

These results directly determine the partitions for the proof schema definitions:
\begin{itemize}
    \item The derivations $T_{1,1}$ and $T_{1,2}$ correspond to the same proof symbol $\rho_1$. Their respective conditions, $(x=0)$ and $(x>0)$, already form a complete partition of the natural numbers. Thus, the definition $D(\rho_1)$ is later composed of these two exhaustive components.
    
    \item The derivation $T_2$ corresponds to the proof symbol $\rho_2$. Since its condition $(x > 0)$ does not form a complete partition on its own, an \texttt{else} branch will be added to $D(\rho_2)$ to satisfy the partition requirements of the proof schema formalism.
    
    \item For the derivation $T_0$ corresponding to $\rho_0$, no case-split rules were encountered earlier. 
    Consequently, this definition consists of a single component with no restriction on the parameters.
\end{itemize}
\end{example}

\begin{algorithm}[H]
\caption{\textsc{normalizeGuards}($\mathcal{G}, \mathcal{N}_i$)}\label{alg.normalize_guards}
\begin{algorithmic}[1]
\Require Set of equations $\mathcal{G}$ without $p$, local parameter set $\mathcal{N}_i$
\Ensure A formula $\Phi$

\State $\mathcal{E} \gets \text{Apply all possible substitutions for } v \notin \mathcal{N}_i$
\State $\Phi \gets \text{True}$

\ForAll{equation $(L = R) \in \mathcal{E}$}
    \State $(v_L, n_L) \gets \textsc{ExtractDepth}(L)$ \Comment{e.g., $s(s(x)) \to (x, 2)$}
    \State $(v_R, n_R) \gets \textsc{ExtractDepth}(R)$ \Comment{e.g., $s(0) \to (0, 1)$}
    
    \If{$v_L \in \mathcal{N}_i$ \textbf{and} $v_R \in \mathcal{N}_i$}
        \State $\Phi \gets \Phi \land \textsc{Align}(v_L, n_L, v_R, n_R)$
        \Comment{e.g., $s(s(y)) = s(z) \to s(y) = z$}
    \ElsIf{$v_L \in \mathcal{N}_i$ \textbf{and} $v_R = 0$}
        \State $\Phi \gets \Phi \land (v_L = n_R - n_L)$
    \ElsIf{$v_L \in \mathcal{N}_i$ \textbf{and} $v_R \notin \mathcal{N}_i \cup \{0\}$}
        \State $\Phi \gets \Phi \land (v_L > n_L + n_R - 1)$
        \Comment{Inequality case}
    \EndIf
\EndFor

\State \Return $\Phi$
\end{algorithmic}
\end{algorithm}


\paragraph{\textsc{extractSubstitutions}}
While \textsc{NormalizeGuards} focuses on the conditions for the partition, the procedure \textsc{extractSubstitutions} prepares the actual transformations required for the s-proofs. The goal is to produce a set of equations where each auxiliary variable $y$ (introduced by case-split rules) is isolated on the left-hand side, and the right-hand side consists strictly of terms constructed from constants, parameters from the local parameter set $\mathcal{N}_i$, and the operators $\{s, p\}$.

These substitutions are essential for the subsequent construction of the proof schema, as they allow us to eliminate all auxiliary variables from the case-linear derivations by replacing them with terms constructed from the local parameters in $\mathcal{N}_i$ (e.g., for guards $x=sx' , x'=sx''$ where only $x$ is a parameter the algorithm would output $\{x' \xleftarrow{} px, x'' \xleftarrow{} ppx \}$). 
The operation \textsc{SolveFor} is a syntactic inversion over successor and predecessor terms.

\begin{example}\label{ex:extractSubstitutions}
Applying \textsc{ExtractSubstitutions} to the derivations $T_{1,2}$ and $T_{2}$, where $\mathcal{G}=\{x=s(y)\}$ and $x$ belongs to the local parameter set of the corresponding proof symbol, yields the substitution
\[
\{y \xleftarrow{} p(x)\}
\]
for both derivations.

The auxiliary variable $y$ is thereby expressed entirely in terms of the local parameter $x$ and the predecessor function $p$. Consequently, all occurrences of $y$ in the corresponding case-linear derivations can be replaced by $p(x)$ during the construction of the proof schema.
\end{example}

\begin{algorithm}[H]
\caption{\textsc{ExtractSubstitutions}($\mathcal{G}, \mathcal{N}_i$)}\label{alg.extract_substitutions}
\begin{algorithmic}[1]
\Require Set of equations $\mathcal{G}$ (guards), local parameter set $\mathcal{N}_i$
\Ensure A set of equalities $\mathcal{E}$ expressing auxiliary variables in terms of $\mathcal{N}_i \cup \{s, p, 0\}$

\State $\mathcal{E} \gets \emptyset$
\State $Defs \gets \text{empty map}$

\Statex \Comment{\textbf{Step 1: Isolation of variables}}
\ForAll{equation $eq \in \mathcal{G}$}
    \State $(L, R) \gets \textsc{IsolateLHS}(eq)$ \Comment{e.g., $s(y)=x \to y=p(x)$}
    \State $Defs[L] \gets R$
\EndFor
\State $\mathcal{E} \gets \{ (v = expr) \mid Defs[v] = expr \}$

\Statex \Comment{\textbf{Step 2: Re-orientation}}
\Repeat
    \State $changed \gets \text{False}$
    \ForAll{$(v = expr) \in \mathcal{E}$}
        \If{$v \in \mathcal{N}_i$ \textbf{and} $expr$ contains some $v' \notin \mathcal{N}_i$}
            \State Remove $(v = expr)$ from $\mathcal{E}$
            \State $(v', expr') \gets \textsc{SolveFor}(v, expr, v')$ \Comment{e.g., $x=s(y) \to y=p(x)$}
            \State Add $(v' = expr')$ to $\mathcal{E}$
            \State $changed \gets \text{True}$
        \EndIf
    \EndFor
\Until{$\neg changed$}

\Statex \Comment{\textbf{Step 3: Dependency Resolution}}
\Repeat
    \State $substituted \gets \text{False}$
    \ForAll{$(v = expr) \in \mathcal{E}$}
        \ForAll{$(v_{def} = expr_{def}) \in \mathcal{E}$ where $v \neq v_{def}$}
            \If{$v_{def}$ occurs in $expr$}
                \State Replace all occurrences of $v_{def}$ in $expr$ with $(expr_{def})$
                \State Simplify $expr$ \Comment{Apply reductions like $s(p(x)) \to x$}
                \State $substituted \gets \text{True}$
            \EndIf
        \EndFor
    \EndFor
        
    \ForAll{$v$ appearing on LHS in multiple equations in $\mathcal{E}$}
        \State Keep only one definition
        \Comment{Which one does not matter}
    \EndFor
\Until{$\neg substituted$}

\Statex \Comment{\textbf{Step 4: Final Filtering}}
\State \Return $\{ (v = expr) \in \mathcal{E} \mid v \notin \mathcal{N}_i \}$
\Comment{Return only definitions for auxiliary variables}
\end{algorithmic}
\end{algorithm}



\paragraph{Macro-rules} For clarity in the following transformations, we introduce two macro-rules. Both macro-rules are used inside an $s$-proof defined under a condition $C$. Recall that, under such a condition, we may use as axioms all atomic sequents that are semantically entailed by $C$ under the standard interpretation; see Definition~\ref{def.semi-proof-schema}. Thus, whenever \[ C \models_{\mathcal M} D \] for an atomic formula $D$, we may use the sequent $\vdash D$ as a $C$-axiom. We denote such an axiom by $\Ax_C(D)$. If $D$ is non-atomic, there exists a simple propositional derivation of $\vdash D$ from atomic sequents $Ax_C (D): \{ S_1, \dots , S_n  \}$ such that $C \models_{\mathcal M} S_i$.

\medskip
\noindent
With this convention in place, we may use the following derived rules as abbreviations for the corresponding $\LK$-derivations.

\medskip
\medskip
\noindent
\begin{minipage}[t]{0.42\textwidth}
\text{$Ax_r^{C,D}:$}
\small
\begin{prooftree}
    \AxiomC{$\Gamma \vdash A,\Delta$}
    \RightLabel{$Ax_r^{C,D}$}
    \UnaryInfC{$\Gamma,\, D \to (A\to B) \vdash B,\Delta$}
\end{prooftree}
\end{minipage}
\hfill
\begin{minipage}[t]{0.56\textwidth}
\text{is defined by:}

\small
\begin{prooftree}
    \AxiomC{$\Ax_C(D)$}
    \RightLabel{$*$}
    \UnaryInfC{$\vdash D$}

    \AxiomC{$\Gamma \vdash A,\Delta$}

    \AxiomC{$A \vdash A$}
    \AxiomC{$B \vdash B$}
    \RightLabel{$\to_l$}
    \BinaryInfC{$A\to B,\, A \vdash B$}

    \RightLabel{$cut$}
    \BinaryInfC{$\Gamma,\, A\to B \vdash B,\Delta$}

    \RightLabel{$\to_l$}
    \BinaryInfC{$\Gamma,\, D \to (A\to B) \vdash B,\Delta$}
\end{prooftree}
\end{minipage}

\medskip\medskip

\noindent
\begin{minipage}[t]{0.42\textwidth}
\text{$Ax_l^{C,D}:$}
\small
\begin{prooftree}
    \AxiomC{$\Gamma,\, A \vdash \Delta$}
    \RightLabel{$Ax_l^{C,D}$}
    \UnaryInfC{$\Gamma,\, D \to (B\to A),\, B \vdash \Delta$}
\end{prooftree}
\end{minipage}
\hfill
\begin{minipage}[t]{0.56\textwidth}
\text{is defined by:}

\small
\begin{prooftree}
    \AxiomC{$\Ax_C(D)$}
    \RightLabel{$*$}
    \UnaryInfC{$\vdash D$}

    \AxiomC{$B \vdash B$}

    \AxiomC{$\Gamma,\, A \vdash \Delta$}

    \RightLabel{$\to_l$}
    \BinaryInfC{$\Gamma,\, B\to A,\, B \vdash \Delta$}

    \RightLabel{$\to_l$}
    \BinaryInfC{$\Gamma,\, D \to (B\to A),\, B \vdash \Delta$}
\end{prooftree}
\end{minipage}

\medskip

\noindent
We write $Ax_r^{D}$ and $Ax_l^{D}$ if $C$ is clear from the context.

\paragraph{Transformation of Subderivations}
The procedure \textsc{transformSubderivation} takes a case-linear derivation and converts it into an s-proof. Its purpose is to simulate all remaining (unary) case-split rules, unfolding rules for inductive predicates on the right-hand side, and the remaining companion and bud nodes under the condition $C$ produced by \textsc{normalizeGuards}. This process yields an \emph{s-proof} whose initial sequents are restricted to regular axioms, C-axioms and proof terms only. For an input derivation with end-sequent $\Gamma \vdash \Delta$, the procedure outputs an s-proof with the end-sequent $\mathit{pdef}, \Gamma \vdash \Delta$, where $\mathit{pdef}$ contains the necessary axioms to simulate the inductive definitions.
We will go through the major transformation steps before giving the algorithm.

\paragraph{(Unary) Split-Case Rules}
In this intermediary stage, split-case rules are unary, representing a single specific branch of the original derivation. In the original cyclic derivation, such a split-case rule is always followed by
applications of the equality rule $(=)$ and weakening $(w)$.
The algorithm simulates this entire block by discarding the rules $(=)$ and $(w)$. We distinguish two scenarios based on whether the premise has case-descendants or not.

\medskip
\noindent
\textbf{Case 1: Premise contains no case-descendants.}
When no case-descendants exist (can be considered a base case), the algorithm substitutes the guard (e.g., $x=0$) directly into the conclusion of this rule. This is semantically valid because this branch is only invoked when the condition associated with it is satisfied.

\begin{example}
Consider the branch of $T_{1,1}$ from Example~\ref{ex.split_cases_application} where the guard is $x=0$. The algorithm replaces the case-split, equality and weakening rules with a direct substitution  into the sequent:
\begin{center}
\begin{minipage}{0.45\textwidth}
\centering \textit{Intermediary Input}
\begin{prooftree}
    \AxiomC{$\vdots$}
    \UnaryInfC{$\vdash N0$}
    \RightLabel{\scriptsize $w$}
    \UnaryInfC{$x= 0 \vdash N0$}
    \RightLabel{\scriptsize $=_l$}
    \UnaryInfC{$x = 0 \vdash Nx$}
    \RightLabel{\scriptsize $(\mathsf{Case}\, E)$}
    \UnaryInfC{$Ex \vdash Nx$}
\end{prooftree}
\end{minipage}
$\to$
\begin{minipage}{0.45\textwidth}
\centering \textit{Transformed Output}
\begin{prooftree}
    \AxiomC{$\vdots$}
    \UnaryInfC{$\mathit{pdef} \vdash N0$}
    \RightLabel{\scriptsize $w$}
    \UnaryInfC{$\mathit{pdef}, E0 \vdash N0$}
\end{prooftree}
\end{minipage}
\end{center}

\end{example}

\medskip
\noindent

\medskip
\noindent
\textbf{Case 2: The premise contains case-descendants.}
In this scenario, the premise takes the form:
\[
\Gamma,\, \mathbf{u} = \mathbf{t}(\mathbf{y}),\, Q_1(\mathbf{u_1}(\mathbf{y})), \dots, Q_h(\mathbf{u_h}(\mathbf{y})),\, P_{j_1}(\mathbf{t_1}(\mathbf{y})), \dots, P_{j_m}(\mathbf{t_m}(\mathbf{y})) \;\vdash\; \Delta,
\]
where $P_{j_1}(\mathbf{t_1}(\mathbf{y})), \dots, P_{j_m}(\mathbf{t_m}(\mathbf{y}))$ are the \emph{case-descendants}. As before, this rule is followed by applications of the equality rule $(=_l)$ and weakening $(w)$. 

The algorithm simulates this entire block by discarding $(=_l)$ and weakening $(w)$ and replacing them with logical inferences utilizing $\mathit{pdef}$. Starting from the conclusion of the split-case rule, it applies $\forall_l$ to the corresponding predicate definition in $\mathit{pdef}$, instantiating the universally quantified variables with the terms $\mathbf{u}$ from the guard. This is followed by an application of $Ax_l^{C,D}$.

\begin{example}
Consider the derivation $T_{1,2}$ from Example~\ref{ex.split_cases_application} for $x = s(y)$, where $Oy$ is a case-descendant. The transformation yields:
\begin{center}
\begin{minipage}{0.45\textwidth}
\centering \textit{Intermediary Input}
\begin{prooftree}
    \AxiomC{$Ox \vdash Nx$ ($a$)}
    \RightLabel{\scriptsize$\mathsf{subst}$}
    \UnaryInfC{$Oy \vdash Ny$}
    \RightLabel{$N_{2r}$}
    \UnaryInfC{$Oy \vdash Ns(y)$}
    \RightLabel{\scriptsize $w$}
    \UnaryInfC{$x= s(y), Oy \vdash Ns(y)$}
    \RightLabel{\scriptsize $=_l$}
    \UnaryInfC{$x = s(y), Oy \vdash Nx$}
    \RightLabel{\scriptsize $(\mathsf{Case}\, E)$}
    \UnaryInfC{$Ex \vdash Nx$}
\end{prooftree}
\end{minipage}
$\to$
\begin{minipage}{0.45\textwidth}
\centering \textit{Transformed Output}
\begin{prooftree}
       \AxiomC{$Ox \vdash Nx$ ($a$)}
    \RightLabel{\scriptsize$\mathsf{subst}$}
    \UnaryInfC{$Oy \vdash Ny$}
    \RightLabel{$N_{2r}$}
    \UnaryInfC{$\mathit{pdef}, Opx \vdash Nx$}
    \RightLabel{\scriptsize $Ax_l^{x \geq \bar{1}}$}
    \UnaryInfC{$\mathit{pdef},  (x \geq \bar{1} \to (Ex \to Opx)), Ex \vdash Nx$}
    \RightLabel{\scriptsize $\forall_l, c_l$}
    \UnaryInfC{$\mathit{pdef}, Ex \vdash Nx$}
\end{prooftree}
\end{minipage}
\end{center}

\noindent
Note that the transformed output has incorrect derivations at this stage.

\end{example}

\paragraph{Unfold-Right Rules}

An unfolding rule applied to a formula on the right-hand side of a sequent is simulated by applying $\forall_l$ to the corresponding formula in $\mathit{pdef}$, followed by an application of the $Ax_r^{C,D}$ macro.

\begin{example}\label{ex:unfoldTransformation}
Consider the intermediary transformation from the example before:
\begin{center}
\begin{minipage}{0.45\textwidth}
\centering \textit{Intermediary Input}
\begin{prooftree}
       \AxiomC{$Ox \vdash Nx$ ($a$)}
    \RightLabel{\scriptsize$\mathsf{subst}$}
    \UnaryInfC{$Oy \vdash Ny$}
    \RightLabel{$N_{2r}$}
    \UnaryInfC{$\mathit{pdef}, Opx \vdash Nx$}
    \RightLabel{\scriptsize $Ax_l^{x \geq \bar{1}}$}
    \UnaryInfC{$\mathit{pdef},  (x \geq \bar{1} \to (Ex \to Opx)), Ex \vdash Nx$}
    \RightLabel{\scriptsize $\forall_l, c_l$}
    \UnaryInfC{$\mathit{pdef}, Ex \vdash Nx$}
\end{prooftree}
\end{minipage}
$\to$
\begin{minipage}{0.45\textwidth}
\centering \textit{Transformed Output}
\begin{prooftree}
       \AxiomC{$Ox \vdash Nx$ ($a$)}
    \RightLabel{\scriptsize$\mathsf{subst}$}
    \UnaryInfC{$Oy \vdash Npx$}
    \RightLabel{\scriptsize$Ax_r^{x \geq \bar{1}}$}
    \UnaryInfC{$\mathit{pdef}, (x \geq \bar{1} \to (Npx \to Nx)), Opx \vdash Nx$}
    \RightLabel{\scriptsize $\forall_l, c_l$}
    \UnaryInfC{$\mathit{pdef}, Opx \vdash Nx$}
    \RightLabel{\scriptsize $Ax_l^{x \geq \bar{1}}$}
    \UnaryInfC{$\mathit{pdef},  (x \geq \bar{1} \to (Ex \to Opx)), Ex \vdash Nx$}
    \RightLabel{\scriptsize $\forall_l, c_l$}
    \UnaryInfC{$\mathit{pdef}, Ex \vdash Nx$}
\end{prooftree}
\end{minipage}
\end{center}

\noindent
Note that the transformed output has incorrect derivations at this stage.

\end{example}


\paragraph{Companion and Bud Nodes}
At this stage, companion and bud nodes appear exclusively as the leaves of the case-linear derivations. Both are transformed into proof terms which references to proof symbols defined in the schema.

\begin{itemize}
    \item \textbf{Companion Nodes:} Let $C$ be a companion node. Recall that \textit{splitAtCompanions} provides a mapping from proof symbols $\rho_i$ to their respective derivations $T_i$. For every companion node in the original proof, there exists a corresponding proof symbol $\rho_i$. We replace the companion node with a substituted proof term $\rho_i(\mathcal{N}_i)\sigma$, where $\mathcal{N}_i$ is the parameter set associated with the proof symbol $\rho_i$ and $\sigma$ is the substitution $subst_i^{C_j}$ extracted by \textsc{extractSubstitutions} to eliminate auxiliary variables.
    \begin{prooftree}
        \AxiomC{$\rho_i(\mathcal{N}_i)\sigma$}
        \RightLabel{\scriptsize $\rho_i\sigma$}
        \UnaryInfC{$C\sigma$}
    \end{prooftree}

\end{itemize}
    \begin{example}
Consider $T_0$ from Example~\ref{ex.split_application}. The proof symbols $\rho_1$ and $\rho_2$ both have local parameter set $\{x\}$. Since both leaves of $T_0$ were companion nodes in the original derivation, the transformation converts them into the proof terms $\rho_1(x)$ and $\rho_2(x)$:
\begin{center}
    \small
    \begin{minipage}{0.45\textwidth}
    \centering \textit{Intermediary Input} \\
    \begin{prooftree}
        \AxiomC{$Ex \vdash Nx$}
        \AxiomC{$Ox \vdash Nx$}
        \RightLabel{\scriptsize $\lor_l$}
        \BinaryInfC{$Ex \lor Ox \vdash Nx$}
    \end{prooftree}
    \end{minipage}
    $\to$
    \begin{minipage}{0.45\textwidth}
    \centering \textit{Transformed Output} \\
    \begin{prooftree}
       \AxiomC{$\rho_1 (x) $}
    \RightLabel{$\rho_1$}
    \UnaryInfC{$Ex \vdash Nx$ }

    \AxiomC{$\rho_2 (x) $}
    \RightLabel{$\rho_2$}
    \UnaryInfC{$Ox \vdash Nx$ }

    \RightLabel{$\lor_l$}
    \BinaryInfC{$Ex \lor Ox \vdash Nx$}
    \end{prooftree}
    \end{minipage}
\end{center}
In this instance, no explicit auxiliary substitutions $\sigma$ are required because no case-split rules preceded these specific companions.
\end{example}
   \begin{itemize}
   
    \item \textbf{Bud Nodes:} Bud nodes typically serve as premises to substitution rules ($\mathsf{subst}$) with an associated substitution $\Theta$. We discard the $\mathsf{subst}$ rule but preserve $\Theta$. We then identify the target proof symbol $\rho_i$ in the dictionary $T$ and append the proof term as we did in the case of companion nodes. Crucially, we apply the auxiliary variable substitution $\sigma$ (from \textsc{extractSubstitutions}) to both the substitution $\Theta$ and the conclusion $B'$ of the original $\mathsf{subst}$ rule. This ensures that all auxiliary variables are eliminated. For a bud $B$ that is the premise of a $\mathsf{subst}$ rule with substitution $\Theta$ and conclusion $B'$, we construct:
    \begin{prooftree}
        \AxiomC{$\rho_i(\mathcal{N}_i)\Theta\sigma$}
        \RightLabel{\scriptsize $\rho_i\Theta\sigma$}
        \UnaryInfC{$B'\sigma$}
    \end{prooftree}

\end{itemize}

\begin{example}
Consider the intermediary derivation from Example~\ref{ex:unfoldTransformation} where the leaf is a bud node $B$. Let $\Theta$ be the substitution $\{x \mapsto y\}$ from the original $\mathsf{subst}$ rule, and $\sigma = \{y \mapsto p(x)\}$ be the substitution from \textsc{extractSubstitutions}:
\begin{prooftree}
    \AxiomC{$\rho_2(x)\Theta\sigma$}
    \RightLabel{\scriptsize $\rho_2\Theta\sigma$}
    \UnaryInfC{$Opx \vdash Npx$}
    \RightLabel{\scriptsize $Ax_r^{x \geq \bar{1}}$}
    \UnaryInfC{$\mathit{pdef},  (x \geq \bar{1} \to (Npx \to Nx)), Opx \vdash Nx$}
    \RightLabel{\scriptsize $\forall_l, c_l$}
    \UnaryInfC{$\mathit{pdef}, Opx \vdash Nx$}
    \RightLabel{\scriptsize $Ax_l^{x \geq \bar{1}}$}
    \UnaryInfC{$\mathit{pdef},  (x \geq \bar{1} \to (Ex \to Opx)), Ex \vdash Nx$}
    \RightLabel{\scriptsize $\forall_l, c_l$}
    \UnaryInfC{$\mathit{pdef}, Ex \vdash Nx$}
\end{prooftree}

\end{example}

\paragraph{Other rules}
For every other rule encountered during the traversal, the rule itself is maintained in the final derivation. However, to ensure consistency across the s-proof, the substitution $\sigma$ (from \textsc{extractSubstitutions}) is applied to both the conclusion and all premises of the rule. This step is vital to eliminate any auxiliary variables that may have been introduced in the original proof's context.

\begin{algorithm}[H]
\caption{\textsc{transformSubderivation}($T, C, \sigma, \mathit{pdef}, \mathcal{R}, \mathcal{T}, \mathcal{N}_{\mathrm{loc}}$)}\label{alg.transform_subderivation}
\begin{algorithmic}[1]
\Require Case-linear derivation $T$, a condition $C$, substitution set $\sigma$, definitions \pdef, bud-to-companion mapping $\mathcal{R}$, dictionary of derivations and proof symbols $\mathcal{T}$, parameter dictionary $\mathcal{N}_{\mathrm{loc}}$
\Ensure s-proof $S$

\State $N \gets \text{root node of } T$ \Comment{The node currently being visited}

\Statex \Comment{\textbf{Step 1: Check Node Type (Leaves vs. Rules)}}
\If{$N$ is a \textbf{Companion Node}} 
    \State $\rho \gets \mathcal{R}[N]$
    \State $\mathcal{N}_{\rho} \gets \mathcal{N}_{\mathrm{loc}}[\rho]$
    \State \Return \Call{BuildProofLink}{$\rho, \mathcal{N}_{\rho}, \sigma$}
\ElsIf{$N$ is a \textbf{Bud Node}} 
    \State $\Theta \gets \text{local substitution associated with } N$
    \State $\rho \gets \text{target proof symbol for } N$
    \State $\mathcal{N}_{\rho} \gets \mathcal{N}_{\mathrm{loc}}[\rho]$
    \State \Return \Call{BuildProofLink}{$\rho, \mathcal{N}_{\rho}, \Theta\sigma$}
\EndIf

\State $R \gets \text{inference rule applied at } N$ \Comment{Identify rule for internal nodes}

\If{$R$ is a \textbf{Case-Split}}
    \If{$R$ has \textbf{no case-descendants}} 
        \State $S \gets$ \Call{ApplyDirectSubstitution}{$\text{conclusion}(N), \text{guard}(R)$}
    \Else \Comment{Unfold-Left via $Ax_l^{C,D}$}
        \State $S \gets$ \Call{SimulateLeftUnfold}{$\text{conclusion}(N), \mathit{pdef}, \text{guard}(R)$}
    \EndIf
\ElsIf{$R$ is an \textbf{Unfold-Right} rule} \Comment{Unfold-Right via $Ax_r^{C,D}$}
    \State $S \gets$ \Call{SimulateRightUnfold}{$\text{conclusion}(N), \mathit{pdef}$}
\Else \Comment{Application of $\sigma$ to standard logical rules}
    \State $S \gets$ \Call{MaintainRule}{$R, \sigma$}
\EndIf

\Statex \Comment{\textbf{Step 2: Recursion}}
\ForAll{premises $P_j$ of $N$} \Comment{Recurse on the sub-derivations above node $N$}
    \State $S_{child} \gets$ \Call{transformSubderivation}{$P_j, C, \sigma, \mathit{pdef}, \mathcal{R}, \mathcal{T}, \mathcal{N}_{\mathrm{loc}}$}
    \State \Call{AttachTo}{$S, S_{child}$}
\EndFor

\State \Return $S$
\end{algorithmic}
\end{algorithm}

\paragraph{\textsc{generateNegationProof}}

As mentioned previously, the conditions in a definition $D(\rho_i)$ for a proof symbol $\rho_i$ might not form a complete partition. In such cases, the algorithm utilizes the respective negation axiom from \pdef to construct a derivation for these cases.

\begin{example}
Consider the definition $D(\rho_2)$ for the proof symbol $\rho_2$, where the extracted conditions do not form a complete partition. Specifically, the case $x=0$ is missing from the partition over the local parameter set of $\rho_2$. To have a total partition, we must construct a derivation for the end-sequent $\mathit{pdef}, Ox \vdash Nx$ utilizing the negation axiom for $O$ from Example~\ref{ex.pdef_translation}: $\forall x (x=0 \to \neg Ox)$.

The procedure \textsc{generateNegationProof} produces the following derivation to close this branch:
\begin{center}
\begin{prooftree}
    \AxiomC{}   
    \UnaryInfC{$\vdash x=0$}

    \AxiomC{$ Ox \vdash Ox$}
    \RightLabel{\scriptsize$w_l , w_r$}
    \UnaryInfC{$\mathit{pdef},  Ox \vdash Ox, Nx$}
    \RightLabel{\scriptsize $\neg_l$}
    \UnaryInfC{$\mathit{pdef}, \neg Ox, Ox \vdash Nx$}
    
    \RightLabel{\scriptsize $\to_l$}
    \BinaryInfC{$\mathit{pdef},  (x=0 \to \neg Ox), Ox \vdash Nx$}
    \RightLabel{\scriptsize $\forall_l$}
    \UnaryInfC{$\mathit{pdef},\forall x (x=0 \to \neg Ox), Ox \vdash Nx$}
    \RightLabel{\scriptsize $c_l$}
    \UnaryInfC{$\mathit{pdef}, Ox \vdash Nx$}
\end{prooftree}
\end{center}
Note that this s-proof gets called only if the condition $x = 0$ is satisfied; hence we can use the $C$-axiom $\vdash x = 0$.

\end{example}


\subsection{Step 5: Schema Assembly}\label{sub.step5}

In the final stage, all components are aggregated into the proof schema
\[
\mathbf{P} = (\rho_0, \mathcal{R}_P, \mathcal{N}_0, t, \mathit{Seq}, \Pi_P).
\]
Here, $\mathcal{N}_0$ is the global parameter set obtained as the union of all local parameter sets. The term assignment $t$ maps each proof symbol $\rho_i$ to its corresponding local parameter set $\mathcal{N}_i$, while $\mathit{Seq}$ maps each proof symbol to its associated end-sequent.

\section{The Two-Hydra Example}\label{sec.2hydraexample}

In this section we consider the \emph{2-Hydra} example introduced by Berardi and Tatsuta
in~\cite{BT.2017} and further developed in~\cite{berardi2019}. Their result shows that the
2-Hydra statement is provable in $\mathsf{CLKID}^{\omega}$ but not in $\mathsf{LKID}$,
thereby refuting the conjectured equivalence between the two systems.

The purpose of this section is threefold. First, we translate the cyclic proof of the 2-Hydra statement into the proof schema formalism using the translation procedure developed in the previous section. Second, combining this construction with the known unprovability result for 2-Hydra in $\mathsf{LKID}$, we show that proof schemata based on point transition systems can prove inductive statements that are not provable in $\mathsf{LKID}$. Third, we extract a Herbrand system from the resulting proof schema.

\subsection{The Two-Hydra Statement}

We work in the first-order language $\Sigma_N = \{ 0, s, N, P \}$,
where $0$ is a constant, $s$ is the successor function, $N$ is an inductive predicate
for natural numbers, and $P$ is a binary predicate symbol.
We assume the standard $(0,s)$-axioms stating that $0$ is not a successor and that
the successor function is injective.
The \emph{2-Hydra} statement formalises a termination game on pairs of natural numbers,
intuitively representing the lengths of two Hydra heads.
The predicate $P(x,y)$ expresses that the game starting from heads of lengths $x$ and $y$ is winning.
Throughout this section, we abbreviate $s(0)$ by $\bar{1}$ and omit parentheses when using the predecessor and successor symbols; hence, $p(p(x))$ is abbreviated as $ppx$. The inductive predicate $N$ for natural numbers is defined as before.

\begin{definition}[2-Hydra Statement]\label{def:2HydraStatement_Cyclic}
The 2-Hydra statement $H$ is defined as
\[
(H_a \wedge H_b \wedge H_c \wedge H_d) \;\rightarrow\; \forall x,y \in N.\, P(x,y),
\]
where the components are given as follows:
\begin{align*}
(H_a)\quad & \forall x \in N.\; P(0,0) \;\wedge\; P(\bar{1},0) \;\wedge\; P(x,\bar{1}), \\
(H_b)\quad & \forall x,y \in N.\; P(x,y) \rightarrow P(sx,ssy), \\
(H_c)\quad & \forall y \in N.\; P(sy,y) \rightarrow P(0,ssy), \\
(H_d)\quad & \forall x \in N.\; P(sx,x) \rightarrow P(ssx,0).
\end{align*}
\end{definition}


Clause $(H_a)$ encodes the terminal winning configurations.
Clauses $(H_b), \dots ,(H_d)$ correspond to the transition rules of the game and ensure that
each move strictly decreases a suitable measure on head lengths.
Under the standard arithmetic axioms, for every closed pair $(n,m)$ exactly one instance
of $(H_a), \dots ,(H_d)$ applies.

\subsection{Translation of the Cyclic 2-Hydra Proof}\label{subsec:hydraTranslation}

To translate the cyclic proof of the 2-Hydra statement into the proof schema
formalism, we apply the translation procedure introduced in the previous section.
The cyclic 2-Hydra proof is formulated in terms of successor-based transition rules. In contrast, proof schemata generate recursive proof calls by means of predecessor terms. Therefore, before applying the translation procedure, we extend the language by the predecessor function symbol $p$ and reformulate the clauses of Definition~\ref{def:2HydraStatement_Cyclic} in predecessor-based form.

\begin{definition}[2-Hydra Statement for Proof Schemata]\label{def:2HydraStatement_schema}
Let $\hat{H}$ consist of the formulas
\[
P(0,0), \quad P(\bar{1},0), \quad H_a, \quad H_b, \quad H_c, \quad H_d,
\] 
where the components are given as follows:
\begin{align*}
    (H_a)\quad & \forall u.  \; ( Nu \to P(u,\bar{1})), \\
    (H_b)\quad & \forall u,v.\; ( Nu \land Nv \to ( u>0 \land v > \bar{1} \to (   P(pu,ppv) \rightarrow P(u,v)))), \\
    (H_c)\quad & \forall v.  \; (Nv \to ( v >\bar{1} \to ( P(pv,ppv) \rightarrow P(0,v)))), \\
    (H_d)\quad & \forall u.  \;  ( Nu \to ( u > \bar{1} \to  ( P(pu,ppu) \rightarrow P(u,0)))). 
\end{align*}
    
\end{definition}

Note that, since proof schemata are defined over the domain of natural numbers, the predicate $N$ is not needed explicitly. Omitting it would yield a considerably shorter proof schema while preserving the statement represented by the original cyclic proof. For illustrative purposes, however, we give a translation into the proof schema formalism that follows the procedure from the previous section.

We now apply this procedure to the cyclic 2-Hydra proof $\Pi_{Hydra}$ from~\cite{BT.2017,berardi2019}. The presentation follows the five steps of the algorithm and records those transformations that determine the resulting proof schema.

\paragraph{Step~1: Initialization}
The inductive predicate axioms for $N$ are
\[
\mathit{pdef} := \{\, N0, \ \forall x.\, (x \ge \bar{1} \rightarrow (Nx \rightarrow Npx)),\ 
\forall x.\, (x \ge \bar{1} \rightarrow (Npx \rightarrow Nx)) \,\}.
\]
The cyclic proof yields the parameter dictionary $\mathcal{N}_{\mathrm{loc}}=\{\rho \mapsto \{x,y\}\}$, and a single companion node, namely the end-sequent $\hat{H}, Nx, Ny \;\vdash\; P(x,y)$.  
Consequently, Step~2 of the translation procedure is vacuous in this example: since there is only one companion node, the cycle decomposition yields a single component associated with the proof symbol $\rho$.

Moreover, all inductive predicates occurring in the 2-Hydra proof have a complete characterization. Hence no negated predicate axioms have to be added to $\mathit{pdef}$, and the calculus need not be extended through additional definitions for this example.

\paragraph{Step~3: Case Partitioning}
$\Pi_{Hydra}$ contains five case-split rules. We partition the proof into six
case-linear derivations $T_1, \dots, T_6$.

\paragraph{Step~4: Subderivation Transformation}
The guards, conditions, and substitutions for $T_1, \dots, T_6$ are summarized in Table~\ref{tab:HydraSubderivations}.

\begin{table}[h]
\centering

\begin{tabular}{|c|l|l|l|}
\hline
Derivation & Guards & Conditions & Substitutions \\
\hline
$T_1$ & $\{y=0 , x= 0 \}$ & $C_1: y= 0 \land x= 0$& $\emptyset$\\
$T_2$ & $\{y=0 , x= sx', x' = 0\}$ &$C_2:y= 0 \land x= \bar{1}$ & $\{ x' \xleftarrow{} 0\}$\\
$T_3$ & $\{y=0 , x= sx', x' = sx''\}$ & $C_3: y= 0 \land x> \bar{1}$& $\{x' \xleftarrow{} px , x''\xleftarrow{} ppx  \}$\\
$T_4$ & $\{ y = sy', y'=0 \}$ & $C_4: y= \bar{1}$& $\{ y' \xleftarrow{} 0\}$\\
$T_5$ & $\{ y = sy' , y' = sy'' , x= 0\}$ &$C_5:y > \bar{1} \land x=0$ & $\{ y' \xleftarrow{} py , y'' \xleftarrow{} ppy \}$\\
$T_6$ & $\{ y = sy', y' = sy'' , x=sx' \}$ &$C_6 :y > \bar{1} \land x>0$ & $\{ y' \xleftarrow{} py , y'' \xleftarrow{} ppy , x' \xleftarrow{} px \}$\\
\hline
\end{tabular}
\caption{Guards, Conditions, and Substitutions extracted from case-linear subderivations.}
\label{tab:HydraSubderivations}
\end{table}

\noindent
The transformation of the subderivations yields: 

\medskip

\begin{minipage}{0.45\textwidth}
$T1$: 
\begin{prooftree}
\def\fCenter{\mbox{\ $\Rightarrow$\ }}
    \AxiomC{$P(0,0) \vdash P(0,0)$}
    \RightLabel{$w_l^*$ }

    \UnaryInfC{$\mathit{pdef}, \hat{H}, N0, N0 \vdash P(0,0)$}

\end{prooftree}
\end{minipage}\begin{minipage}{0.45\textwidth}
$T2$: 
\begin{prooftree}
\def\fCenter{\mbox{\ $\Rightarrow$\ }}
    \AxiomC{$P(\bar{1},0) \vdash P(\bar{1},0)$}
    \RightLabel{$w_l^*$ }

    \UnaryInfC{$\mathit{pdef}, \hat{H}, N\bar{1}, N0 \vdash P(\bar{1},0)$}

\end{prooftree}
\end{minipage}

\medskip
$T3$: 
\small
\begin{prooftree}
\def\fCenter{\mbox{\ $\Rightarrow$\ }}

\AxiomC{$Nx \vdash Nx$}

    \AxiomC{$\rho (px,ppx)$}
    \RightLabel{$\rho \{x \xleftarrow{} px , y \xleftarrow{} ppx \}$}
    \UnaryInfC{$\mathit{pdef}, \hat{H}, Npx, Nppx \vdash P(px,ppx)$}

    \RightLabel{$Ax_r^{x > \bar{1}}$}
    \UnaryInfC{$\mathit{pdef}, \hat{H},Npx,  Nppx, (x > \bar{1} \to (P(px ,ppx) \to P(x,0)))  \vdash P(x,0)$}

    \RightLabel{$\to_l$}

    \BinaryInfC{$\mathit{pdef}, \hat{H}, Nx,Npx,  Nppx, ( Nx \to (x > \bar{1} \to (P(px ,ppx) \to P(x,0))))  \vdash P(x,0)$}
    \RightLabel{$\forall_l$}
    
    \UnaryInfC{$\mathit{pdef}, \hat{H}, Nx, Npx,  Nppx, \forall u. ( Nu \to ( u> \bar{1} \to (P(pu ,ppu) \to P(u,0 )))) \vdash P(x,0)$}
    \RightLabel{$c_l$}

     \UnaryInfC{$\mathit{pdef}, \hat{H},Nx, Npx,  Nppx  \vdash P(x,0)$}
     \RightLabel{$Ax_l^{px \geq \bar{1}}$}
     \UnaryInfC{$\mathit{pdef}, \hat{H},Nx, Npx, Npx, (px \geq \bar{1}) \to ( Npx \to Nppx) \vdash P(x,0)$}
     \RightLabel{$\forall_l$}
     \UnaryInfC{$\mathit{pdef}, \hat{H},Nx, Npx, Npx, \forall w. ( w\geq \bar{1} \to (Nw \to Npw)) \vdash P(x,0)$}
     \RightLabel{$c_l$}
     \UnaryInfC{$\mathit{pdef}, \hat{H},Nx, Npx , Npx \vdash P(x,0)$}
     \RightLabel{$c_l$}
     \UnaryInfC{$\mathit{pdef}, \hat{H}, Nx,Npx \vdash P(x,0)$}
     \RightLabel{$Ax_l^{x \geq \bar{1}}$}
     \UnaryInfC{$\mathit{pdef}, \hat{H}, Nx, Nx, (x \geq \bar{1}) \to ( Nx \to Npx) \vdash P(x,0)$}
     \RightLabel{$c_l$}

     \UnaryInfC{$\mathit{pdef}, \hat{H}, Nx, (x \geq \bar{1}) \to ( Nx \to Npx) \vdash P(x,0)$}     
     \RightLabel{$\forall_l$}
     \UnaryInfC{$\mathit{pdef}, \hat{H}, Nx,  \forall w.( w\geq \bar{1} \to (Nw \to Npw)) \vdash P(x,0)$}
     \RightLabel{$c_l$}

     \UnaryInfC{$\mathit{pdef}, \hat{H}, Nx \vdash P(x,0)$}
     \RightLabel{$w_l$}
    \UnaryInfC{$\mathit{pdef}, \hat{H}, Nx, N0 \vdash P(x,0)$}

\end{prooftree}

\noindent $T_3$ is valid under $C_3\colon y=0 \land x>\bar{1}$. Since 
\[ C_3 \models_{\mathcal M} x>\bar{1}, \qquad C_3 \models_{\mathcal M} x \geq \bar{1}, \qquad\text{and}\qquad C_3 \models_{\mathcal M} px \geq \bar{1}, \] the sequents $\vdash x>\bar{1}$, $\vdash x\geq \bar{1}$, and $\vdash px\geq \bar{1}$ are available as $C_3$-axioms.

\medskip
$T4$: 
\begin{prooftree}
\def\fCenter{\mbox{\ $\Rightarrow$\ }}

\AxiomC{$Nx \vdash Nx$}

    \AxiomC{$P(x,\bar{1}) \vdash P(x,\bar{1})$}
    \RightLabel{$w_l^*$ }
    \UnaryInfC{$\mathit{pdef}, \hat{H}, N\bar{1}, P(x,\bar{1}) \vdash P(x,\bar{1})$}

    \RightLabel{$\to_l$}
    \BinaryInfC{$\mathit{pdef}, \hat{H},Nx, N\bar{1}, (Nx \to P(x,\bar{1})) \vdash P(x,\bar{1})$}
    \RightLabel{$\forall_l$}
    \UnaryInfC{$\mathit{pdef}, \hat{H},Nx, N\bar{1}, \forall u. (Nu \to  P(u,\bar{1})) \vdash P(x,\bar{1})$}
    \RightLabel{$c_l$}
    \UnaryInfC{$\mathit{pdef}, \hat{H}, Nx, N\bar{1} \vdash P(x,\bar{1})$}

\end{prooftree}
\newpage
\medskip
$T5$: 
\begin{prooftree}
\def\fCenter{\mbox{\ $\Rightarrow$\ }}

    \AxiomC{$Ny \vdash Ny$}

    \AxiomC{$\rho (py,ppy)$}
    \RightLabel{$\rho \{x \xleftarrow{} py , y \xleftarrow{} ppy \}$}
    \UnaryInfC{$\mathit{pdef},\hat{H}, Npy, Nppy \vdash P(py,ppy)$}
    \RightLabel{$Ax_r^{y > \bar{1}}$}
    \UnaryInfC{$\mathit{pdef},\hat{H},Npy, Nppy, (y> \bar{1} ) \to (P(py ,ppy) \to P(0,y)) \vdash P(0,y)$}

    \RightLabel{$\to_l$}
    \BinaryInfC{$\mathit{pdef},\hat{H},Ny, Npy, Nppy, ( Ny \to (y> \bar{1} ) \to (P(py ,ppy) \to P(0,y))) \vdash P(0,y)$}
    \RightLabel{$\forall_l$}
    \UnaryInfC{$\mathit{pdef},\hat{H},Ny, Npy, Nppy, \forall v . ( Nv \to ( v> \bar{1} \to (P(pv ,ppv) \to P(0,v) )))\vdash P(0,y)$}
    \RightLabel{$c_l$}
    \UnaryInfC{$\mathit{pdef},\hat{H},Ny, Npy, Nppy \vdash P(0,y)$}
    \RightLabel{$Ax_l^{py \geq \bar{1}}$}
    \UnaryInfC{$\mathit{pdef},\hat{H},Ny, Npy, Npy, (py \geq \bar{1}) \to (Npy \to Nppy) \vdash P(0,y)$}
    \RightLabel{$\forall_l$}
    \UnaryInfC{$\mathit{pdef},\hat{H}, Ny,Npy, Npy, \forall w . ( w\geq \bar{1} \to (Nw \to Npw)) \vdash P(0,y)$}
    \RightLabel{$c_l$}
    \UnaryInfC{$\mathit{pdef},\hat{H},Ny, Npy, Npy \vdash P(0,y)$}
    \RightLabel{$c_l$}
    \UnaryInfC{$\mathit{pdef},\hat{H},Ny, Npy \vdash P(0,y)$}
    \RightLabel{$Ax_l^{y \geq \bar{1}}$}

    \UnaryInfC{$\mathit{pdef},\hat{H}, Ny, Ny, (y \geq \bar{1}) \to (Ny \to Npy) \vdash P(0,y)$}
    \RightLabel{$c_l$}
    \UnaryInfC{$\mathit{pdef},\hat{H}, Ny, (y \geq \bar{1}) \to (Ny \to Npy) \vdash P(0,y)$}
    \RightLabel{$\forall_l$}
    \UnaryInfC{$\mathit{pdef},\hat{H}, Ny, \forall w . ( w\geq \bar{1} \to (Nw \to Npw)) \vdash P(0,y)$}
    \RightLabel{$c_l$}
    \UnaryInfC{$\mathit{pdef},\hat{H}, Ny \vdash P(0,y)$}
    \RightLabel{$w_l$}  
    \UnaryInfC{$\mathit{pdef},\hat{H}, N0, Ny \vdash P(0,y)$}

    \end{prooftree}

    \noindent
 $T_5$ is valid under the condition $C_5\colon y>\bar{1} \land x=0$. Since \[ C_5 \models_{\mathcal M} y>\bar{1}, \qquad C_5 \models_{\mathcal M} py\geq \bar{1}, \qquad\text{and}\qquad C_5 \models_{\mathcal M} y\geq \bar{1}, \] the sequents $\vdash y>\bar{1}$, $\vdash py\geq \bar{1}$, and $\vdash y\geq \bar{1}$ may be used as $C_5$-axioms in $T_5$.

\medskip
$T6$: 

\begin{prooftree}
\def\fCenter{\mbox{\ $\Rightarrow$\ }}

\AxiomC{$Ax$}
\UnaryInfC{$\vdots$}

\UnaryInfC{$Nx, Ny \vdash Nx \land Ny$\kern-1.5em}
    \AxiomC{$\rho (px,ppy)$}
    \RightLabel{$\rho \{x \xleftarrow{} px , y \xleftarrow{} ppy \}$}
    \UnaryInfC{$\mathit{pdef}, \hat{H}, Npx, Nppy \vdash P(px,ppy)$}
    \RightLabel{$Ax_r^{x > 0 \land y > \bar{1}}$}
    \UnaryInfC{$\mathit{pdef}, \hat{H}, Npx, Nppy, ((x>0 \land y > \bar{1}) \to (P(px ,ppy) \to P(x,y)))  \vdash P(x,y)$}
    \RightLabel{$\to_l$}
    \BinaryInfC{$\mathit{pdef}, \hat{H}, Nx,Npx,Ny, Nppy, ( Nx \land Ny \to ((x>0 \land y > \bar{1}) \to (P(px ,ppy) \to P(x,y))))  \vdash P(x,y)$}
    \RightLabel{$\forall_l$ 2x}
    \UnaryInfC{$\mathit{pdef}, \hat{H}, Nx, Npx,Ny, Nppy, \forall u,v .( Nu \land Nv \to ( u>0 \land v > \bar{1} \to  ( P(pu ,ppv) \to P(u,v)))) \vdash P(x,y)$}
    \RightLabel{$c_l$} 
    \UnaryInfC{$\mathit{pdef}, \hat{H}, Nx,Npx,Ny, Nppy \vdash P(x,y)$}
    \RightLabel{$Ax_l^{py \geq \bar{1}}$}
    \UnaryInfC{$\mathit{pdef}, \hat{H},Nx, Npx,Ny, Npy , (py \geq \bar{1}) \to (Npy \to Nppy) \vdash P(x,y)$}
    \RightLabel{$\forall_l$}
    \UnaryInfC{$\mathit{pdef}, \hat{H},Nx, Npx,Ny, Npy , \forall w .( w\geq \bar{1} \to (Nw \to Npw))\vdash P(x,y)$}
    \RightLabel{$c_l$}    
    \UnaryInfC{$\mathit{pdef}, \hat{H},Nx, Npx,Ny, Npy \vdash P(x,y)$}
    \RightLabel{$Ax_l^{y \geq \bar{1}}$}
    \UnaryInfC{$\mathit{pdef}, \hat{H},Nx, Npx, Ny, Ny,(y \geq \bar{1}) \to ( Ny \to Npy) \vdash P(x,y)$}
    \RightLabel{$c_l$}
    \UnaryInfC{$\mathit{pdef}, \hat{H},Nx, Npx, Ny,(y \geq \bar{1}) \to ( Ny \to Npy) \vdash P(x,y)$}
    \RightLabel{$\forall_l$}
    \UnaryInfC{$\mathit{pdef}, \hat{H},Nx, Npx, Ny, \forall w. (w\geq \bar{1} \to ( Nw \to Npw) )\vdash P(x,y)$}
    \RightLabel{$c_l$}
    \UnaryInfC{$\mathit{pdef}, \hat{H}, Nx,Npx, Ny \vdash P(x,y)$}
    \RightLabel{$Ax_l^{x \geq \bar{1}}$}
    \UnaryInfC{$\mathit{pdef}, \hat{H}, Nx, Nx, Ny, (x \geq \bar{1}) \to ( Nx \to Npx) \vdash P(x,y)$}
\RightLabel{$c_l$}
    \UnaryInfC{$\mathit{pdef}, \hat{H}, Nx, Ny, (x \geq \bar{1}) \to ( Nx \to Npx) \vdash P(x,y)$}
    \RightLabel{$\forall_l$}
    \UnaryInfC{$\mathit{pdef}, \hat{H}, Nx, Ny, \forall w.(w\geq \bar{1} \to ( Nw \to Npw ))\vdash P(x,y)$}
    \RightLabel{$c_l$}
    \UnaryInfC{$\mathit{pdef}, \hat{H}, Nx, Ny \vdash P(x,y)$}
\end{prooftree} 

\noindent
$T_6$ is valid under the condition $C_6\colon x>0 \land y>\bar{1}$. Since
\[
C_6 \models_{\mathcal M} x>0,
\qquad
C_6 \models_{\mathcal M} y>\bar{1},
\qquad
C_6 \models_{\mathcal M} py\geq \bar{1},
\qquad
C_6 \models_{\mathcal M} y\geq \bar{1},
\qquad\text{and}\qquad
C_6 \models_{\mathcal M} x\geq \bar{1},
\]
the corresponding sequents may be used as $C_6$-axioms in $T_6$. In particular, the non-atomic formula $x>0 \land y>\bar{1}$ is handled by using the $C_6$-axioms $\vdash x>0$ and $\vdash y>\bar{1}$ as its atomic components.

\paragraph{Step 5: Schema Assembly}

The procedure returns a proof schema
\[
\mathbf{P}_{Hydra} = (\rho,\{\rho\},\mathcal N_0,t,Seq,\Pi_P),
\]
where
$\mathcal N_0=\{x,y\}$.
Since the schema contains only a single proof symbol, the term assignment $t$ therefore maps $\rho$ to the whole parameter set $\{x,y\}$, while $Seq$ maps $\rho$ to its end-sequent.
Finally, $\Pi_P$ assigns the following definition to $\rho$:


\begin{eqnarray*}
\rho(x,y) &\defeq& 
\{x = 0 \land y = 0 \colon T_1 (x,y), \\
&& x = \bar{1} \land y = 0 \colon T_2 (x,y), \\
&& x > \bar{1} \land y = 0 \colon T_3 (x,y), \\
&& y = \bar{1} \colon T_4 (x,y), \\
&& x = 0 \land y > \bar{1} \colon T_5 (x,y), \\
&& x > 0 \land y > \bar{1} \colon T_6 (x,y) \}
\end{eqnarray*}

The resulting proof schema is denoted by $\mathbf P_{Hydra}$.

\subsection{Termination of the Associated Point Transition System}

It remains to verify that $\mathbf P_{Hydra}$ satisfies the termination condition for proof schemata. For this purpose, we construct the canonical point transition system associated with $\mathbf P_{Hydra}$ and exhibit a ranking function that strictly decreases along every non-terminal transition.

The canonical point transition system associated with $\mathbf P_{Hydra}$ is
\[
\mathbf P_{Hydra}^* = \{ \delta, \Delta^*, \Delta_e, \mathbf P_{Hydra} \}.
\]
Since the schema contains only one proof symbol, we have $\delta = \{\Delta^*\}$. The set of end-states is $\Delta_e = \{\delta_e\}$, and the transition relation is given by:

\begin{eqnarray*}
\mathbf P_{Hydra}(\delta) &=& \{(\delta,(x,y)) \to \{(\delta_e,(x,y))\} \colon x=0 \land y=0,\\ 
                          & & (\delta,(x,y)) \to \{(\delta_e,(x,y))\} \colon x=\bar{1} \land y=0,\\
                         & & (\delta,(x,y)) \to \{(\delta,(p(x),pp(x)))\} \colon x>\bar{1} \land y=0,\\ 
                         & & (\delta,(x,y)) \to \{(\delta_e,(x,y))\} \colon y=\bar{1},\\
                         & & (\delta,(x,y)) \to \{(\delta,(p(y),pp(y)))\} \colon x=0 \land y>\bar{1},\\ 
                         & & (\delta,(x,y)) \to \{(\delta,(p(x),pp(y)))\} \colon x>0 \land y>\bar{1} \}.\\[1ex]
\end{eqnarray*}

To demonstrate that $\mathbf P_{Hydra}^*$ is terminating, we define a well-founded ordering $\prec$ on labeled points such that for every point transition $(\delta, p) \to \Lcal : C$ in $\mathbf P_{Hydra}$, and for every $(\delta', q) \in \Lcal$, the condition $(\delta', q) \prec (\delta, p)$ holds under any parameter assignment $\sigma$ satisfying $\sigma(C) = \top$.

\begin{definition}[Ranking Function for Points]
 We define the ranking function $\Phi: \mathbb{N}^2 \to \mathbb{N}$ as:
\begin{equation}
    \Phi(x, y) = \max(x, y)
\end{equation}
We define the ordering on labeled points such that $(\delta, q) \prec (\delta, p)$ if $\Phi(q) < \Phi(p)$. Because $(\mathbb{N}, <)$ is well-founded, $\prec$ is a well-founded ordering.
\end{definition}

\begin{theorem}\label{thm:hydraTermination}
The point transition system $\mathbf P_{Hydra}^*$ is terminating for all parameter assignments $\sigma$.
\end{theorem}

\begin{proof}
For every non-terminal transition, the value
$\Phi(x,y)=\max(x,y)$ strictly decreases.
\begin{itemize}
    \item If $x>\bar{1} \land y=0$, the successor state is $(px,ppx)$ and
\[
\max(px,ppx)=px < x=\max(x,0).
\]

\item If $x=0 \land y>\bar{1}$, the successor state is $(py,ppy)$ and
\[
\max(py,ppy)=py < y=\max(0,y).
\]

\item If $x>0 \land y>\bar{1}$, the successor state is $(px,ppy)$ and both coordinates are strictly smaller
than their predecessors, hence
\[
\max(px,ppy) < \max(x,y).
\]

\end{itemize}

Therefore every transition decreases the ranking
function, and since $(\mathbb N,<)$ is well-founded,
$\mathbf P_{Hydra}^{*}$ is terminating.

\end{proof}

\subsection{Expressive Power of Proof Schemata}

The preceding subsections constructed a proof schema for the Two-Hydra statement
and established termination of its associated point transition system. We can now
state the main consequence of this construction.

\begin{theorem}\label{thm:2hydraIsProvable}
The Two-Hydra statement is provable by a proof schema based on a
point transition system.
\end{theorem}

\begin{proof}
The construction in Subsection~\ref{subsec:hydraTranslation} yields the proof schema $\mathbf P_{Hydra}$ for the end-sequent $\hat H, Nx, Ny \vdash P(x,y)$. By Theorem~\ref{thm:hydraTermination}, the associated point transition system $\mathbf P_{Hydra}^*$ is terminating. Hence $\mathbf P_{Hydra}$ is a valid proof schema proving the Two-Hydra statement.
\end{proof}

The significance of this theorem follows from the result of Berardi and Tatsuta~\cite{BT.2017}
showing that the Two-Hydra statement is not provable in $\mathsf{LKID}$.

\begin{corollary}
Proof schemata based on point transition systems can prove inductive statements
that are not provable in $\mathsf{LKID}$.
\end{corollary}

\begin{proof}
By Theorem~\ref{thm:2hydraIsProvable}, the Two-Hydra statement is provable by a proof schema based on a point transition system. On the other hand, Berardi and Tatsuta~\cite{BT.2017} show that the Two-Hydra statement is not provable in $\mathsf{LKID}$. Hence proof schemata based on point transition systems can prove inductive statements that are not provable in $\mathsf{LKID}$.
\end{proof}

\subsection{Herbrand System Extraction}

We now provide the Herbrand system for $\mathbf P_{Hydra}$. 
For every $F \in \hat{H}$, we define $\Theta_F(x,y)$ as the schematic set of substitutions for the quantified variables of the formula $F$, depending on the parameters $x$ and $y$. The Herbrand system is given by:

%
\begin{minipage}[t]{0.45\textwidth}

\begin{eqnarray*}
\Theta_{H_a}(x,y) &=_d&  \ x=0 \land y=0\ : \ \emptyset,\\
                                  & & \ x=\bar{1} \land y=0\  : \ \emptyset,\\
                                  & &  \ x>\bar{1} \land y=0\  : \  \Theta_{H_a}(px,ppx),\\
                                   & & \ y=\bar{1}\ : \ \{\{u \ass x\}\},\\
                                   & &  \ x=0 \land y>\bar{1}\   : \  \Theta_{H_a}(py,ppy),\\
                                   & &  \ x>0 \land y>\bar{1}\  : \  \Theta_{H_a}(px,ppy).
\end{eqnarray*}
\end{minipage}
\begin{minipage}[t]{0.45\textwidth}

\begin{eqnarray*}
\Theta_{H_b}(x,y) &=_d& \ x=0 \land y=0\ : \ \emptyset,\\
                                  & & \ x=\bar{1} \land y=0\  : \ \emptyset,\\
                                  & &  \ x>\bar{1} \land y=0\  : \  \Theta_{H_b}(px,ppx),\\
                                   & & \ y=\bar{1}\ : \ \emptyset,\\
                                   & &  \ x=0 \land y>\bar{1}\  : \  \Theta_{H_b}(py,ppy),\\
                                   & &  \ x>0 \land y>\bar{1}\  : \  \{\{u \ass x,v \ass y\}\} \union \\ 
                                   & & \qquad \qquad \qquad \qquad \quad \quad \, \Theta_{H_b}(px,ppy).
\end{eqnarray*}
\end{minipage}

\begin{minipage}[t]{0.45\textwidth}
\begin{eqnarray*}
\Theta_{H_c}(x,y) &=_d& \ x=0 \land y=0\ : \ \emptyset,\\
                                  & & \ x=\bar{1} \land y=0\  : \ \emptyset,\\
                                  & &  \ x>\bar{1} \land y=0\  : \  \Theta_{H_c}(px,ppx),\\
                                   & & \ y=\bar{1}\ : \ \emptyset,\\
                                   & &  \ x=0 \land y>\bar{1}\  : \  \{\{v \ass y\}\} \union \\ 
                                   & & \qquad \qquad \qquad  \quad \, \, \Theta_{H_c}(py,ppy),\\
                                   & &  \ x>0 \land y>\bar{1}\  : \  \Theta_{H_c}(px,ppy).
\end{eqnarray*}
\end{minipage}
\begin{minipage}[t]{0.45\textwidth}

\begin{eqnarray*}
\Theta_{H_d}(x,y) &=_d& \ x=0 \land y=0\ : \ \emptyset,\\
                                  & & \ x=\bar{1} \land y=0\  : \ \emptyset,\\
                                  & &  \ x>\bar{1} \land y=0\  : \  \{\{u \ass x\}\} \union \\ 
                                  & & \qquad \qquad \qquad  \quad \, \,  \Theta_{H_d}(px,ppx),\\
                                   & & \ y=\bar{1}\ : \ \emptyset,\\
                                   & &  \ x=0 \land y>\bar{1}\  : \  \Theta_{H_d}(py,ppy),\\
                                   & &  \ x>0 \land y>\bar{1}\  : \  \Theta_{H_d}(px,ppy).
\end{eqnarray*}
\end{minipage}

\begin{example}
The Herbrand system with parameter assignment $\sigma(x) = \bar{1}$ and $\sigma(y) = \bar{5}$ yields:
\[
\begin{array}{l}
\Theta_{H_a}(\bar{1},\bar{5}) = \Theta_{H_a}(0,\bar{3}) = \Theta_{H_a}(\bar{2},\bar{1}) = \{\{u \ass \bar{2}\}\},\\[1ex]
\Theta_{H_b}(\bar{1},\bar{5}) = \{\{u \ass \bar{1}, v \ass \bar{5}\}\} \cup \Theta_{H_b}(0,\bar{3}) = 
\{\{u \ass \bar{1}, v \ass \bar{5}\}\} \cup \Theta_{H_b}(\bar{2},\bar{1}) = \{\{u \ass \bar{1}, v \ass \bar{5}\}\},\\[1ex]
\Theta_{H_c}(\bar{1},\bar{5}) = \Theta_{H_c}(0,\bar{3}) = \{\{v \ass \bar{3}\}\} \cup \Theta_{H_c}(\bar{2},\bar{1}) = \{\{v \ass \bar{3}\}\},\\[1ex]
\Theta_{H_d}(\bar{1},\bar{5}) = \Theta_{H_d}(0,\bar{3}) = \Theta_{H_d}(\bar{2},\bar{1}) = \emptyset.
\end{array}
\]
\end{example}

\begin{example}
A more complex example. For the parameter assignment $\sigma(x) = \bar{6}$ and $\sigma(y) = \bar{8}$, we need four instantiations of the formula $H_b$:
\[
\begin{array}{l}
\Theta_{H_b}(\bar{6},\bar{8}) = \{\{u \ass \bar{6}, v \ass \bar{8}\}\} \cup \Theta_{H_b}(\bar{5},\bar{6}) =\\
\{\{u \ass \bar{6}, v \ass \bar{8}\}, \{u \ass \bar{5}, v \ass \bar{6}\}\} \cup \Theta_{H_b}(\bar{4},\bar{4}) =\\
\{\{u \ass \bar{6}, v \ass \bar{8}\}, \{u \ass \bar{5}, v \ass \bar{6}\}, \{u \ass \bar{4}, v \ass \bar{4}\}\} \cup \Theta_{H_b}(\bar{3},\bar{2}),\\[1ex]
\Theta_{H_b}(\bar{3},\bar{2}) = \{\{u \ass \bar{3}, v \ass \bar{2}\}\} \cup \Theta_{H_b}(\bar{2},0) = \{\{u \ass \bar{3}, v \ass \bar{2}\}\} \cup \Theta_{H_b}(\bar{1},0) =\\
\{\{u \ass \bar{3}, v \ass \bar{2}\}\}. \\[1ex]

\text{Therefore, }\\
\Theta_{H_b}(\bar{6},\bar{8}) = \{\{u \ass \bar{6}, v \ass \bar{8}\}, \{u \ass \bar{5}, v \ass \bar{6}\}, \{u \ass \bar{4}, v \ass \bar{4}\}, \{u \ass \bar{3}, v \ass \bar{2}\}\}.
\end{array}
\]
\end{example}

In the special case of the 2-Hydra example, the Herbrand system can be read off directly from the point transition system, without first reconstructing the full proof schema. This is possible because the only variables requiring substitution instances are parameters, and the relevant quantifier instances are already reflected in the transitions of the point transition system. This simplification is specific to the present example and does not extend to the general case, as illustrated in Example~\ref{ex.schema-pts}.

\section{Conclusion}\label{sec.conclusion}

In this paper, we have investigated the relationship between proof schemata, cyclic proofs, and the extraction of Herbrand information in the presence of induction. 
Previous approaches have been extended by incorporating proof schemata based on point transition systems, thereby increasing the expressive power of recursive proof representations. 
Moreover, we showed that for such schemata with at most quantifier-free cuts Herbrand systems can be effectively constructed based on traditional proof-theoretic methods (Definition \ref{def.Herbrandschema} and Theorem \ref{th.HerbrandsystemSoundComplete}).

Furthermore, we established a connection between cyclic proofs and proof schemata by demonstrating that a relevant subclass of cyclic proofs can be algorithmically translated into proof schemata.
An interesting example that falls into this subclass of cyclic proofs is the 2-Hydra example. 
The translation of this example into the proof schema formalism (Theorem \ref{thm:2hydraIsProvable}) illustrates the strength of our approach, showing that proof schemata based on point transition systems can capture inductive arguments beyond $\mathsf{LKID}$.

The present work opens a number of directions for further investigation. In particular, we plan to investigate the correspondence between more general classes of cyclic proofs and proof schemata by defining translations between these two systems that may extend the underlying theory. An important aspect in this context is that many cyclic proofs are based on Henkin semantics, whereas proof schemata are typically interpreted under standard semantics, and aligning these perspectives poses an interesting challenge. 
Another direction for future research is the analysis and comparison of our approach to Herbrand schema computation with the Herbrand structures extracted from cyclic proofs via higher-order recursion schemes \cite{afshari2025herbrand}. At present, the relationship between these Herbrand structures and the Herbrand systems obtained from proof schemata in our framework is not well understood. In particular, it remains open whether the Herbrand structures arising from cyclic proofs can always be recovered by translating the cyclic proof into the proof schema formalism and subsequently computing its Herbrand system, and vice versa. Establishing such a correspondence, or identifying limitations of the translation, is left for future investigation.
We also plan to extend our methods to richer forms of induction, which would further clarify the scope and limitations of the proof schema formalism. Ultimately, we expect that these developments will contribute to a more comprehensive understanding of inductive reasoning and the extraction of computational content from complex proofs.

\bibliographystyle{plain}
\bibliography{references}

\end{document}